\documentclass[11pt, letterpaper]{article}
\usepackage{booktabs}
\usepackage[numbers]{natbib} 

\usepackage{style}
\usepackage{authblk}
\usepackage{MnSymbol}

\newcommand{\vecket}[1]{|{#1} \rrangle}
\newcommand{\vecbra}[1]{\llangle {#1}|}
\newcommand{\smallnorm}[1]{\Vert #1\Vert }
\newcommand{\smallabs}[1]{ |#1|}
\newcommand{\eps}{\varepsilon}
\renewcommand{\L}{\mathcal L}

\begin{document}

\title{Near-optimal quantum simulation of lattice Lindbladian dynamics}

\author[1]{Rahul Trivedi}
\author[2]{Xiehang Yu}
\author[3]{Daniel Malz}
\makeatletter
\@namedef{@sep4}{,\protect\\[0.25em]}
\makeatother
\affil[1]{Max Planck Institute of Quantum Optics, Garching bei München, Germany}
\affil[2]{Department of Computing and Mathematical Sciences, California Institute of Technology, Pasadena, USA}
\affil[3]{Department of Physics, University of Basel, Switzerland}

\date{\today}

\maketitle
\begin{abstract}
We consider the problem of simulating dynamics under geometrically local Lindbladians acting on a lattice of qudits.
We present an algorithm to approximate the time evolution under these Lindbladians for evolution time $t$ and on $N$ qudits to diamond norm error $\varepsilon$ using geometrically local circuits of depth $O(t \ \polylog(Nt/\varepsilon))$.
Our results thus achieve a performance comparable to that of the Haah-Hastings-Kothari-Low (HHKL) algorithm for lattice Hamiltonian simulation.
\end{abstract}

\section{Introduction}
One of the most promising applications of quantum computers is the simulation of quantum many-body systems~\cite{Feynman1982}.
Consequently, there has been a long effort to discover and optimize quantum algorithms for simulating many-body dynamics.
For closed systems described by a Hamiltonian, existing algorithms range from simple algorithms such as Trotterization~\cite{Lloyd1996} to more sophisticated tools such as quantum signal processing~\cite{LowChuang2017} as well as higher-order Trotter formulae~\cite{Suzuki1991, Childs2012}. Recent years have also seen experimental efforts to implement these algorithms, and their use is likely to grow as large-scale, error-corrected quantum computers become a reality.

An especially physically relevant class of many-body models which arise in both high-energy and low-energy physics consists of systems with geometrically local interactions on a lattice. Their appearance across a wide range of contexts has motivated the quest to develope algorithms that exploit their geometrical locality to perform better than general-purpose algorithms. In the closed-system setting, where the many-body model is a geometrically local Hamiltonian, several algorithms have been developed achieving essentially optimal scalings in gate-count and circuit depth. These include algorithms based on higher-order product formulae \cite{Childs2019} as well as those exploiting Lieb-Robinson bounds to reduce the full Hamiltonian evolution to evolution on local patches~\cite{Haah2021}.

Realistic physical models, however, typically include some interaction with their environment and are therefore ``open''.
The simplest and most widely used model of such systems describes their dynamics by a Lindblad master equation~\cite{GoriniKossakowskiSudarshan1976,Lindblad1976}, 
a generalization of Schr\"odinger's equation.
Mathematically, it is the most general Markovian dynamical model that generates a single-parameter semigroup of channels \cite{Davies1979,Lindblad1976,WolfCirac2008}, just as Hamiltonian dynamics generates a single-parameter group of unitaries.

Several notable efforts have been made to develop algorithms to simulate Lindblad master equations on quantum computers~\cite{CleveWang2017,LiWang2023}.
However, generalizing existing techniques developed for Hamiltonian dynamics to Lindbladian dynamics faces a fundamental obstacle: The irreversibility of Lindbladian dynamics.
In particular, state-of-the-art algorithms for simulating lattice Hamiltonians achieving near-optimal scaling require carefully engineering forward and backward time Hamiltonian evolution.
Since Lindbladian dynamics is irreversible, existing constructions do not directly translate.
Nevertheless, using linear combinations of unitaries and higher-order series expansions, quantum simulation algorithms for Lindbladian dynamics achieve gate counts scaling nearly linearly with the evolution time and polylogarithmically with the inverse target precision~\cite{CleveWang2017,LiWang2023}.
However, for lattice Lindbladians, these algorithms scale unfavorably with respect to the system size due to the overheads involved in implementing the input-model required in these algorithms \cite{ChildsMaslovNamRossSu2018}. More success has been achieved in achieving near-optimal scalings for certain open-system models that are non-Markovian i.e., not described by Lindbladians, where some of the mathematical difficulties in analyzing perfectly Markovian models can be avoided \cite{YuLiCiracTrivedi2025,LiWangNonMarkovian2023}. While it might be expected that a non-Markovian model can be used to well-approximate a Markovian model, these approaches tend to scale unfavorably close to the Markovian limit. Therefore, a central question has remained open:

\begin{quote}
\emph{Do there exist quantum algorithms to simulate lattice Lindbladians with a gate cost quasilinear in $Nt$ and polylogarithmic in $1/\eps$?}
\end{quote}

In this paper, we answer this question affirmatively by providing a quantum algorithm that achieves essentially optimal scaling with respect to time and system size for lattice Lindbladians (\emph{cf.} \cref{th:algorithm}).

\section{Results}
\subsection{Statement of main result}
For clarity, we present our results for the 1D nearest-neighbour setting, but our analysis carries over directly to higher-dimensional lattice models. We consider $N$ qudits with local Hilbert space $\mathbb{C}^d$ arranged on a 1D lattice with their dynamics being described by a Lindbladian
\begin{equation}
\mathcal{L} = \sum_{e}\mathcal{L}_{e},\qquad
	\mathcal{L}_{e} = -i[h_{e}, \cdot] + \sum_{\alpha} \mathcal{D}_{L_{\alpha, e}},\label{eq:basic_lindbladian}
\end{equation}
where the sum runs over all the edges $\{(1, 2), (2, 3), (3, 4) \dots \}$ on the 1D lattice. For $e = (x, x + 1)$, $h_{e}$ (the local Hamiltonian term) and $L_{\alpha, e}$ (the local jump operators) are both operators that act on the sites $x$ and $x + 1$ with $h_{e}$ additionally being Hermitian. Given an operator $L$, $\mathcal{D}_L$ is a super-operator defined by $\mathcal{D}_L(X) =LXL^\dagger - \{L^\dagger L, X\}/2$.
For each edge, we assume that the index $\alpha$ takes at most $d^4 - 1$ distinct values. We further assume the normalization 
\begin{equation}\label{eq:assumption_of_local_norm}
\norm{h_{e}} \leq 1 \quad \text{and} \quad \sum_\alpha \norm{L_{\alpha, e}} \leq 1,
\end{equation}
where $\norm{X}$ denotes the operator norm (i.e., the largest singular value) of the operator $X$.

Given an evolution time $t$, we seek to design a channel $\mathcal{E}$ that can be decomposed into elementary single and two-site gates which approximates $\exp({\mathcal{L}t})$ to $\varepsilon$ error i.e.,
\begin{align}
\norm{\mathcal{E} - \exp(\mathcal{L}t)}_\diamond \leq \varepsilon,\label{eq:channel_target}
\end{align}
where $\norm{\mathcal{S}}_\diamond$ is the diamond norm of a super-operator $\mathcal{S}$.
The main result of this paper is a near-optimal algorithm to implement $\mathcal{E}$.
\begin{theorem}
    A channel $\mathcal{E}$ satisfying \cref{eq:channel_target} can be implemented with $O(Nt \polylog(Nt/\varepsilon))$ geometrically local gates and with $O(\polylog(Nt/\varepsilon))$ ancillas per site.
    \label{th:algorithm}
\end{theorem}
This result thus generalizes the HHKL algorithm for lattice Hamiltonians \cite{Haah2021} to the case of Lindbladians, and achieves the same asymptotic gate count.
Furthermore, since Ref.~\cite{Haah2021} established a lower bound $\Omega(Nt)$ for simulating piecewise-constant lattice Hamiltonians on $N$ qubits with evolution time $t$, and since Hamiltonian dynamics is a special case of Lindblad dynamics, our result achieves optimal scaling in $N$ and $t$ up to polylogarithmic factors.

\subsection{Algorithm sketch and proof strategy}

\paragraph{The HHKL algorithm.} A key step in the HHKL algorithm is to use geometric locality to break the time evolution unitary $U$ acting on the whole system into smaller pieces acting on subsystems.
The time evolution on subsystems can be run in parallel and ``merged'' using local operations at their interfaces.
This removes almost all dependence on system size and thus enables the algorithm to achieve essentially optimal scaling.

Mathematically, the effect of locality is captured by Lieb-Robinson bounds~\cite{LiebRobinson1972}. Given a time evolution unitary derived from a local lattice Hamiltonian, $U=\exp(-iHt)$, a Lieb-Robinson bound captures that for a local operator $X$,  $U^\dagger XU$ will have support localized to  within a radius of size $vt$ (approximately, with exponential tails), where $v$ is the system-size-independent Lieb-Robinson velocity capturing the speed of information flow.
Splitting the chain into three disjoint regions $A, B, C$, one can use Lieb-Robinson bounds to approximate~\cite{Haah2021}
\begin{equation}
	e^{-itH_{A\cup B\cup C}} \approx e^{-itH_{A\cup B}}e^{itH_{B}}e^{-it H_{B\cup C}}
	\label{eq:HHKL-split}
\end{equation}
as long as the size of the screening region $|B|$ sufficiently exceeds $vt$.
Dividing the total time $t$ into $T$ constant-sized intervals, we can simulate the time evolution during each interval by repeatedly applying \cref{eq:HHKL-split} to divide the global time evolution unitary into unitaries on evenly sized chunks of size $l=C_\eps\log(Nt/\eps)$ for some constant $C_\eps$ that depends on the desired approximation error $\eps$ but not on system size.
Thus, one can approximate the time evolution to $\eps$ error by only simulating small systems of size $O(\log(Nt/\eps))$.
Even if simulating the time evolution of a chunk requires a circuit depth polynomial in its size, since the latter is only logarithmic, the whole algorithm ends up having polylogarithmic depth for one time step.

\paragraph{Algorithm for Lindbladian systems.}

Lieb-Robinson bounds can also be derived for dissipative systems governed by a Lindbladian~\cite{Poulin2010}, 
which motivates exploring a similar construction for Lindbladian lattice systems, replacing $H$ by $\L$ in the splitting \cref{eq:HHKL-split}.
This runs into the problem that negative (backward) time evolution under a Lindbladian does not correspond to a physical process and therefore cannot be implemented on a quantum computer directly.
This has prevented a straightforward generalization of HHKL to Lindbladians.

We overcome this problem in three steps.
First, instead of the Lindblad dynamics, we consider a \emph{dilation}~\cite{HudsonParthasarathy1984}: we express the Lindblad dynamics of the system as the unitary dynamics of a system interacting with an environment, in such a way that the dynamics of the reduced system is the same~(\cref{sec:dilation}).
The use of dilations for quantum simulation has been explored before in Ref.~\cite{Kliesch2011,CleveWang2017}.

Since the combined system and environment undergo unitary dynamics, their evolution can be run backwards. Thus, in a second step (\cref{sec:lieb-robinson-bounds,sec:splitting}), we use a generalized Lieb-Robinson bound for the system-environment dynamics from Ref.~\cite{TrivediRudner2024}, which enables a splitting analogous to \cref{eq:HHKL-split} (\emph{cf.} \cref{fig:splitting_unitary}), which can be used to write the global time evolution for one constant-sized time step in terms of three layers of quasi-local forward and backward evolution steps (\emph{cf.} \cref{fig:chopping-unitaries}).

This by itself does not give us the desired algorithm, since we still need to compile the dilation Hamiltonian into a local circuit with ancillas.
Formally, the dilation Hamiltonian contains infinitely many infinite-dimensional bosonic modes.
Standard discretization reduces this to finitely many modes, but two problems remain: first, each mode is still infinite dimensional, and second, the number of modes required to control the approximation error scales linearly in the inverse error.
Together with the need to compile into local circuits, this would recover a scaling polynomial in $N$ similar to Ref.~\cite{CleveWang2017}.

To overcome this, in a third step (\cref{sec:collision,sec:compressing}), we use that the number of excitations ``emitted'' into the environment due to any one of the jump operators in \cref{eq:basic_lindbladian} during a time $t$ is bounded (with exponential tails) by a constant times $t$.
This allows us to express the environment state in a ``first-quantized'' basis, where the basis states are ordered lists of the positions of the emitted particles.
For up to $n_c$ excitations in $T$ qubits, there are $\sum_{k=0}^{n_c}\binom{T}{k}$ basis states, which can be stored in $O(n_c\log T)$ qubits.
Thus, with a polylogarithmic excitation cutoff $n_c$, the number of ancilla qubits is polylogarithmic even when $T$ scales polynomially.
Although this removes the disadvantageous scaling in the number of ancillas, it is not immediately clear whether the evolution in this compressed space can be simulated efficiently, as the gate complexity might be exponential in the number of qubits (as it would be for a general operation).
However, as we show, there exists an efficient block-encoding of the system-environment dynamics in the truncated basis that therefore achieves the desired optimal scaling.

\section{Technical Discussion}
\subsection{Dilation}\label{sec:dilation}
Our analysis will heavily exploit an exact dilation of the Lindblad master equation as provided by the theory of Quantum Stochastic Calculus \cite{HudsonParthasarathy1984}. For readability, we will use the physics second quantized notation for describing this dilation which is formally equivalent to quantum stochastic calculus formalism. Consider a Lindbladian on a system with Hilbert space $\mathcal{H}_S$ described by a Hamiltonian $H_S$ and jump operators $L_1, L_2 \dots L_m$.
For each jump operator, we introduce an ancillary Hilbert space that is a bosonic Fock space over square-integrable functions $\text{Fock}(L^2(\mathbb{R}))$.
We refer to the ancilliary Hilbert spaces together as ``environment", which correspondingly has the Hilbert space $\mathcal{H}_E = (\text{Fock}(L^2(\mathbb{R})))^{\otimes m}$. 
Then, on the full system-environment Hilbert space $\mathcal{H} = \mathcal{H}_S \otimes \mathcal{H}_E$ we define the Hamiltonian
\begin{align}
    H(t) = H_S + \sum^{m}_{\alpha = 1}(L_\alpha a^\dagger_{\alpha}(t) + \text{h.c.}),
\end{align}
where $a_\alpha(t)$ is an annihilation operator acting on the $\alpha^{\text{th}}$ ancillary Fock space.
The annihilation operators obey the commutation relations $[a_\alpha(t), a_\beta(s)] =0$ and $[a_\alpha(t), a^\dagger_{\beta}(s)] = \delta(t - s) \delta_{\alpha, \beta}$.
We take the $\delta$-function to be symmetric in the sense that for any continuous function $f:\mathbb{R}\to \mathbb{C}$ and an interval $[a, b]$, 
\begin{equation}
\int^b_{a}\delta(t - s) f(s) ds = \begin{cases}
    f(t) & \text{ if } t \in (a, b), \\
    f(t)/2 & \text{ if } t \in \{a, b\}, \\
    0 & \text{ otherwise,}
\end{cases}
\end{equation}
\begin{figure}
    \centering
    \includegraphics[width=0.8\linewidth]{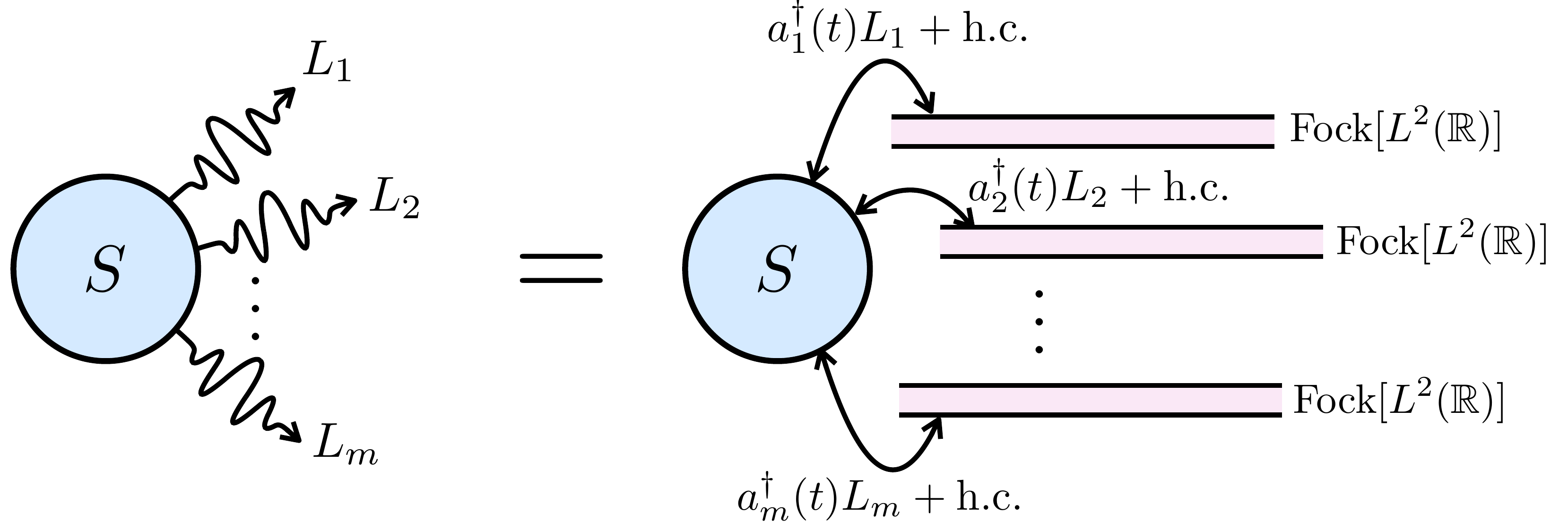}
    \caption{Dilation of the Lindblad evolution. Each jump operator $L_\alpha$ couples the system $S$ to an independent bosonic field through $a_\alpha^\dagger(t)L_\alpha+\mathrm{h.c.}$.
    With all fields initially in the vacuum, tracing out the environment recovers the system's Lindblad dynamics.}
    \label{fig:generic-dilation}
\end{figure}
i.e., when the peak of the $\delta$-function is at the edge of the interval being integrated over, the integral will acquire an overall factor of $1/2$ compared to when the peak of the $\delta$-function is in the interior of the interval. With this choice of the $\delta-$function, the infinitesimal operator $a_\alpha(t)dt$ can be viewed as equivalent to $dB_\alpha(t)$ where $B_\alpha(t)$ is the quantum Wiener process and the Schr\"odinger equation generated by $H(t)$ is formally equivalent to a quantum It\^o equation. Furthermore, it is relatively easy to show that, starting from the ancillary Fock spaces in the vacuum state, evolving under the dynamics generated by $H(t)$ defined above and tracing out the ancillary Fock spaces indeed yields system dynamics described by a Lindblad master equation.
\begin{lemma}[Vacuum dilation~\cite{HudsonParthasarathy1984}]
\label{lem:sys-env-dilation}
    Suppose $U(t, s) = \mathcal{T}\exp(-i\int^t_{s}H(\tau)d\tau)$ is the dynamics generated by $H(t)$ and $\ket{\textnormal{vac}}$ is the joint vacuum state on the ancillary Fock spaces. Then, for $t>0$,
    \begin{equation}
    \Tr _E[U(t, 0)(\rho_S\otimes \proj{\mathrm{vac}})U(0, t)] = e^{\mathcal{L}t}(\rho_S),
    \end{equation}
    where $\Tr _E$ traces out the ancillary bosonic Fock spaces.
\end{lemma}
It is also useful, to understand the key ideas in using this model, to visualize the ancillary Fock spaces as a 1D bosonic quantum field which is itself static (i.e., not evolving with time) but the system interacts with it at position $x = t$ at time $t$. This picture also makes it intuitively clear why the system dynamics is Markovian, as stated by \cref{lem:sys-env-dilation}: The portions of this field that the system interacts with at any given time never interact with it again in the future.

We will make extensive use of vectorized notation, which maps operators on a Hilbert space $\mathcal{H}$ to vectors on the Hilbert space $\mathcal{H}\otimes \mathcal{H}$.
Given an operator $\phi = \sum_{i, j}\phi_{i,j} \ket{i}\!\bra{j}$, we will denote its vectorized state by $\vecket{\phi} = \sum_{i, j} \phi_{i, j} \ket{i, j}$.
Superoperators on the Hilbert space $\mathcal{H}$ will then map to operators on $\mathcal{H}\otimes \mathcal{H}$. In particular, given an operator $X$, we will denote by $X_{\text{L}}$ the superoperator that left-multiplies by $X$, i.e., $X_{\text{L}}(\phi) = X\phi$ and, when vectorized, $X_{\text{L}}= X \otimes I$. Similarly, we will denote by $X_{\text{R}}$ the superoperator which right-multiplies by $X$ i.e., $\text{X}_{\text{R}}(\phi) = \phi X $ and, when vectorized, $\text{X}_{\text{R}}= I \otimes X^{\text{T}}$. Given a unitary $U$, we will denote by $\mathcal{U}$ the super-operator $\mathcal{U}(\phi) = U \phi U^\dagger$, or as vectorized operator $\mathcal{U} = U_{\text{L}}(U^\dagger)_{\text{R}}=  U \otimes U^{*}$. Given an operator $A$, we will denote by $\mathcal{C}_A$ the super-operator $\mathcal{C}_A(\phi) = [A, \phi]$, or as vectorized operator $\mathcal{C}_A = A_{\text{L}}- A_{\text{R}}$. To simplify notation, whenever the operator $A$ is a Hamiltonian or a Hamiltonian term, then we will often use $\mathcal{A} = \mathcal{C}_A$.
We write $\vecket{\mathrm{vac}}:=\vecket{\mathrm{vac,vac}}$
for the vectorized vacuum projector.

For $a, b \in \mathbb{R}$, we will denote by $[a, b]_{\textnormal{ord}}= [\min(a, b), \max(a, b)]$. Note that if $a \leq b$, then $[a, b]_{\textnormal{ord}}= [a, b]$ but if $a > b$, then $[a, b] = \emptyset$ and thus $[a, b]_{\textnormal{ord}}\neq [a, b]$.
Similarly, $(a,b)_{\textnormal{ord}}:=(\min(a,b),\max(a,b))$.
While dealing with the set of integers, we will use the notation $[m:n] = \{m, m + 1, m + 2 \dots n\}, (m:n) = \{m + 1, m +2 \dots n - 1\}, (m:n] = \{m + 1, m + 2 \dots n\}, [m:n) = \{m, m + 1 \dots n - 1\}$. Given a set $A$, $\mathcal{I}_A$ will be the indicator function on $A$ i.e., $\mathcal{I}_A(x) = 1$ if $x \in A$ and $0$ otherwise. Given two sets $A$ and $B$, we will set $\mathcal{I}_A(B) =  1$ if $B \subseteq A$ and $0$ otherwise.

\subsection{Applying Lieb-Robinson bounds}
\label{sec:lieb-robinson-bounds}
We start with an exact system-environment dilation. Following the previous subsection, the dilated Hamiltonian for the 1D nearest-neighbour Lindbladian in \cref{eq:basic_lindbladian},
\begin{align}\label{eq:hamiltonian_dilation_nn}
H(t) = \sum_{e} H_{e}(t),
\quad\text{where} \quad
H_{e}(t) = h_{e} + \sum_\alpha (L^{\dagger}_{\alpha, e}a_{\alpha, e}(t) + \text{h.c.})
\end{align}
and $[a_{\alpha, e}(t), a^\dagger_{\alpha',e'}(t')] =  \delta(t -t') \delta_{e, e'} \delta_{\alpha, \alpha'}$. Furthermore, the initial state of the environment is the vacuum state $\rho_E(0) = \proj{\text{vac}}$. We will denote by $U(t, s)$ the full system-environment unitary generated by $H(t)$:
\begin{equation}
U(t, s) = \mathcal{T}\exp\bigg(-i\int^t_{s}H(\tau) d\tau\bigg).
\end{equation}
\begin{figure}
    \centering
    \includegraphics[width=0.8\linewidth]{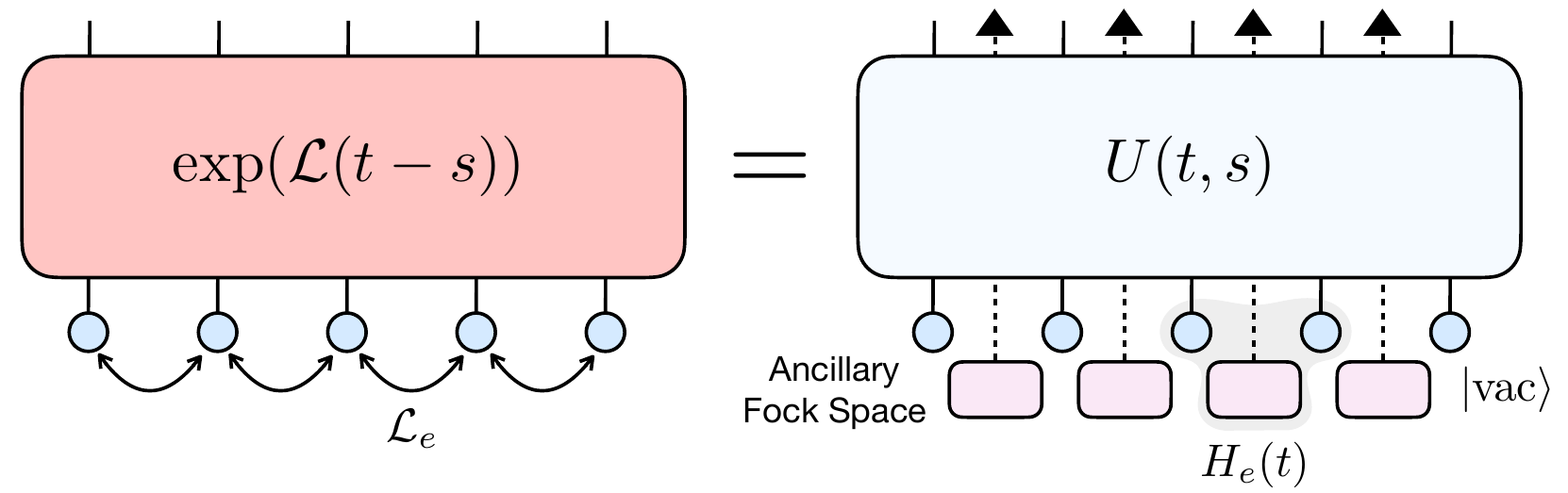}
    \caption{Dilation of nearest-neighbour Lindblad dynamics on a chain.
    Each edge term $\mathcal{L}_e$ is represented by a system--environment interaction $H_e(t)$ involving the two adjacent sites (blue disks) and ancillary bosonic fields (pink rectangles).
    The fields start in the vacuum and are traced out after the joint unitary evolution $U(t,s)$, with $t\geq s$.}
    \label{fig:dilation}
\end{figure}

Given a set of edges $\mathcal{A} \subseteq \{ (1, 2), (2, 3) \dots \}$ or a set of sites $X \subseteq \{1, 2 \dots N\}$, we define
\begin{equation}
H_{\mathcal{A}}(t)= \sum_{e \in \mathcal{A}} H_e(t) \quad \text{and}\quad H_X(t) = \sum_{e \subseteq X} H_{e}(t).
\end{equation}
With this notation, $H(t) = H_{\{(1, 2), (2, 3) \dots \}}(t) = H_{[1: N]}(t)$. We will also denote by $U_{\mathcal{A}}(t, s)$ and $U_X(t, s)$ the unitary generated by $H_{\mathcal{A}}(\tau)$ and $H_X(\tau)$ respectively during the evolution from $s\to t$ i.e., 
\begin{equation}
	U_{\mathcal{A}}(t, s) = \mathcal{T}\exp\bigg(-i\int^t_{s}H_{\mathcal{A}}(\tau) d\tau\bigg) \quad \text{and}\quad U_X(t, s) = \mathcal{T}\exp\bigg(-i\int^t_{s}H_X(\tau) d\tau\bigg).
    \label{eq:subset-notation}
\end{equation}

Let us now consider the full system-environment unitary within a time-period $t_0$: $U(t + t_0, t)$. We will approximately split it into a sequence of ``local" unitaries following the strategy of Haah-Hastings-Kothari-Low (HHKL)~\cite{Haah2021}. 
The following lemma provides an expression for the error incurred in this splitting.
The proof follows lemma 6 from Ref.~\cite{Haah2021}, but instead of applying norm bounds, we keep track of the exact expression for the remainder.

In Ref.~\cite{Haah2021}, since the authors were only concerned with a bounded lattice Hamiltonian, an application of a norm-bound together with the Lieb-Robinson bounds was sufficient to control the remainder in \cref{lemma:splitting_remainder}. However, the dilated Hamiltonian is unbounded, which makes usual Lieb-Robinson bounds divergent~\cite{LiebRobinson1972}.
Therefore, a direct application of the procedure in Ref.~\cite{Haah2021} fails. To get around this issue, we need to use the fact that the initial state of the environment is vacuum, under evolution relative to the system-environment Hamiltonian, it does not accumulate a large number of excitations and, within an effective low-excitation space, the terms in the dilated Hamiltonian are effectively bounded. This problem, in a more general class of possibly non-Markovian system-environment models, has been considered in Ref.~\cite{TrivediRudner2024}.
In particular, Ref.~\cite{TrivediRudner2024} identified operator spaces in which the annihilation operators $a_{\alpha, e}(t)$ behave effectively as bounded operators, which enabled the derivation of a Lieb-Robinson bound with a finite Lieb-Robinson velocity.
Since Ref.~\cite{TrivediRudner2024} considered a more general setting than we are interested in, below we specialize and state the definition of relevant operator spaces as well as the Lieb-Robinson bounds \citep[Lemma~16]{TrivediRudner2024} that are relevant for our analysis.

We begin by defining the operator spaces identified in Ref.~\cite{TrivediRudner2024}.
\begin{definition}[Operator space $\mathcal{S}_K$]\label{def:op_space} For $K \in \mathbb{Z}_{+}$, a system-environment operator $\phi \in \mathcal{S}_K$ if it can be expressed as
\begin{equation}
\vecket{\phi} = \Omega_1\mathcal{U}_{\mathcal{A}_1}(t_1, s_1) \Omega_2 \mathcal{U}_{\mathcal{A}_{2}}(t_2, s_2) \dots \Omega_M \mathcal{U}_{\mathcal{A}_{M}}(t_M, s_M)\vecket{\sigma_S, \textnormal{vac}},
\end{equation}
where $\mathcal{A}_i$ are sets of edges, $\sigma_S$ is a system-operator satisfying $\norm{\sigma_S}_1 \leq 1$, $\Omega_i$ are system super-operators satisfying $\norm{\Omega_i}_\diamond \leq 1$, and the interval multiset $\{(t_i, s_i)_{\textnormal{ord}}: i\in \{1, 2 \dots M\}\}$ has at most $K$ overlapping intervals, i.e.
\begin{equation}
\forall t \in \mathbb{R}: \abs{\{i: t \in (t_i, s_i)_{\textnormal{ord}}\}} \leq K.
\end{equation}
\end{definition}
The motivation to introduce $\mathcal S_K$ is that it contains the relevant system-environment states, but additionally allows us to bound the number of excitations in the environment.
Recall the interpretation of the dilation in terms of 1D fields presented in \cref{sec:dilation}: Each ancillary bosonic Fock space can be thought of as a static 1D field, with the system interacting with the field at position $x = t$ at time $t$. Now, in any element $\phi \in \mathcal{S}_K$, almost all positions on the 1D field experience interactions with the system at most $K$ times, because at most $K$ of the time intervals in its definition overlap at any point. Since the environment (i.e., the 1D fields) initially start from the vacuum state and can only get excited by interacting with the system, $K$ becomes a measure of how ``displaced'' the field is at any given location.
Thus, for states that lie within $\mathcal S_K$, the action of otherwise unbounded operators such as $a_{\alpha, e}(\tau), a^\dagger_{\alpha, e}(\tau)$ can be controlled, which we make precise in  \cref{lemma:op_space_prop} below.

\begin{lemma}[Boundedness of $a_{\alpha, e}(t)$ in $\mathcal S_K$]\label{lemma:op_space_prop}
    Suppose $\phi \in \mathcal{S}_K$ then for any $\alpha,e,\tau$: $a_{\alpha,e}(\tau) \phi = \sum_{i} c_i \phi_i$ with $\phi_i \in \mathcal{S}_{K}$ and $\sum_i \abs{c_i} \leq K$. The same holds for $\phi a^\dagger_{\alpha,e}(\tau)$.
\end{lemma}
A direct consequence of \cref{lemma:op_space_prop} is that for states $\phi \in \mathcal{S}_K$, $\norm{a_{\alpha, e}(t)\phi}, \norm{\phi a^\dagger_{\alpha, e}(t)} \leq \sum_i\abs{c_i} \leq K$. Therefore, within these state-spaces, the otherwise unbounded operators $a_{\alpha, e}(t), a^\dagger_{\alpha, e}(t)$ become well-behaved and effectively bounded by the parameter $K$.

To prove \cref{lemma:op_space_prop}, we first prove another simple but useful result commonly used in the input-output formalism in open-system literature~\cite{GardinerCollett1985}.
\begin{lemma}\label{lemma:input_output}
For times $t, s, \tau \in \mathbb{R}$ with $t\neq s$ and $\alpha, e$ labelling the jump operator $\alpha$ on edge $e$, 
\begin{equation}
[a_{\alpha, e}(\tau), U_{\mathcal{A}}(t, s)] = -i \mathcal{I}_{\mathcal{A}}(e)\Theta_{t, s}(\tau) U_{\mathcal{A}}(t, \tau) L_{\alpha, e} U_{\mathcal{A}}(\tau, s),
\end{equation}
where $\mathcal{I}_{\mathcal{A}}(e) = 1$ if the edge $e\in\mathcal{A}$ and $0$ otherwise, and
    \begin{equation}
    \Theta_{t, s}(\tau) = \textnormal{sgn}(t - s)\begin{cases}
        1/2 & \text{ if } \tau \in \{s, t\}, \\
        1 & \text{ if } \tau \in (\min(t, s), \max(t,s )), \\
        0 & \text{ otherwise}.
    \end{cases}
    \end{equation}
\end{lemma}
\begin{proof}
    We note that
    \begin{equation}
    \frac{d}{dt} U_{\mathcal{A}}(s, t) a_{\alpha, e}(\tau) U_{\mathcal{A}}(t, s) = -i\mathcal{I}_{\mathcal{A}}(e)\delta(t - \tau) U_{\mathcal{A}}(s, t)L_{\alpha, e} U_{\mathcal{A}}(t, s).
    \end{equation}
    We integrate this from $s \to t$ to obtain
    \begin{equation}
    U_{\mathcal{A}}(s, t)a_{\alpha, e}(\tau) U_{\mathcal{A}}(t, s) -a_{\alpha, e}(\tau) = -i\mathcal{I}_{\mathcal{A}}(e) U_{\mathcal{A}}(s, \tau)L_{\alpha, e}U_{\mathcal{A}}(\tau, s)\int^t_{s}\delta(t' - \tau) dt'.
    \end{equation}
    Noting that $\int^t_{s}\delta(t' - \tau) dt' = \Theta_{t, s}(\tau)$ and left-multiplying this equation on both sides by $U_{\mathcal{A}}(t, s)$, we obtain the lemma statement.
\end{proof}

\begin{proof}[Proof of \cref{lemma:op_space_prop}]
    We show only the argument for $a_{\alpha,e}(\tau) \phi$, since the argument for $\phi a^{\dagger}_{\alpha,e}(\tau)$ works in the same way.
    Given
    \begin{equation}
    \vecket{\phi} = \Omega_1 \mathcal{U}_{\mathcal{A}_1}(t_1, s_1) \Omega_2 \mathcal{U}_{\mathcal{A}_2}(t_2,s_2) \dots \Omega_M \mathcal{U}_{\mathcal{A}_M}(t_M, s_M) \vecket{\sigma_S, \text{vac}},
    \end{equation}
    we repeatedly apply \cref{lemma:input_output} together with $a_{\alpha, e}(\tau) \ket{\text{vac}} = 0$
    to obtain
    \begin{align}
        \vecket{a_{\alpha, e}(\tau)\phi} = \sum^{M}_{i = 1}c_i\vecket{\phi_i},
    \end{align}
    where $c_i=\Theta_{t_i,s_i}(\tau)$. Note that since $\phi\in\mathcal S_K$, the multi-set of intervals appearing in $\phi_i$ 
    \begin{equation}
    \{(t_1, s_1)_{\textnormal{ord}}\dots (t_{i - 1},s_{i - 1})_{\textnormal{ord}},  (t_i, \tau)_{\textnormal{ord}}, (\tau, s_i)_{\textnormal{ord}}, (t_{i + 1}, s_{i + 1})_{\textnormal{ord}}\dots (t_M,s_M)_{\textnormal{ord}}\},
    \end{equation}
    has at most $K$ overlapping intervals
    as long as $\tau \in [t_i,s_i]_{\textnormal{ord}}$ (otherwise $c_i = 0$). 
    Together with the fact that $\norm{(L_{\alpha, e})_{\textnormal{L}}}_\diamond \leq \norm{(L_{\alpha, e})} \leq 1$, it follows that $\phi_i \in \mathcal{S}_{K}$.
    
    Finally, we now consider $\sum_i \abs{c_i}$: We note that
    \begin{equation}
    \abs{\Theta_{t_i, s_i}(\tau)} =\lim_{\varepsilon \to 0} \frac{1}{2} \bigg(\mathcal{I}_{(t_i, s_i)_{\text{ord}}}(\tau - \varepsilon) + \mathcal{I}_{(t_i, s_i)_{\text{ord}}}(\tau + \varepsilon) \bigg).
    \end{equation}
    Since at most $K$ of the intervals $\{(t_i, s_i)_{\text{ord}}\}_{i \in [1:M]}$ overlap, we have that
    \begin{equation}
    \sum_i \mathcal{I}_{(t_i, s_i)_{\text{ord}}}(\tau \pm \varepsilon) \leq K,
    \end{equation}
    and therefore
    \begin{equation}
    \sum_i \abs{c_i} = \sum_i \abs{\Theta_{s_i, t_i}(\tau)} \leq K,
    \end{equation}
    which proves the lemma statement.

\end{proof}
We now state the key result from \cite{TrivediRudner2024}: Within the operator space $\mathcal{S}_K$, the system-environment Hamiltonian (\cref{eq:hamiltonian_dilation_nn}) has a Lieb-Robinson bound with a Lieb-Robinson velocity that grows linearly with $K$.

\begin{lemma}[System-environment Lieb-Robinson, specialization of Lemma 16 from Ref.~\cite{TrivediRudner2024}]\label{lemma:lr_bound}
    Suppose $O_X$ is a system-observable supported on $X \subseteq [1:N]$, $\mathcal{K}_Y$ a system super-operator supported on $Y \subseteq [1:N]$ which satisfies $\mathcal{K}_Y(I) = 0$ and $\phi \in \mathcal{S}_K$. Then, for any $\mathcal{B} \subseteq \{(x, x + 1): x \in [1:N ) \}$
    \begin{equation}
    \abs{\Tr [\mathcal{K}_Y(U_\mathcal{B}(s, t)O_X U_\mathcal{B}(t, s))\phi]}
	\leq C_0\norm{O_X}\norm{\mathcal{K}_Y}_{\infty \to \infty} |Y|\exp((v(K)\abs{t - s} - d(X, Y))/\xi_0),
    \end{equation}
    where $v(K) = c_1 + c_2K$ and $C_0, c_1,c_2, \xi_0$ are numerical constants.
\end{lemma}

\subsection{Splitting global time evolution into local pieces}
\label{sec:splitting}
Having introduced the Lieb-Robinson bound, we now turn to the question of approximating the time evolution due to the full system-environment Hamiltonian (\cref{eq:hamiltonian_dilation_nn}) with evolution over the Hamiltonian restricted to smaller sets of sites.
We proceed similarly to Ref.~\cite{Haah2021}, 
with the difference that we additionally need the machinery from the previous section to deal with the unbounded dilation Hamiltonian~\cref{eq:hamiltonian_dilation_nn}.

Our first step is a general lemma that provides an expression for the error incurred when replacing the time evolution on a union of three regions $A \cup B \cup C $ by forward evolution on $B \cup C$, followed by a backward time evolution on $B$ and a final forward time evolution on $A \cup B$ [\emph{cf.} \cref{fig:splitting_unitary}(b) for an illustration].
\begin{figure}
    \centering
    \includegraphics[width=1.0\linewidth]{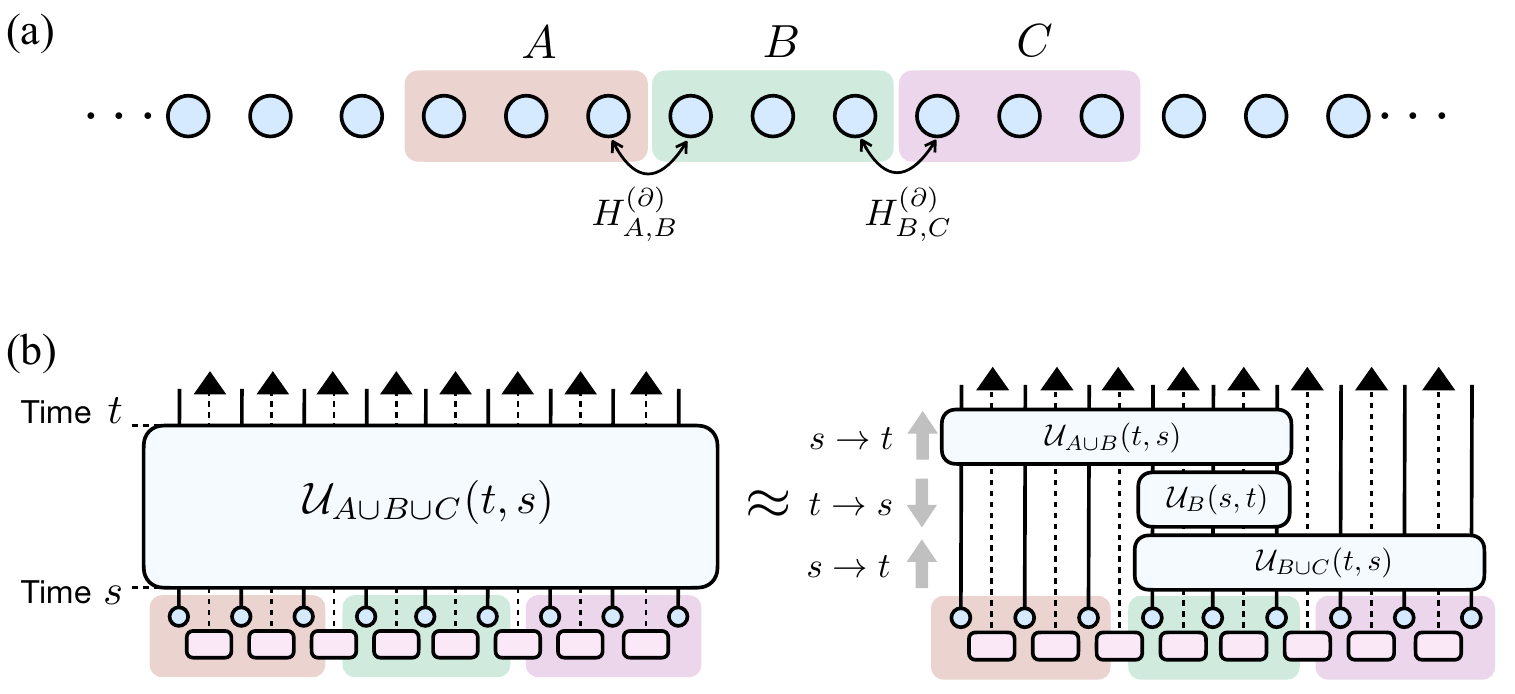}
    \caption{Local splitting of the system--environment evolution.
    (a) Regions $A$, $B$, and $C$, with boundary interactions $H^{(\partial)}_{A,B}$ and $H^{(\partial)}_{B,C}$.
    (b) The full evolution is approximated by evolution on $B\cup C$, backward evolution on $B$, and evolution on $A\cup B$.
    The environment is retained throughout all three operations and traced out only at the end.}
    \label{fig:splitting_unitary}
\end{figure}
\begin{lemma}\label{lemma:splitting_remainder}
Let $0\leq s\leq t$. Suppose $A, B, C \subseteq [1: N]$ are three disjoint regions with $d(A, C) \geq 2$ (an example is schematically depicted in \cref{fig:splitting_unitary}), then
\begin{equation}
    \begin{aligned}
        &\mathcal{U}_{A\cup B\cup C}(t, s) - \mathcal{U}_{A\cup B}(t, s) \mathcal{U}_{B}(s, t) \mathcal{U}_{B \cup C}(t, s) \\
        &\qquad  =\int^t_s\!\int^{\tau}_{s}\mathcal{V}_{A, B, C}(t, s, \tau, \tau')[\mathcal{H}^{(\partial)}_{A,B}(\tau'), \mathcal{U}_B(\tau', \tau)\mathcal{H}^{(\partial)}_{B, C}(\tau)\mathcal{U}_B(\tau, \tau')]\tilde{\mathcal{V}}_{A, B, C}(t, s, \tau, \tau') d\tau d\tau',
    \end{aligned}
\end{equation}
where $\mathcal{H}^{(\partial)}_{A, B}(t) = \mathcal{H}_{A\cup B}(t) - \mathcal{H}_A(t) -\mathcal{H}_B(t)$ (and similarly for $\mathcal{H}^{(\partial)}_{B,C}$),
and
\begin{equation}
\begin{aligned}
    V_{A, B, C}(t, s, \tau, \tau') &= U_{A \cup B}(t, s)U_B(s, t) U_{B \cup C}(t, \tau) U_B(\tau, s) U_{A\cup B}(s, \tau'),\\
    \tilde{V}_{A, B, C}(t, s, \tau, \tau') &= U_{A\cup B}(\tau', \tau) U_{A\cup B \cup C}(\tau, s),
\end{aligned}
\end{equation}
with $\mathcal{V}_{A, B, C}(t, s, \tau, \tau')$ and $\tilde{\mathcal{V}}_{A, B, C}(t, s, \tau, \tau')$ being the channels corresponding to ${V}_{A, B, C}(t, s, \tau, \tau')$ and $\tilde{V}_{A, B, C}(t, s, \tau, \tau')$ respectively.
\end{lemma}
The proof of this lemma only requires straightforward manipulation of Schr\"odinger's equation for the system-environment unitaries and is provided in \cref{sec:splitting-proof}.
The next lemma uses the Lieb-Robinson bound (\cref{lemma:lr_bound}) to provide an upper bound on the norm of the commutators of superoperators appearing in the expression in \cref{lemma:splitting_remainder}.

\begin{lemma}\label{lemma:superop_comm_bound}
    Let $0\leq\tau'\leq\tau$, $\mathcal B$, $\mathcal A_j$ and $\mathcal A'_j$ be sets of lattice edges, and $t_j,s_j,t'_j,s'_j\in\mathbb R$. Suppose $e, f \in \{(x, x + 1): x\in [1:N) \}$ with $e\neq f$, then
    \begin{equation}
		\begin{aligned}
			&\left \Vert\vecbra{I_E}\bigg(\prod^{M}_{j = 1}\mathcal{U}_{\mathcal{A}_j}(t_j, s_j)\bigg)[\mathcal{H}_{e}(\tau'), \mathcal{U}_{\mathcal{B}}(\tau', \tau) \mathcal{H}_{f}(\tau)\mathcal{U}_{\mathcal{B}}(\tau, \tau')] \bigg(\prod^{M'}_{j = 1}\mathcal{U}_{\mathcal{A}'_j}(t_j', s_j')\bigg)\vecket{\textnormal{vac}}\right \Vert_{\diamond}\\
			&\qquad \qquad \qquad \qquad \qquad \leq 2C(2M + M') \exp(\xi_0^{-1}(v(2M + M')\smallabs{\tau - \tau'} - d(e, f))),
		\end{aligned}
    \end{equation}
    where $C(\cdot)$ is defined in \cref{lemma:commutator_H_bound} and $v(\cdot)$ is defined in \cref{lemma:lr_bound}.
\end{lemma}
The proof is given in \cref{sec:commutator-proofs}.
Finally, using \cref{lemma:splitting_remainder,lemma:superop_comm_bound}, we now establish the main result of this subsection: We consider the full system-environment unitary and approximate it with evolution under the system-environment unitary restricted to regions of size at most $2l$ (\cref{fig:chopping-unitaries}). We establish that, if the initial state of the ancillary Fock spaces is $\ket{\text{vac}}$ and these spaces are finally traced out, then the error between the two time evolutions decreases exponentially with the block size $l$.

\begin{figure}
    \centering
    \includegraphics[width=1\linewidth]{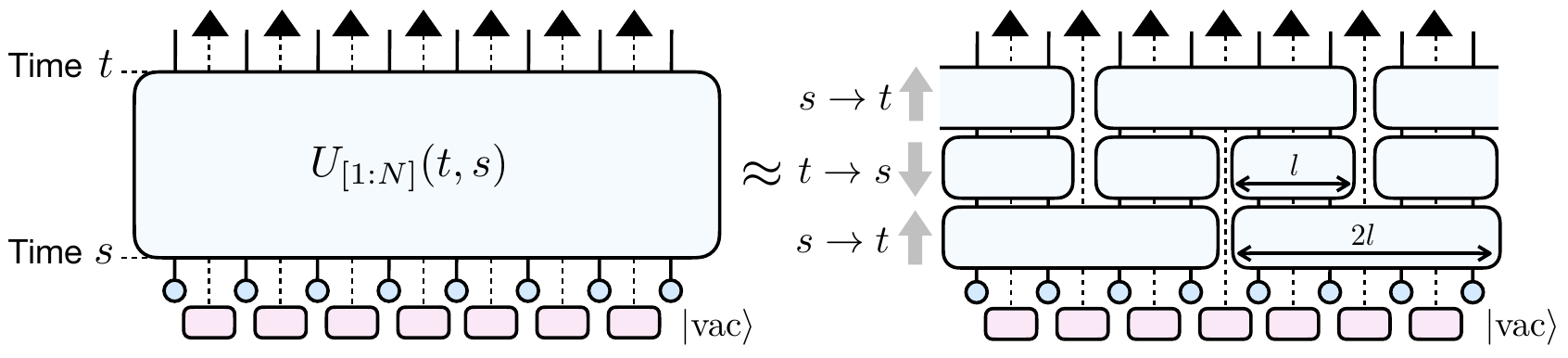}
    \caption{Three-layer decomposition of the dilated evolution.
    Each layer comprises unitaries that factorize over blocks of width at most $2l$. In the final algorithm, $l=O(\log(Nt/\eps))$.}
    \label{fig:chopping-unitaries}
\end{figure}
\begin{lemma}\label{lemma:chopping_unitaries}
    Let $0\leq s\leq t$, $0<\varepsilon<1$, and $l\in\mathbb Z_{>0}$. Consider the system-environment unitary $U^{(l)}(t, s)$ schematically depicted in \cref{fig:chopping-unitaries} and defined by
    \begin{equation}
    U^{(l)}(t, s) = U^{(l)}_{3}(t, s)U^{(l)}_{2}(s, t)U^{(l)}_{1}(t, s),
    \end{equation}
    where $U^{(l)}_{1}(t, s)$, $U^{(l)}_{2}(s, t)$ and $U^{(l)}_{3}(t, s)$ are the three unitary layers shown in \cref{fig:chopping-unitaries} and given by (using the notation introduced in \cref{eq:subset-notation})
    \begin{equation}
		\begin{aligned}
			&U^{(l)}_{1}(t, s) = \prod^{\lfloor N/2l\rfloor}_{j=0}U_{(2jl :2(j + 1)l]}(t,s), \\
			&U^{(l)}_2(s,t) = \prod^{\lfloor N / l \rfloor}_{j=0}U_{((j + 1)l:(j + 2)l]}(s, t) \quad \text{and} \\
			&U^{(l)}_{3}(t, s) = \prod^{\lfloor N/2l \rfloor}_{j=0}U_{((2j + 1)l :(2j + 3)l]}(t,s).
		\end{aligned}
    \end{equation}
Then the channel $\mathcal{E}^{(l)}(t, s) = \Tr _E[U^{(l)}(t, s)((\cdot)\otimes \proj{\mathrm{vac}})(U^{(l)}(t, s))^\dagger]$ satisfies
    \begin{equation}
    \norm{\mathcal{E}^{(l)}(t, s) - \exp(\mathcal{L}(t - s)) }_\diamond \leq \varepsilon.
    \end{equation}
    if $l$ is chosen as
    \begin{equation}
    l = \Theta(\abs{t  - s}) + \Theta(\log(N/\varepsilon)).
    \end{equation}
\end{lemma}
\begin{proof}
    The key idea of this proof is to apply \cref{lemma:splitting_remainder} repeatedly from the left of the lattice to the right. For the sake of notational convenience, for three disjoint regions $A, B, C$, we will denote by $\Delta_{A, B, C}(t, s)$
    \begin{align}
\Delta_{A, B, C}(t, s) = \mathcal{U}_{A\cup B \cup C}(t, s) - \mathcal{U}_{A\cup B}(t, s)\mathcal{U}_B(s, t)\mathcal{U}_{B \cup C}(t, s).\label{eq:delta_def}
\end{align}
    We pick a positive integer $l$.
    For notational convenience, given $k \in \{0, 1, 2\dots\}$, define the subregions
    \begin{equation}
    X_k = (kl:(k + 2)l],\quad
    Y_k = (kl:(k+1)l],\quad
    \text{and}\quad Z_k = (kl:N].
    \end{equation}
    Next, pick $A = Z_{2k + 2}$, $B = Y_{2k + 1}$ and $C = Y_{2k}$ and note that $A \cup B = Z_{2k + 1}$, $B \cup C = X_{2k}$. From \cref{eq:delta_def}, it then follows that
    \begin{align}
        \mathcal{U}_{Z_{2k}}(t, s) = \mathcal{U}_{Z_{2k + 1}}(t, s) \mathcal{U}_{Y_{2k + 1}}(s, t) \mathcal{U}_{X_{2k}}(t, s) + \Delta_{Z_{2k + 2},Y_{2k + 1}, Y_{2k} }(t, s).
    \end{align}
    Then, pick $A = Y_{2k + 1}$, $B = Y_{2k +2 }$ and $C = Z_{2k + 3}$ and note that $A \cup B = X_{2k + 1}$, $B \cup C = Z_{2k+2}$. Again, using \cref{eq:delta_def}, it follows that,
    \begin{align}
        \mathcal{U}_{Z_{2k + 1}}(t, s) = \mathcal{U}_{X_{2k + 1}}(t, s) \mathcal{U}_{Y_{2k + 2}}(s, t) \mathcal{U}_{Z_{2k + 2}}(t, s) + \Delta_{Y_{2k + 1}, Y_{2k +2 }, Z_{2k + 3}}(t,s).
    \end{align}
    Therefore, we obtain 
    \begin{align}
        \mathcal{U}_{Z_{2k}}(t, s) = \big(\mathcal{U}_{X_{2k + 1}}(t, s)\mathcal{U}_{Y_{2k + 1}}(s, t)\mathcal{U}_{Y_{2k + 2}}(s, t)\mathcal{U}_{X_{2k}}(t, s)\big)\mathcal{U}_{Z_{2k + 2}}(t, s) + \Delta_{k}(t, s),
    \end{align}
    where
    \begin{align}
        \Delta_k(t, s) = \Delta_{Z_{2k + 2},Y_{2k + 1}, Y_{2k} }(t, s) + \Delta_{Y_{2k + 1}, Y_{2k +2 }, Z_{2k + 3}}(t,s)\mathcal{U}_{Y_{2k + 1}}(s, t) \mathcal{U}_{X_{2k}}(t, s)
    \end{align}
    Recursing this equation and noting that $\mathcal{U}_{Z_0}(t, s) = \mathcal{U}_{[1:N]}(t, s) = \mathcal{U}(t, s)$, we obtain that
    \begin{align}
        \mathcal{U}(t, s) - \mathcal{U}^{(l)}(t, s) = \sum^{\lfloor N/2l\rfloor}_{k=0}\mathcal{V}_k(t, s)\Delta_k(t, s),
    \end{align}
    where
    \begin{equation}
		\begin{aligned}
			\mathcal{V}_k(t, s) &= \bigg(\prod^{k-1}_{k' = 0}\mathcal{U}_{X_{2k' + 1}}(t, s)\bigg) \bigg(\prod^{2k-1}_{k'=0}\mathcal{U}_{Y_{k' + 1}}(s, t)\bigg) \bigg(\prod^{k-1}_{k' = 0}\mathcal{U}_{X_{2k'}}(t, s)\bigg)  \\
			&= \mathcal{U}_{\mathcal{A}_k}(t, s) \mathcal{U}_{\mathcal{B}_k}(s, t)\mathcal{U}_{\mathcal{C}_k}(t, s),
		\end{aligned}
    \end{equation}
    with $\mathcal{A}_k =\bigcup^{k-1}_{k' = 0}\{\text{edges } e: e \subseteq X_{2k' + 1}\}$, $\mathcal{B}_k =\bigcup^{2k-1}_{k' = 0}\{\text{edges } e: e \subseteq Y_{k' + 1}\}$ and $\mathcal{C}_k = \bigcup^{k-1}_{k' = 0}\{\text{edges } e: e \subseteq X_{2k'}\}$.

    Next, consider the error between the system-environment unitaries $\mathcal{U}(t,s)$ and $\mathcal{U}^{(l)}(t, s)$ when the environment was initially (at time $= s$) in the vacuum state and the environment is traced out at time $= t$: For a system observable $O_S$ and an initial system state $\sigma_S$,
    \begin{equation}
		\begin{aligned}
			&\abs{\vecbra {O_S,I_E}(\mathcal{U}(t, s) - \mathcal{U}^{(l)}(t, s))\vecket{\sigma_{S}, \text{vac}}} \leq \sum^{\lfloor N/2l \rfloor}_{k=0}\abs{\vecbra{O_S, I_E}\mathcal{V}_{k}(t, s)  \Delta_k(t, s)\vecket{\sigma_S, \text{vac}}} \\
			&\leq \sum^{\lfloor N/2l \rfloor}_{k=0}\bigg(\abs{\vecbra{O_S, I_E}\mathcal{U}_{\mathcal{A}_k}(t, s) \mathcal{U}_{\mathcal{B}_k}(s, t)\mathcal{U}_{\mathcal{C}_k}(t, s)\Delta_{Z_{2k + 2},Y_{2k + 1}, Y_{2k} }(t, s)\vecket{\sigma_S, \text{vac}}} + \\
			&\qquad \abs{\vecbra{O_S, I_E }\mathcal{U}_{\mathcal{A}_k}(t, s) \mathcal{U}_{\mathcal{B}_k}(s, t)\mathcal{U}_{\mathcal{C}_k}(t, s)  \Delta_{Y_{2k + 1}, Y_{2k +2 }, Z_{2k + 3}}(t,s)\mathcal{U}_{Y_{2k + 1}}(s, t) \mathcal{U}_{X_{2k}}(t, s)\vecket{\sigma_S, \text{vac}}} \bigg).
		\end{aligned}
    \end{equation}
    Each term in the summation on right-hand side in this inequality can be upper bounded using \cref{lemma:splitting_remainder,lemma:superop_comm_bound}: Using \cref{lemma:splitting_remainder} for $\Delta_{Z_{2k + 2}, Y_{2k + 1}, Y_{2k}}$ and $\Delta_{Y_{2k + 1}, Y_{2k + 2}, Z_{2k + 3}}$, we can express each of these terms in the form considered in \cref{lemma:superop_comm_bound} with $d(e, f) = l$ and $M, M' \leq 10$. Thus, each term can be upper bounded by $c_0 \exp(v_0 \abs{t - s} - l/\xi_0)$ for some constants $c_0, v_0$ and $\xi_0$. Noting that the number of terms in the summation is $O(N)$, the lemma follows.
\end{proof}

\subsection{Collision models}
\label{sec:collision}
The next step in our analysis is to digitize the ancillary Fock spaces introduced in the previous section as a dilation of the Lindbladian. In this section, we adopt the strategy used commonly in the ``Collision model" framework for open quantum systems \cite{Ciccarello2022} to digitize the dynamics. We first develop the digitization procedure more generally without restricting ourselves to the specific unitary obtained in \cref{lemma:chopping_unitaries}, and then apply this procedure to this specific case.

Consider a system-environment unitary with $M$ stages of the form 
\begin{align}\label{eq:collision_model_U}
U=U_{\mathcal{A}_{M}}(t_{M},s_{M})U_{\mathcal{A}_{M - 1}}(t_{M - 1},s_{M - 1})\dots U_{\mathcal{A}_{1}}(t_{1},s_{1}) ,
\end{align}
where $\mathcal{A}_{i}\subseteq\{(1,2),(2,3),\ldots,(N-1,N)\}$, are some set of edges, $t_{i}$ can either be larger than $s_{i}$
(corresponding to forward time evolution) or smaller than $s_{i}$
(corresponding to backward time evolution). Note that the unitary developed in \cref{lemma:chopping_unitaries} is of this form with $M = 3$ stages. The system channel of interest will be obtained by applying the system-environment unitary $U$ on an initial state with the environment (i.e., the ancillary Fock spaces) being in the vacuum state and then tracing out the environment:
\begin{align}\label{eq:multi_stage_U_channel}
\mathcal{E}(X)=\operatorname{Tr}_{E}\!\left[U\bigl(X\otimes|\mathrm{vac}\rangle\langle\mathrm{vac}|\bigr)U^{\dagger}\right].
\end{align}
The underlying principle behind the collision model is rather simple: Given a time step $\delta > 0$, an ancillary Fock space $\text{Fock}(L^2(\mathbb{R}))$ can be decomposed as\footnote{To be precise, the tensor product here is not the usual tensor product used for finite-dimensional spaces, but the grounded tensor product of a set of Fock spaces.}
\begin{align}\label{eq:fock_space_decomposition}
\text{Fock}(L^2(\mathbb{R})) \cong \bigotimes_{n\in \mathbb{Z}} \text{Fock}(L^2(I_n)) \quad \text{where}\quad I_n = ((n - 1)\delta, n\delta),
\end{align}
with the vacuum of the Fock space $\text{Fock}(L^2(\mathbb{R}))$ being identified with $\otimes_{n\in\mathbb{Z}} \ket{\text{vac}_n}$ where $\ket{\text{vac}_n}$ is the vacuum corresponding to $\text{Fock}(L^2(I_n))$. With this decomposition of each of the ancillary Fock-spaces, the annihilation operators $a_{\alpha,e}(t)$ for $t \in I_n$ can be assumed to act only on the Fock space $\text{Fock}(L^2(I_n))$. Consequently, since the dilated Hamiltonian described in \cref{eq:hamiltonian_dilation_nn} is time-local (i.e., the Hamiltonian at time $t$ is expressible in terms of the annihilation operators $a_{\alpha, e}(t)$), we then obtain that a unitary of the form $\mathcal{U}_{\mathcal{A}}(n\delta, (n - 1)\delta)$ or $\mathcal{U}_{\mathcal{A}}((n - 1)\delta, n\delta)$ acts entirely on the Fock space $\text{Fock}(I_n)$. With this in mind, to set up a collision model for the unitary in \cref{eq:collision_model_U}, we proceed as follows:
\begin{enumerate}
    \item First, we represent each of the time evolutions for each of the stages $U_{\mathcal{A}_i}(t_i, s_i)$ as a product of time evolutions over the intervals $I_n =((n - 1)\delta, n\delta)$. Using the decomposition of the ancillary Fock spaces as shown in \cref{eq:fock_space_decomposition}, it follows that only the unitaries during the time-interval $I_n$ interact with the same ancillary Fock space $\text{Fock}(I_n)$.
    \item The unitaries within each time step are Trotterized with a simple first-order Trotter formula and expressed as a product of evolution due to each of the terms in the system Hamiltonian i.e., $\exp(-i\delta h_e)$ as well as in the system-environment Hamiltonian i.e., $\exp[-i \int_{n\delta}^{(n + 1)\delta} (a_{\alpha, e}(s) L_{\alpha, e}^\dagger + \text{h.c.})ds]$.
    \item Finally, each Fock space over $I_n$, i.e., $\text{Fock}(I_n)$, is approximated with a single ancillary qubit with the identification
    \begin{equation}
    \frac{1}{\sqrt{\delta}} \int_{(n - 1)\delta}^{n\delta}a_{\alpha,e}(t) dt \quad \text{with}\quad \sigma_{\alpha, e}[n],
    \end{equation}
    where $\sigma_{\alpha, e}[n]$ is the lowering operator on the introduced ancillary qubit.
\end{enumerate}
\begin{figure}
    \centering
    \includegraphics[width=1.0\linewidth]{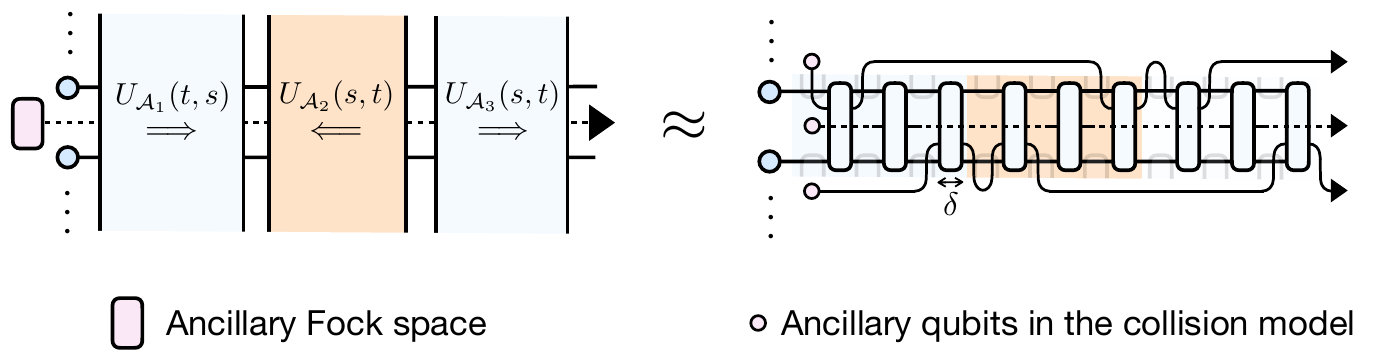}
    \caption{Collision-model approximation (right) of a sequence of forward and backward system--environment evolutions (left).
    The bosonic fields are replaced by ancillary qubits associated with time bins of width $\delta$.
    Each ancillary qubit is used three times: first during forward evolution, again during the corresponding backward time step and a third time in the final forward time step.
    In the end, the ancillary qubits are traced out.}
    \label{fig:collision-model}
\end{figure}
The final collision model mapping is schematically depicted in \cref{fig:collision-model} with an example of a three-stage unitary: We reemphasize that, as depicted in the figure, the ancillary qubits introduced in the final collision model can possibly interact with the system qudits in each stage.

We now define the collision model more precisely and provide an upper bound on the error between the collision model and the exact evolution. For the time step $\delta > 0$, assume that 
\begin{equation}
t_{i}=n_{i}\delta\quad\text{and}\quad s_{i}=m_{i}\delta,
\end{equation}
for some $n_{i},m_{i}\in\mathbb{Z}$: We remark that this implicitly assumes that the time-differences $\abs{t_i - s_i}$ are commensurate (i.e., related to each other via multiplication by rational numbers), however this won't be very restrictive in our case and can always be satisfied. It will be convenient to define
a signed variable $\epsilon_{i}$ indicating whether
we are dealing with a forward or backward time evolution in the $i^{\text{th}}$
unitary $U_{\mathcal{A}_{i}}(t_{i},s_{i})$: 
\begin{equation}
\epsilon_{i}=+1\text{ if }n_{i}\geq m_{i}\text{ and }-1\text{ if }n_{i}<m_{i}.
\end{equation}
Then, we can write the unitary for each stage, $U_{\mathcal{A}_i}(t_i,s_i)$ as a product of unitaries corresponding to time steps of length $\delta$ i.e.,
\begin{subequations}\label{eq:decomposing_unitary_timestep}
\begin{align}
U_{\mathcal{A}_i}(t_i, s_i)=\prod_{j:n_{i}\to m_{i}}U_{i}[j]
\end{align}
where $U_{i}[j]$ is a system-environment unitary, which is either forward or
backward depending on $\epsilon_i$, corresponding to
the time step $(j-1)\delta\to j\delta$ : 
\begin{align}
U_{i}[j]=\begin{cases}
U_{\mathcal{A}_{i}}(j\delta,(j-1)\delta) & \text{ if }\epsilon_{i}=+1,\\
U_{\mathcal{A}_{i}}((j-1)\delta,j\delta) & \text{ if }\epsilon_{i}=-1.
\end{cases}
\end{align}
and 
\begin{align}
\prod_{j:n_{i}\to m_{i}}U_{i}[j]=\begin{cases}
U_{i}[n_{i}]U_{i}[n_{i}-1]\dots U_{i}[m_{i}+1] & \text{if }\epsilon_{i}=+1,\\
U_{i}[n_{i}+1]U_{i}[n_{i}+2]\dots U_{i}[m_{i}] & \text{if }\epsilon_{i}=-1.
\end{cases}
\end{align}
\end{subequations}
Now, for every time step $(j-1)\delta \to j\delta$ as well as each edge $e$ and index $\alpha$, we introduce an ancillary qubit with lowering operator $\sigma_{\alpha, e}[j]$. We then approximate each of the unitaries corresponding to this time step, $U_i[j]$, with $\hat{U}_i[j]$ as follows: For $\epsilon_i = +1$ (i.e., forward-time evolution), 
\begin{align}
    \hat{U}_{i}[j] = \bigg[\prod_{e \in \mathcal{A}_i^\text{e}} \bigg(e^{-i \delta h_e} \prod_{\alpha}e^{-i\sqrt{\delta}(\sigma^\dagger_{\alpha, e}[j] L_{\alpha, e} +\text{h.c.})}\bigg)\bigg] \bigg[\prod_{e \in \mathcal{A}_i^\text{o}} \bigg(e^{-i \delta h_e} \prod_{\alpha}e^{-i\sqrt{\delta}(\sigma^\dagger_{\alpha, e}[j] L_{\alpha, e} +\text{h.c.})}\bigg)\bigg],
\end{align}
where $\mathcal{A}_i^\text{e} = \mathcal{A}_i \cap \{(2x, 2x + 1): x\in [1: \lceil N/2\rceil )\}$ is the set of even edges in $\mathcal{A}_i$ and $\mathcal{A}_i^\text{o} = \mathcal{A}_i \cap \{(2x - 1, 2x) : x \in [1:\lceil N/2\rceil]\}$ is the set of odd edges in $\mathcal{A}_i$. Similarly, for $\epsilon_i = -1$ (i.e., backward-time evolution),
\begin{align}
    \hat{U}_{i}[j] = \bigg[\prod_{e \in \mathcal{A}_i^\text{o}} \bigg(e^{i \delta h_e} \prod_{\alpha}e^{i\sqrt{\delta}(\sigma^\dagger_{\alpha, e}[j] L_{\alpha, e} +\text{h.c.})}\bigg)\bigg]\bigg[\prod_{e \in \mathcal{A}_i^\text{e}} \bigg(e^{i \delta h_e} \prod_{\alpha}e^{i\sqrt{\delta}(\sigma^\dagger_{\alpha, e}[j] L_{\alpha, e} +\text{h.c.})}\bigg)\bigg].
\end{align}
The full collision model for $U$, denoted by $\hat{U}$ is then given by
\begin{align}
    \hat{U} = \hat{U}_M \hat{U}_{M - 1} \dots \hat{U}_1,
\end{align}
where $\hat{U}_i$, the collision model corresponding to the $i^\text{th}$ stage, is obtained by replacing $U_i[j]$ by $\hat{U}_i[j]$ in \cref{eq:decomposing_unitary_timestep}. The system-channel corresponding to the collision model can then be defined analogously to \cref{eq:multi_stage_U_channel}:
\begin{align}
    \hat{\mathcal{E}}(X) = \Tr _E[\hat{U}(X \otimes \ket{\bm{0}}_E\!\bra{\bm{0}})\hat{U}^\dagger],
\end{align}
where $E$ now refers to all the ancillary qubits in the collision model, $\ket{\bm{0}}_E$ refers to the state in which all the ancillas are in $\ket{0}$. The main lemma in this subsection then upper-bounds the error between $\hat{\mathcal{E}}$ and $\mathcal{E}$.
\begin{lemma}\label{lemma:collision_model_bound}
    For $4M^2\delta \leq 1$, 
    \begin{align}
        \norm{\mathcal{E} - \hat{\mathcal{E}}}_\diamond \leq C M^4 N t_\textnormal{max} \delta,
    \end{align}
    where $t_\textnormal{max} = \max_{i, j \in \{1, 2 \dots M\}} \abs{t_i - s_j}$ and $C$ is a numerical constant.
\end{lemma}
The proof of this lemma uses Taylor expansion of the unitaries corresponding to the different stages and carefully accounting for the second and higher-order remainders. A full proof is provided in \cref{sec:collision-analysis}.

Finally, we apply this lemma to the concrete unitary developed in \cref{lemma:chopping_unitaries}. Starting from the 3-layer unitary $U^{(l)}(t, s) = U^{(l)}_3(t, s) U^{(l)}_2(s, t) U^{(l)}_1(t, s)$ defined in \cref{lemma:chopping_unitaries} and using the procedure described above, we obtain the corresponding collision model unitary $\hat{U}^{(l)} = \hat{U}^{(l)}_3\hat{U}^{(l)}_2 \hat{U}^{(l)}_1$. We remark that this collision model has $O(N \abs{t - s}/\delta)$ ancillary qubits. Using \cref{lemma:collision_model_bound}, we then straightforwardly obtain the following error bound.
\begin{lemma}\label{lemma:collision_model_chopped_unitary}
    Suppose $\hat{U}^{(l)} = \hat{U}^{(l)}_3\hat{U}^{(l)}_2 \hat{U}^{(l)}_1$ is the collision model, with time step $\delta = (t - s)/T_0$ for some $T_0 \in \mathbb{Z}_{>0}$, corresponding to the unitary $U^{(l)}(t, s) = U^{(l)}_3(t, s) U^{(l)}_2(s, t) U^{(l)}_1(t, s)$ defined in \cref{lemma:chopping_unitaries}, then
    \begin{align}
        \norm{\mathcal{E}^{(l)}(t, s) - \hat{\mathcal{E}}^{(l)}}_\diamond \leq O\big(N(t-s)\delta\big),
    \end{align}
    where $\mathcal{E}^{(l)}(t, s) = \Tr _E[U^{(l)}(t, s)((\cdot)\otimes \ket{\textnormal{vac}}\!\bra{\textnormal{vac}})U^{(l)}(s, t)]$ and $\hat{\mathcal{E}}^{(l)} = \Tr _E[\hat{U}^{(l)}((\cdot)\otimes \ket{\bm 0_E}\!\bra{\bm 0_E})\hat{U}^{(l)}]$.
\end{lemma}
\begin{proof}
    The proof of this follows directly from \cref{lemma:collision_model_bound} and noting that the unitary defined in \cref{lemma:chopping_unitaries} has $M = 3$.
\end{proof}

\subsection{Compressing and simulating the collision model}
\label{sec:compressing}
The collision model introduced in the previous section is digitized i.e., the infinite-dimensional bosonic Fock spaces introduced as a dilation of the Lindbladian have been approximated away by replacing them with qubits.
At this point, we can simulate the collision model on a quantum computer: However, as shown in \cref{lemma:collision_model_chopped_unitary}, to control the error in the collision model, the time step $\delta$ has to be chosen as $1/\poly(N, t)$. Thus, a direct simulation of the collision model will not be optimal.

Instead, here we use that since each ancillary qubit only interacts with the system for three short time steps of length $\delta$, it remains predominantly in its initial $\ket{0}$ state.
More precisely, we can truncate the total number of excitations generated in the enviroment during a time interval $t-s$ to some number $n_c$, while incurring an error that is independent of the number of time steps $T_0$ used, and exponentially small in $n_c$.
The Hilbert space dimension of the environment truncated to $n_c$ excitations scales only polynomially in $T_0$ rather than exponentially, which allows us to attain the optimal scalings in \cref{th:algorithm}. We remark that the idea of limiting the ancillary Hilbert space to low-excitation spaces has appeared in previous work by Cleve and Wang \cite{CleveWang2017}, although not in the context of lattice Lindbladians.

\subsubsection{Truncating excitation numbers}
In this subsection, we develop some technical lemmas that allow us to rigorously truncate the Hilbert space of the ancillary qubits and quantify the truncation error.

\begin{lemma}\label{lemma:expand_exp}
On the Hilbert space $\mathcal{H}_S\otimes \mathbb{C}^2$, suppose $V = L\otimes \sigma^\dagger + \textnormal{h.c.}$, with $\sigma = \ket{0}\!\bra{1}$.
Then,
\begin{align}
e^{-i\sqrt{\delta}V} = \sum_{a, b \in \{0, 1\}}(\sqrt{\delta})^{a\oplus b}{R}^{(\delta)}_{a,b}(L) \otimes \ket{a}\!\bra{b},
\end{align}
where
\begin{equation}
    \begin{aligned}
    {R}^{(\delta)}_{0,0}(L) &= \cos(\sqrt{\delta} \sqrt{L^\dagger L}) ,\quad
    {R}^{(\delta)}_{1,0}(L) =- iL \textnormal{sinc}(\sqrt{\delta}\sqrt{L^\dagger L}),\\
    {R}^{(\delta)}_{1,1}(L) &= \cos(\sqrt{\delta}\sqrt{LL^\dagger}),\quad
    {R}^{(\delta)}_{0,1}(L) = -iL^\dagger \textnormal{sinc}(\sqrt{\delta}\sqrt{LL^\dagger}).
    \end{aligned}
\end{equation}
\end{lemma}
\begin{proof}
    We note that
    \begin{equation}
    V^2 = L^\dagger L \otimes \proj0 + LL^\dagger \otimes \proj1,
    \end{equation}
    and therefore
    \begin{subequations}\label{eq:powers_V}
    \begin{align}
    &V^{2n} = (L^\dagger L)^n \otimes \ket{0}\!\bra{0} + (LL^\dagger)^n \otimes \ket{1}\!\bra{1},\\
    &V^{2n + 1} = L(L^\dagger L)^{n} \otimes \ket{1}\!\bra{0} + L^\dagger (LL^\dagger)^{n}\otimes \ket{0}\!\bra{1}.
    \end{align}
    \end{subequations}
    Now, using the Taylor expansion of $\exp(-i\sqrt{\delta}V)$, we obtain that
    \begin{align}
\exp(-i\sqrt{\delta}V) = \sum^{\infty}_{n = 0}\frac{(-1)^n}{(2n)!}\delta^{n} V^{2n} - i\sqrt{\delta}\sum^{\infty}_{n = 0}\frac{(-1)^n}{(2n + 1)!} \delta^n V^{2n + 1}.\label{eq:taylor_exp_exp}
\end{align}
    Substituting \cref{eq:powers_V} into \cref{eq:taylor_exp_exp} and identifying $\cos x = \sum^{\infty}_{n = 0}(-1)^n x^{2n} / (2n)!$ and $\text{sinc} \ x = \sum^{\infty}_{n = 0}(-1)^n x^{2n} / (2n + 1)!$, we obtain the lemma.
\end{proof}
\begin{lemma}\label{lemma:basic_op_ineq}
    Suppose $V = L\otimes \sigma^\dagger + \textnormal{h.c.}$ and $\hat n = \sigma^\dagger \sigma=\proj1$. Then, for $t \geq 1$,
    \begin{align}
        e^{i\sqrt{\delta} V} t^{\hat n}e^{-i\sqrt{\delta} V} \preceq (1 + 2\delta(t-1)) (2t - 1)^{\hat n}.
    \end{align}
\end{lemma}
\begin{proof}
    Note that $t^{\hat n} - I = (t -1)\proj1$. Therefore,
    \begin{align}
         e^{i\sqrt{\delta} V} t^{\hat n}e^{-i\sqrt{\delta} V} = I + (t - 1)e^{i\sqrt{\delta} V} (I\otimes \ket{1}\!\bra{1})e^{-i\sqrt{\delta} V}.
    \end{align}
    Now consider a state $\ket{\psi} = \ket{\phi_0}\ket{0}+\ket{\phi_1}\ket{1}$. Then,
    \begin{align}\label{eq:expan_t_n}
        \bra{\psi}e^{i\sqrt{\delta} V} t^{\hat n} e^{-i\sqrt{\delta} V}\ket{\psi} = \norm{\psi}^2 + (t - 1)\bra{\psi}e^{i\sqrt{\delta}V}(I\otimes \ket{1}\!\bra{1})e^{-i\sqrt{\delta}V}\ket{\psi}.
    \end{align}
    Additionally,
    \begin{align}
        \bra{1}e^{-i\sqrt{\delta}V}\ket{\psi} = \sqrt{\delta}R^{(\delta)}_{1,0}(L)\ket{\phi_0} +  R^{(\delta)}_{1,1}(L)\ket{\phi_1},
    \end{align}
    where $R^{(\delta)}_{0, 0}$ and $R^{(\delta)}_{1, 1}$ are as defined in lemma \ref{lemma:expand_exp}. Using the fact that $\norm{R_{1, 0}^{(\delta)}(L)}, \norm{R_{1, 1}^{(\delta)}(L)} \leq 1$, we obtain that
    \begin{align}
        \norm{\bra{1}e^{-i\sqrt{\delta}V}\ket{\psi}}^2 &\leq \big(\sqrt{\delta}\norm{{\phi_0}} + \norm{{\phi_1}}\big)^2\leq 2(\delta \norm{\phi_0}^2 + \norm{\phi_1}^2).
    \end{align}
    This with Eq.~\eqref{eq:expan_t_n} yields
    \begin{align}
         \bra{\psi}e^{i\sqrt{\delta} V} t^n e^{-i\sqrt{\delta} V}\ket{\psi} &\leq \norm{\psi}^2 + 2(t - 1)(\delta \norm{\phi_0}^2 + \norm{\phi_1}^2)\nonumber\\
         &\leq (1 + 2\delta(t-1))\norm{\phi_0}^2 + (2t - 1)\norm{\phi_1}^2 \nonumber \\
         &\leq (1 + 2\delta(t-1))\bra{\psi}(2t - 1)^{\hat n} \ket{\psi},
    \end{align}
    where in the last step, we have used that $t \geq 1$. Since this holds for any $\ket{\psi}$, the lemma statement follows.
\end{proof}

We next bound the leakage to states with higher excitation number by bounding the moment-generating function of the number operator on each ancilla.
\begin{lemma}\label{lemma:ex_trunc}
Consider a unitary $\hat{U}$ on a system $S$ with $K$ ancillas ($\mathcal{H}_S \otimes ((\mathbb{C})^{2})^{\otimes K}$) where $\hat{U} = \prod^{K}_{i = 1}W_i e^{-i\sqrt{\delta}(L_i \sigma^\dagger_{i}+ L^\dagger_{i}\sigma_i)}$ where $W_i$ are system unitaries and $L_i$ are system operators with $\norm{L_i} \leq 1$. Then,
\begin{align}
    \norm{(I - \Pi^{\leq m_f}) \hat{U}\Pi^{\leq m_i}} \leq e^{2\delta K} e^{m_i}e^{-m_f/2},
\end{align}
where $\Pi^{\leq m}$ is the projector on the subspace of $\mathcal{H}_S\otimes (\mathbb{C}^2)^{\otimes K}$ with at most $m$ ancillas being in $\ket{1}$.
\end{lemma}
\begin{proof}
    Let $\hat n =\sum_{i = 1}^K \sigma_i^\dagger \sigma_i$.
    Suppose $\ket{\psi} \in \text{Im}(\Pi^{\leq m_i})$, it follows from a repeated application of lemma \ref{lemma:basic_op_ineq} that
    \begin{align}
        \bra{\psi}\hat{U}^\dagger t^{\hat n}\hat{U}\ket{\psi}\preceq  (1 + 2\delta(t-1))^K \bra{\psi}(2t- 1)^{\hat n}\ket{\psi} \leq  (1 + 2\delta(t-1))^K  (2t - 1)^{m_i}.
    \end{align}
    Since $I - \Pi^{\leq m_f} \preceq  t^{\hat n}/t^{m_f + 1}$,
    \begin{align}
        \norm{\Pi^{>m_f} \hat{U}\ket{\psi}}^2 \leq t^{-(m_f + 1)}\bra{\psi}\hat{U}^\dagger t^{\hat n}\hat{U}\ket{\psi} \leq (1 + 2\delta(t-1))^K  (2t - 1)^{m_i} t^{-(m_f + 1)} .
    \end{align}
    Since this is true for any $\ket{\psi} \in \text{Im}(\Pi^{\leq m_i})$ and any $t \geq 1$, setting $t = e$ and upper-bounding $(1 + 2\delta(t - 1))^K \leq e^{2(e-1)\delta K} \leq e^{4\delta K}$, $(2t -1)^{m_i} \leq e^{2m_i}$ we obtain the lemma statement.
\end{proof}

\subsubsection{Simulating the compressed collision model}
\label{sec:compressed-collision-model}
Finally, we show how to efficiently simulate the compressed collision model from \cref{lemma:ex_trunc}. The previous section establishes that the state of ancillary qubits can be well approximated by its low-excitation restriction. We next introduce a compressed representation of the ancilla state which uses the fact that it does not have too many excitations and only tracks the locations of the different excitations.
\begin{definition}[Excitation location representation]\label{def:ex_loc}
Suppose $B^{\leq n_c} \subseteq \{0, 1\}^n$ is defined as $B^{\leq n_c} = \{b \in \{0, 1\}^n: \abs{b} = \sum_i b_i \leq n_c\}$ i.e., it is the set of $n$-bit strings that have at most $n_c$ ones. Then, for $b \in B^{\leq n_c}$, $\mathsf{ex}(b)$ is the list
\begin{align}
    \mathsf{ex}(b) = \{j_1, j_2 \dots j_{\abs{b}}\} \quad \text{such that}\quad b_{j_i} = 1 \text{ for }i \in [1,\abs{b}] \quad \text{and}\quad j_1 < j_2 \dots <j_{\abs{b}},
\end{align}
i.e., $\mathsf{ex}(b)$ stores a sorted list of locations of $1s$ in $b$. Furthermore, $\mathsf{ex}(b)$ is stored as a bit-string with $n_c\ceil{\log_2 n} + \ceil{\log_2(n_c + 1)} = O(n_c \log n)$ bits as follows:
\begin{enumerate}
    \item[(1)] The first $\abs{b}$ blocks of $\ceil{\log_2 n}$ bits are used to store all the integers $\{j_1, j_2 \dots j_{\abs{b}}\}$.
        \item[(2)] The next $n_c - \abs{b}$ blocks of $\ceil{\log_2 n}$ bits are set to $0$.
        \item[(3)]  The last $\ceil{\log(n_c + 1)}$ bits are used to store $k \in [0,n_c]$.
\end{enumerate}
Furthermore, given sets $B_i^{\leq n_c} = \{b \in \{0, 1\}^{n_i}: \abs{b} \leq n_c\}$ for $i=1,\dots,p$, and $n_i\leq n$,
an excitation number representation for $b = (b(1), b(2) \dots b(p)) \in B_1^{\leq n_c}\times B_2^{\leq n_c} \times \dots B_p^{\leq n_c}$ is defined by
\begin{align}
\mathsf{ex}(b) = \mathsf{ex}(b(1)) \cup \mathsf{ex}(b(2)) \cup \dots \mathsf{ex}(b(p)),
\end{align}
and it will be stored in $p(n_c\ceil{\log_2 n} + \ceil{\log_2(n_c + 1)})$ bits by concatenating the bits representing $\mathsf{ex}(b(i))$.
\end{definition}

In this section, we will focus on designing an efficient circuit that implements a general collision model with ancillas being projected to a low-excitation state. More specifically, consider a system with $n$ qudits, a Hamiltonian $H = \sum_{j = 1}^{K} h_j$ where $h_j$ are $O(1)-$local Hermitian operators with $\norm{h_j} \leq1$ and a set of $O(1)-$local jump operators $L_1, L_2 \dots L_{K}$ satisfying $\norm{L_{i}} \leq 1$. Construct a collision unitary $\hat{U}$ with time step $\delta$, a total of $T_0 = t_0/\delta$ time steps on the system qudits and $K T_0$ ancillary qubits
        \begin{align}\label{eq:sc_collision_model}
        \hat{U} = \hat U_{T_0}\cdots \hat U_1, \quad\text{where }\hat U_\tau \coloneqq \prod_{j = 1}^{K} \exp({-i\delta h_j}) \exp[-i\sqrt{\delta}(L_j^\dagger \sigma_j[\tau] + \text{h.c.})],
        \end{align}
    where $\sigma_j[\tau]$ is the lowering operator for the ancillary qubit introduced at time step $\tau$ corresponding to the $j^\text{th}$ jump operator. Furthermore, suppose that the set of jump operators is divided into $p=O(1)$ groups, $\mathcal{J}_1, \mathcal{J}_2 \dots \mathcal{J}_p \subseteq \{1, 2 \dots K\}$.
    Then, the set $\mathcal{J}_i \times [1,T_0]$ indexes all ancillas corresponding to the jump operators in $\mathcal{J}_i$ introduced during evolution under $\hat{U}$.
    Given an excitation number cutoff $n_c>0$, we let $B_i^{\leq n_c}$ be the set of computational basis configurations of the ancillas in $\mathcal{J}_i \times [1:T_0]$ with $\leq n_c$ excited ancillas i.e.,
    \begin{align}
        B_i^{\leq n_c} = \bigg\{\bm b \in \{0, 1\}^{{\mathcal{J}_i}\times [1:T_0]} : \sum_{j, \tau} \bm b(j, \tau) \leq n_c\bigg\},
    \end{align}
    and let $B^{\leq n_c}$ defined by
    \begin{align}
        B^{\leq n_c} = B_1^{\leq n_c} \times B_2^{\leq n_c}\times \dots B_p^{\leq n_c}.
    \end{align}
    $B^{\leq n_c}$ is thus the set of computational basis configurations of the ancillas where the ancillas in each of the $p$ groups defined above have $\leq n_c$ excitations.
    Furthermore, denote by $\mathcal{S}_i^{\leq n_c}$ and $\mathcal{S}^{\leq n_c}$ denote the subspaces corresponding $B_i^{\leq n_c}$ and $B^{\leq n_c}$ respectively i.e.,
    \begin{align}
        \mathcal{S}_i^{\leq n_c} = \textnormal{span}(\{\ket{\bm b_i} : \bm b_i \in B_i^{\leq n_c}\}) \quad \text{and}\quad \mathcal{S}^{\leq n_c} = \bigotimes_{i = 1}^p \mathcal{S}_i^{\leq n_c}.
    \end{align}
    
     In this subsection, we are interested in considering the setting where the ancillas are initially in $\mathcal{S}^{\leq n_c}$, the system and ancillas are evolved under the collision model and finally projected back into $\mathcal{S}^{\leq n_c}$. Specifically, we aim at designing a unitary circuit implementing (approximately) the possibly non-unitary operator
    \begin{align}
    \Pi^{\leq n_c^{\text{out}}}\hat{U}\Pi^{\leq n_c^\text{in}},
    \end{align}
    where $\Pi^{\leq n}$ is the projector on $\mathcal{S}^{\leq n}$ and $n_c^{\text{in}}, n_c^\text{out} \leq n_c$. We remark that we consider $\hat{U}$ being applied on a state with at most $n_c^\text{in}$ excitations and consider the restriction of the output space to a possibly different space ($n_c^\text{out}\neq n_c^\text{in}$). This is because, as shown in the previous subsection, application of $\hat{U}$ on a space with $n_c^\text{in}$ excitations can possibly grow the number of excitations. However, for most part while constructing the required circuits, we will work with the ancillary space $\mathcal{S}^{\leq n_c}$ since $\mathcal{S}^{\leq n_\text{in}}, \mathcal{S}^{\leq n_\text{out}}$ are subspaces of $\mathcal{S}^{\leq n_c}$.
    
    Finally we also remark that, as discussed in \cref{sec:collision}, to well approximate continuous-time dynamics $1/\delta$ scales polynomially with the system-size and evolution time. Therefore, our goal is to design a unitary circuit approximating the collision model with number of gates that scale at most logarithmically with $1/\delta$ and this is the key challenge addressed in this section. The next lemma, which is the main result of this section, shows that this is indeed possible for short evolution times $t_0$.
    
    \begin{lemma}[Compressed simulation of a collision model]\label{lemma:coll_model_main}
    Suppose $\norm{(I - \Pi^{\leq n_c^\textnormal{out}})\hat{U}\Pi^{\leq n_c^\textnormal{in}}} \leq \epsilon_c$ with $n_c^\textnormal{in}, n_c^\textnormal{out} \leq n_c$. Furthermore, assume that $n_c t_0 K p \leq 1$ and $\delta \leq \epsilon/6K^2 t_0$, then there is a unitary $V$ acting on $n$ system qudits and the ancillary subspace $\mathcal{S}^{\leq n_c}$ represented with $O(n_c \log_2(KT_0))$ qubits using the excitation-location representation (\cref{def:ex_loc}) and satisfying 
    \[
    \norm{V - \Pi^{\leq n_c^\textnormal{out}}\hat{U} \Pi^{\leq n_c^\textnormal{in}}}\leq \epsilon + O(\epsilon_c),
    \]
    that can be implemented with $O(\textnormal{poly}(K, n_c, \log(1/\delta), \log(1/\varepsilon)))$ single and two-site gates and $O(\textnormal{poly}(K, n_c, \log(1/\delta), \log(1/\varepsilon)))$ additional workspace qubits.
    \end{lemma}
    
    The rest of this section is devoted to proving this lemma.  Throughout this analysis, we will employ the excitation-location representation (\cref{def:ex_loc}) to represent the Hilbert space of the ancillas. We remark that the ancillas in the collision model are indexed by two integers $(j, \tau)$, with $j$ corresponding to the jump operator that they correspond to and $\tau$ representing the time step in which the ancilla is introduced. Given a list of excitation locations, we will order them in increasing order of $\tau$ and in case of locations with the same $\tau$, order them in increasing order of $j$.
    
    We first begin by providing a compact expression for $\Pi^{\leq n_c}\hat{U}\Pi^{\leq n_c}$.

\begin{lemma}\label{lemma:expression_matrix_elements}
    The operator $ \Pi^{\leq n_c} \hat{U} \Pi^{\leq n_c} $ is given by
    \begin{align}
  \Pi^{\leq n_c} \hat{U} \Pi^{\leq n_c} = \sum_{\bm a, \bm b \in B^{\leq n_c}}(\sqrt{\delta})^{\abs{\bm a \triangle \bm b} }\tilde{R}^{(\delta)}_{\bm a, \bm b}\otimes \ket{\mathsf{ex}(\bm a)}\!\bra{\mathsf{ex}(\bm b)},
\end{align}
where $\bm a \triangle \bm b $ is given by
\begin{equation}
\bm a \triangle \bm b = \{(j, \tau): \bm a(j, \tau) \oplus \bm b(j, \tau) = 1\},
\end{equation}
and
\begin{align}\label{eq:tilde_R_full}
    {R}^{(\delta)}_{\bm a, \bm b} = \prod_{\tau = T}^1 \prod_{j = 1}^K e^{-i\delta h_j} \tilde{R}^{(\delta)}_{\bm a(j, \tau), \bm b(j, \tau)}(L_j),
\end{align}
with $\tilde{R}^{(\delta)}_{a,b}(L)$ being as defined in \cref{lemma:expand_exp}.
    \end{lemma}
    \begin{proof}
    From \cref{lemma:expand_exp} it follows that
    \begin{align}
    \exp(-i\sqrt{\delta}(L_j^\dagger \sigma_j[\tau] + \text{h.c.})) = \sum_{a, b \in \{0, 1\}} (\sqrt{\delta})^{a\oplus b} \tilde{R}^{(\delta)}_{a, b}(L)\otimes \ket{a}\!\bra{b}.
    \end{align}
    From this and the definition of $\hat{U}$ from \cref{eq:sc_collision_model}, it then follows that
    \begin{align}
    \bra{\mathsf{ex}(\bm a)}\hat{U}\ket{\mathsf{ex}(\bm b)} = (\sqrt{\delta})^{\abs{\bm a \triangle \bm b}} {R}^{(\delta)}_{\bm a, \bm b}.
    \end{align}
    Finally, writing $\Pi^{\leq n_c} \hat{U}\Pi^{\leq n_c} = \sum_{\bm a, \bm b \in B^{\leq n_c}} \ket{\mathsf{ex}(\bm a)}\!\bra{\mathsf{ex}(\bm b)} \bra{\mathsf{ex}(\bm a)}\hat{U}\ket{\mathsf{ex}(\bm b)}$, we obtain the lemma statement.
    \end{proof}
    Given this structure of $\Pi^{\leq n_c} \hat{U}\Pi^{\leq n_c}$, our strategy for designing a circuit implementing $\Pi^{\leq n_c}\hat{U}\Pi^{\leq n_c}$ will be based on the following lemma.
    \begin{lemma}\label{lemma:basic_impl_lemma}
    Suppose $O$ is an operator on the Hilbert space $\mathcal{H}_1\otimes \mathcal{H}_2$ such that
    \begin{align}
    O = \sum_{x, y}f(x, y) O_{x, y} \otimes \ket{x}\!\bra{y} ,
    \end{align}
    where $f(x, y) > 0$ and $f(x,y) = f(y, x)$. Suppose $W$ is an operator on $\mathcal{H}_1\otimes \mathcal{H}_2^{\otimes 2}  $ defined by
    \begin{align}
    W = \sum_{x, y} O_{y, x}\otimes \ket{x, y}\!\bra{x,y}  ,
    \end{align}
    and $U_\textnormal{in}$ is a unitary $\mathcal{H}_2^{\otimes 2} \otimes \mathcal{H}_A$ satisfying
    \begin{align}
U_\textnormal{in}\ket{x}_{\mathcal{H}_2}\ket{0}_{\mathcal{H}_2}\ket{0}_{\mathcal{H}_A} = \frac{1}{\sqrt{Z}}\sum_{y} \sqrt{g(x, y)} \ket{x}_{\mathcal{H}_2} \ket{y}_{\mathcal{H}_2}\ket{0}_{\mathcal{H}_A} + \ket{\perp_x}
    \end{align}
    where $g(x, y) g(y, x) = f^2(x, y)$ and $\ket{\perp_x}$ is a unnormalized state satisfying $(I_{\mathcal{H}_2}\otimes I_{\mathcal{H}_2}\otimes \ket{0}_{\mathcal{H}_A}\!\bra{0}) \ket{\perp_x} = 0$ and $Z > \max_x \sum_y g(x, y)$. Then
    \begin{align}
        \frac{1}{Z}O = \vcenter{\hbox{\includegraphics[scale=0.5]{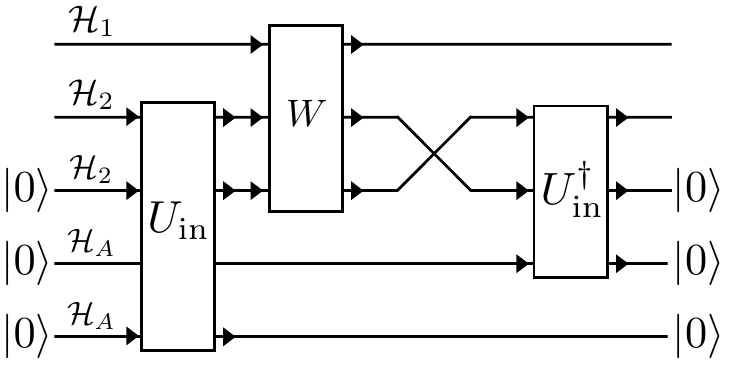}}}.
    \end{align}
    \end{lemma}
    \begin{proof}
    Note that from the definition of $U_\text{in}$ it follows that for all $\ket{x}, \ket{y}, \ket{x'}\in \mathcal{H}_2$,
    \begin{align}
    \bra{x'}_{\mathcal{H}_2} \bra{0}_{\mathcal{H}_2}\bra{0}_{\mathcal{H}_A}U_\text{in}^\dagger \ket{x}_{\mathcal{H}_2}\ket{y}_{\mathcal{H}_2}\ket{0}_{\mathcal{H}_A} &= \big(\bra{x}_{\mathcal{H}_2}\bra{y}_{\mathcal{H}_2}\bra{0}_{\mathcal{H}_A} U_\text{in} \ket{x'}_{\mathcal{H}_2} \ket{0}_{\mathcal{H}_2}\ket{0}_{\mathcal{H}_A}\big)^* \nonumber \\
&=\delta_{x, x'} \frac{\sqrt{g(x, y)}}{\sqrt{Z}}.
    \end{align}
    Therefore, we obtain that
    \begin{align}
    U_\text{in}^\dagger \ket{x}_{\mathcal{H}_2}\ket{y}_{\mathcal{H}_2} \ket{0}_{\mathcal{H}_A} =  \frac{\sqrt{g(x, y)}}{\sqrt{Z}}\ket{x}_{\mathcal{H}_2}\ket{0}_{\mathcal{H}_2}\ket{0}_{\mathcal{H}_A} + \ket{\perp'_{x, y}},
    \end{align}
    where $(I_{\mathcal{H}_2} \otimes \ket{0}_{\mathcal{H}_2}\!\bra{0} \otimes \ket{0}_{\mathcal{H}_A}\!\bra{0})\ket{\perp'_{x, y}}=0$.
    Next, using this together with the definitions of $W$ and $U_\text{in}$, we obtain
    \begin{align}
   \vcenter{\hbox{\includegraphics[scale=0.5]{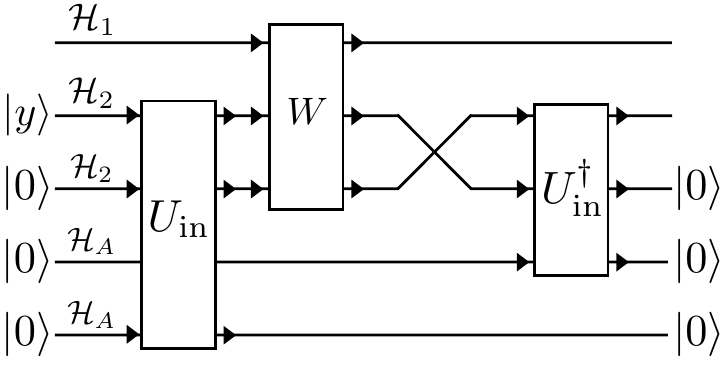}}} &= \sum_{x}\bigg(\frac{g(y, x)}{Z}\bigg)^{1/2}\vcenter{\hbox{\includegraphics[scale=0.5]{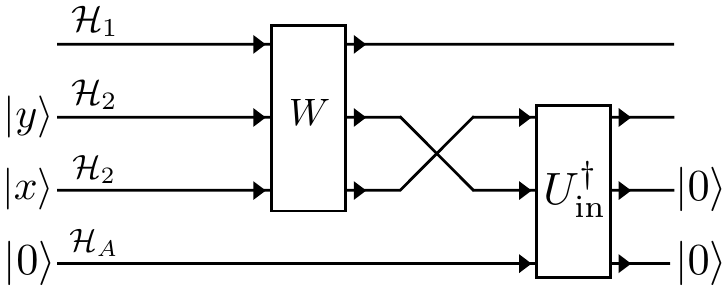}}} \nonumber \\
&=\sum_{x}\bigg(\frac{g(y, x)}{Z}\bigg)^{1/2}\vcenter{\hbox{\includegraphics[scale=0.5]{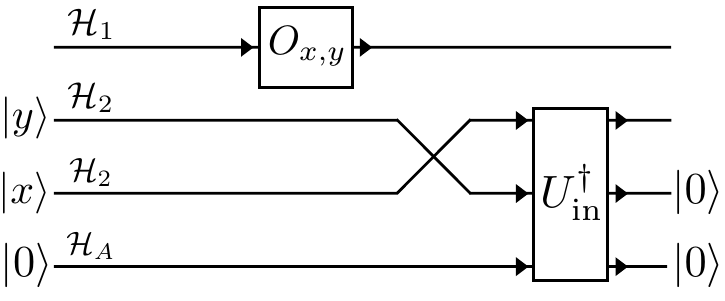}}} \nonumber \\
&=\sum_{x}\bigg(\frac{g(y, x)g(x, y)}{Z^2}\bigg)^{1/2}\vcenter{\hbox{\includegraphics[scale=0.5]{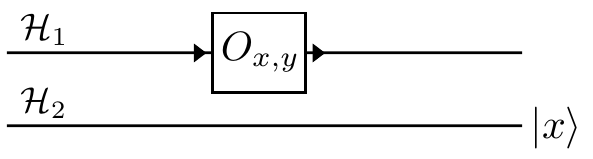}}} \nonumber \\
&=\frac{1}{Z} \sum_x f(x, y)\bigg(\vcenter{\hbox{\includegraphics[scale=0.5]{figures/figure_last_step.pdf}}}\bigg).
    \end{align}
    This establishes the lemma.
    \end{proof}

Next, we show how to implement the operator $W$ corresponding to $\Pi^{\leq n_c}\hat{U}\Pi^{\leq n_c}$. Using the expression for $\Pi^{\leq n_c}\hat{U}\Pi^{\leq n_c}$ from \cref{lemma:expression_matrix_elements}, it follows that $W$ would be chosen as
\begin{align}\label{eq:target_control_W}
W = \sum_{\bm{a}, \bm{b} \in {B}^{n_{c}}} \ket{\mathsf{ex}(\bm{b}), \mathsf{ex}(\bm{a})}\!\bra{\mathsf{ex}(\bm{b}), \mathsf{ex}(\bm{a})} \otimes{R}^{(\delta)}_{\bm{a}, \bm{b}},
\end{align}
where $R_{\bm a, \bm b}^{(\delta)}$ has also been defined in \cref{lemma:expression_matrix_elements}. However, note that $R_{\bm a, \bm b}^{(\delta)}$ still involves a product of $O(KT_0) = O(Kt_0/\delta)$ operators which if directly implemented would yield a polynomial scaling with $1/\delta$ instead of the desired logarithmic scaling $1/\delta$. However, since ${\bm a}, {\bm b}$ each contain at most $n_c \times p$ excitations, we can combine the product of $e^{-i\delta h_j}\tilde{R}^{(\delta)}_{0,0}(L_j)$ in \cref{eq:tilde_R_full} into exponential of the operator
\begin{align}\label{eq:Q}
Q = -iH - D \quad\text{where}\quad H = \sum_{j = 1}^K h_j,D =\frac{1}{2}\sum_j L_j^\dagger L_j.
\end{align}
This is accomplished in the \cref{lemma:combining_eff_hamil} below. To state and prove this lemma, it is convenient to map the excitation location representation of $\bm a, \bm b \in B^{\leq n_c}$ to another representation which keeps track of lengths of the time-intervals between the excitations.

\begin{definition}\label{def:time_diff_list}
 Define $\mathsf{TimeDiff}(\bm a, \bm b)$ and $\mathsf{ExType}(\bm a, \bm b)$ as follows: First define $\bm c$ by taking
\begin{align}
\bm c = \bigg(\bigcup_{i = 1}^p \mathsf{ex}(\bm a(i))\bigg) \cup \bigg(\bigcup_{i = 1}^p \mathsf{ex}(\bm b(i))\bigg).
\end{align}
Then, if the resulting list $\bm c$ is represented as $\bm c = \{(j_1, \tau_1), (j_2,\tau_2) \dots \}$, construct $\{(j'_1, {\tau}'_1), (j_2', {\tau}'_2) \dots \}$ by sorting $\bm c$ to ensure $\tilde{\tau}'_1 \leq \tilde{\tau}'_2 \dots $. Then, 
\begin{align}
&\mathsf{TimeDiff}(\bm a, \bm b) = \{(\tau'_1 - 1)\delta, (\tau'_2 - \tau'_1 - 1)\delta, (\tau'_3 - \tau'_2 - 1)\delta \dots \}, \\
&\mathsf{ExType}(\bm a, \bm b) = \{(\bm a(j_1', \tau_1'), \bm b(j_1', \tau_1'), j_1')), (\bm a(j_2', \tau_2'), \bm b(j_2', \tau_2'), j_2')) \dots \}
\end{align}
\end{definition}
    
\begin{lemma}\label{lemma:combining_eff_hamil}
For $\bm a, \bm b \in B^{\leq n_c}$, suppose $\mathsf{TimeDiff}(\bm a, \bm b) = \{s_0, s_1, s_2 \dots s_{q}\}$ and $\mathsf{ExType}(\bm a, \bm b) = \{(a_1, b_1, j_1), (a_2, b_2, j_2) \dots (a_q, b_q, j_q)\}$ and define
\begin{align}
\tilde{R}^{(\delta)}_{\bm a, \bm b} = \bigg(\prod_{k = q}^{1} \exp(s_k Q) R_{a_k, b_k}(L_{j_k}) \bigg) \exp(s_0 Q).
\end{align}
Furthermore, define $\tilde{W}$ by
\begin{align}
\tilde{W} = \sum_{\bm a, \bm b \in B^{\leq n_c}} \ket{\mathsf{ex}(b), \mathsf{ex}(a)}\!\bra{\mathsf{ex}(b), \mathsf{ex}(a)} \otimes \tilde{R}^{(\delta)}_{\bm a, \bm b},
\end{align}
with $Q$ being as in \cref{eq:Q}. Then,
\begin{align}
\norm{\tilde{W} - W} \leq 3K^2 t_0 \delta.
\end{align}
\end{lemma}
\begin{proof}
From the definition of $\tilde{R}^{(\delta)}_{\bm a,\bm b}$, it follows that
\begin{align}\label{eq:R_rewritten}
{R}^{(\delta)}_{\bm a,\bm b} = \bigg[\prod_{k = q}^1 \bigg(\prod_{j = 1}^K \exp(-i\delta h_j)R_{0, 0}(L_j)\bigg)^{s_k/\delta} R_{a_k, b_k}(L_k)\bigg]\bigg(\prod_{j = 1}^K \exp(-i\delta h_j)R_{0, 0}(L_j)\bigg)^{s_0/\delta}.
\end{align}
It will be useful to note that, for $s\geq 0$, $\norm{\exp(sQ)} \leq \norm{\exp(s (Q + Q^\dagger)/2)} = \norm{\exp(-s D)} \leq 1$. Additionally, since $h_j, L_j$ satisfy $\norm{h_j}, \norm{L_j} \leq 1$, it follows from the Taylor expansion of $\cos \sqrt{x}$ and the BCH formula that
\begin{align}
\left \Vert {\exp(-i\delta h_j) \cos (\delta L_j^\dagger L_j)^{1/2} - \exp\bigg[-\delta\bigg(ih_j + \frac{L_j^\dagger L_j}{2}\bigg)\bigg]}\right \Vert \leq \delta^2.
\end{align}
Therefore, by telescoping we obtain 
\begin{align}\label{eq:taylor_exp_cos_teles}
\left \Vert {\prod_{j = 1}^K \exp(-i\delta h_j) \cos (\delta L_j^\dagger L_j)^{1/2} -\prod_{j = 1}^K\exp\bigg[-\delta\bigg(ih_j + \frac{L_j^\dagger L_j}{2}\bigg)\bigg]}\right \Vert \leq K\delta^2.
\end{align}
Furthermore, from a first-order Trotter analysis, it follows that
\begin{align}\label{eq:taylor_Q_trotter}
\left \Vert \exp(\delta Q) - \prod_{j = 1}^K \exp\bigg[-\delta\bigg(ih_j + \frac{L_j^\dagger L_j}{2}\bigg)\bigg]\right  \Vert \leq \frac{\delta^2}{2}\sum_{j, j' = 1}^K \left \Vert\bigg[ih_j + \frac{L_j^\dagger L_j}{2}, ih_{j'} + \frac{L_{j'}^\dagger L_{j'}}{2}\bigg]\right \Vert \leq 2K^2 \delta^2.
\end{align}
Recalling that $R_{0, 0}^{(\delta)}(L_j) = \cos (\delta L_j^\dagger L_j)^{1/2}$, using \cref{eq:taylor_exp_cos_teles} and \cref{eq:taylor_Q_trotter} together with the triangle inequality
\begin{align}
    \left \Vert \exp(\delta Q)-\prod_{j = 1}^K \exp(-i\delta h_j) R_{0, 0}^{(\delta)}(L_j) \right \Vert 2K^2 \delta^2 + K\delta^2 \leq 3K^2 \delta^2.
\end{align} 
Finally, using this together with \cref{eq:R_rewritten} and telescoping, we obtain
\begin{align}
\norm{\tilde{R}^{(\delta)}_{\bm a, \bm b} - R^{(\delta)}_{\bm a, \bm b}} \leq 3K^2 \delta^2 \sum_{k = 0}^q \bigg(\frac{s_k}{\delta}\bigg) \leq 3K^2 t_0 \delta.
\end{align}
Finally, noting that $\norm{W - \tilde{W}} = \max_{\bm a, \bm b \in B^{\leq n_c}}\norm{{R}^{(\delta)}_{\bm a, \bm b} - \tilde{R}^{(\delta)}_{\bm a, \bm b} }$, we obtain the lemma statement.
\end{proof}
   
Next, we provide a circuit implementing $\tilde{W}$ and estimate the number of gates required. Since $\tilde{W}$ is not necessarily a unitary operator, it cannot exactly be implemented as a circuit and we will instead implement its block-encoding.
We begin by recalling the concept of block encoding of a non-unitary matrix~\cite{Gilyen2019}
\begin{definition}[\cite{Gilyen2019}]
    Suppose $M$ is an operator on $n$ qubits, then a unitary $U_M$ acting on $n + n_a$ qubits is a $(\alpha, \epsilon)-$block-encoding of $A$ using $n_a$ ancillary qubits if 
\begin{equation}
\norm{\alpha \bra{0}_a U_M  \ket{0}_a - M} \leq \epsilon,
\end{equation}
where $\ket{0}_a$ is the $0$-state on the $n_a$ ancillary qubits.
\end{definition}
Note that if $M$ has a $(\alpha, \epsilon)-$block-encoding, then $\norm{M} \leq \alpha + \epsilon $. Thus, $\alpha$ can be roughly interpreted as an estimate of the $\norm{M}$, and thus $M$ has to be re-scaled by $\alpha$ in order to allow it to be embedded into a unitary.

\begin{lemma}[Block encoding of $\tilde{W}$]\label{lemma:controlled_coefficients}
For $\tilde{W}$ defined in \cref{lemma:combining_eff_hamil}, there is a unitary $U_{\tilde{W}}$ which is a $(\beta, \epsilon)$-block encoding of $\tilde{W}$ where $\beta \leq O(1)$. Furthermore, $U_{\tilde{W}}$ can be implemented with
\begin{equation}
{O}\!\left(\operatorname{poly}\!\left(n_c,K,T_0 \delta,\log(1/\epsilon), \log(1/\delta)\right)\right).
\end{equation}
gates and with the same order of additional ancillas.
\end{lemma}
\begin{proof}
We begin by constructing a pre-processing unitary $U_\text{PRE}$ that performs the following classical operation using some ancillary register
\begin{align}\label{eq:preprocessing_unitary}
    U_\text{PRE}\ket{\mathsf{ex}(b), \mathsf{ex}(\bm a)}\ket{0}_a =\ket{\mathsf{TimeDiff}(\bm a, \bm b), \mathsf{ExType}(\bm a, \bm b)}\ket{0}_a, 
\end{align}
Recall that $\mathsf{TimeDiff}(\bm a, \bm b), \mathsf{ExType}(\bm a, \bm b)$ are constructed by first constructing the list $\bm c$.
\begin{align}\label{eq:time_order}
    \bm c = \bigg(\bigcup_{i = 1}^p \text{ex}(\bm a(i))\bigg) \cup \bigg(\bigcup_{i = 1}^p \text{ex}(\bm b(i)\bigg).
\end{align}
Suppose $\bm{c} =\{(j_1, \tau_1), (j_2, \tau_2) \dots \} $, then it is sorted to obtain a list $\{(j_1', \tau_1'), (j_2', \tau_2') \dots \}$ where $\tau_1' \leq \tau_2' \leq \tau_3' \dots $\footnote{In case two time indices $\tau_i, \tau_j$ are equal, we will assume the convention that $\mathsf{TimeOrder}$ leaves their positions relative to each other unchanged.}. Then, $\mathsf{TimeDiff}(\bm a, \bm b) = \{(\tau_1 - 1)\delta, (\tau_2 - \tau_1 - 1)\delta \dots \}$ and \\$\mathsf{ExType}(\bm a, \bm b) = \{(\bm a(j_1, \tau_1), \bm b(j_1, \tau_1), j_1), (\bm a(j_2, \tau_2), \bm b(j_2, \tau_2), j_2) \dots \}$. At this point, we remark on how $\mathsf{TimeDiff}(\bm a, \bm b), \mathsf{ExType}(\bm a, \bm b)$ are stored. Note that $\mathsf{TimeDiff}(\bm a, \bm b)$ can have a maximum of $n_c p + 1$ elements and $\mathsf{ExType}(\bm a, \bm b) $ can have $ n_c p$ elements. However, the exact number of elements would depend on how many excitations are there in $\bm a$ and $\bm b$: While constructing the block encodings below, it would be convenient to ensure $\mathsf{TimeDiff}(\bm a, \bm b)$ and $\mathsf{ExType}(\bm a, \bm b) $ have a length independent of $\bm a$ and $\bm b$. To ensure this, we will simply pad $\mathsf{TimeDiff}(\bm a, \bm b)$ with $0$ till its length is $n_c p + 1$ and pad $\mathsf{ExType}(\bm a, \bm b)$ with $(0, 0, 0)$ till its length is $n_c p$. Finally, if $\delta = 2^{-r}$,\footnote{Since all of our analysis only requires an upper bound on $\delta$ to be satisfied, we can always ensure that this is the case. Specifically, if we are given $\delta$ not of this form, then we can choose $r = \lfloor\log_2(1/\delta) \rfloor = O(\log (1/\delta))$} then we will store $\mathsf{TimeDiff}(\bm a, \bm b)$ with a block of $(\lfloor \log_2 T_0] + r)$ bits per element and $\mathsf{ExType}(\bm a, \bm b)$ with a block of $2 + \lfloor\log_2( K + 1 )\rfloor $ bits per element.

Next, observe that $U_\text{PRE}$ can be implemented with $O(\poly(n_c, K, \log KT_0))$ gates. To see this, note that classical operation performed by $U_\text{PRE}$ has the following properties:
\begin{enumerate}
    \item[(1)] We can (irreversibly) map $\mathsf{ex}(\bm a), \mathsf{ex}(\bm b) \to  \mathsf{TimeDiff}(\bm a, \bm b), \mathsf{ExType}(\bm a, \bm b)$ in $O(\poly(n_c, \log KT_0))$ time since the required sorting operation can be done, using a standard sorting algorithm, in time scaling polynomially in the length of the input list (which is at most $O(n_c)$) as well as the number of bits used to represent each list-member (which is at most $O(\log KT_0)$). Once the sorted list $\{(j_1', \tau_1'), (j_2', \tau_2') \dots\}$ is constructed, we can iterate through this list, compute $(\tau_{i + 1}' - \tau_i' - 1)\delta$ and append it to $\Delta \bm \tau$. Furthermore, during this loop for each $(j_i', \tau_i')$, we then loop through the lists $\text{ex}(\bm a_1), \text{ex}(\bm a_2) \dots $ and set $\bm a(j_i', \tau_i')$ to be 1 if $(j_i', \tau_i')$ is found else 0. A similar loop can be used to compute $\bm b(j_i, \tau_i)$. Since all of these lists are at most of size $O(n_c)$ and all the arithmetic is over $O(\log KT_0) + O(\log (1/\delta))$ bits, this can be accomplished in $\poly(n_c, \log KT_0, \log(1/\delta))$ time.
    \item[(2)] We can also irreversibly map $\mathsf{TimeDiff}(\bm a, \bm b), \mathsf{ExType}(\bm a, \bm b)$ to $\mathsf{ex}(\bm b), \mathsf{ex}(\bm a)$ in \\$O(\poly(n_c, \log KT_0, \log(1/\delta)))$ time. A loop through the list $\mathsf{TimeDiff}(\bm a, \bm b)$ can allow us to compute $\tau_1', \tau_2', \tau_3' \dots$. Then, iterate through the list $\mathsf{ExType}(\bm a, \bm b)$: If $\bm a(j_i', \tau_i') = 1$, then find $k$ such that $j_i' \in \mathcal{J}_k$ and then append $(j_i', \tau_i')$ to $\text{ex}(\bm a_k)$. This constructs $\text{ex}(a)$. Follow a similar procedure to construct $\text{ex}(\bm b)$. All the lists are of size at most $\poly(n_c, K)$ (note that the lists $\mathcal{J}_i$ are also of size at most $K$) and all the arithemtic is over $O(\log KT_0)+ O(\log(1/\delta))$ bits, this can aso be accomplished in $O(\poly(n_c, K, \log KT_0, \log(1/\delta)))$ time.
\end{enumerate}
Finally, note the standard result that if we have irreversible poly-time implementation of a classical invertible function $f:\{0,1\}^n \to \{0, 1\}^n$ as well as its inverse, then there is also a poly-time reversible implementation of $f$ where any ancillary qubits used are finally in $\ket{0}$\footnote{This can be seen by observing that for any classical function $f: \{0, 1\}^n \to \{0, 1\}^n$ implementable with $\poly(n)$ gates, we can implement  unitaries $U_f$ and $U_{f^{-1}}$ with $\poly(n)$ gates such that $U_f\ket{x}\ket{y} = \ket{x}\ket{y\oplus f(x)}$ and $U_{f^{-1}}\ket{x}\ket{y} = \ket{x}\ket{y\oplus f^{-1}(x)}$. Now, 
\[
U_{f^{-1}}\cdot \text{SWAP}\cdot U_f \ket{x}\ket{0} = U_{f^{-1}}\cdot \text{SWAP}\ket{x}\ket{f(x)} = U_{f^{-1}}\ket{f(x)}\ket{x} = \ket{f(x)}\ket{x\oplus f^{-1}(f(x))} = \ket{f(x)}\ket{0}
\]}. This yields a reversible implementation, and thus a unitary circuit, for $U_\text{PRE}$.

Next, we will need the block-encodings of $H = \sum_{j = 1}^K h_j$, $D = \sum_{j = 1}^K L_j^\dagger L_j/ 2$ and the operator $R$ defined by
\[
R = \sum_{a, b \in \{0, 1\}, j\in[0:M]}|a, b, j\rangle \!\bra{a, b, j} \otimes {R}_{a, b}^{(\delta)}(L_j),
\]
where $R^{(\delta)}_{a, b}(L_j)$ are local operators acting on the support of $L_j$ as defined in \cref{lemma:expand_exp} and if $j = 0$, then we set $\tilde{R}_{a, b}^{(\delta)}(L_j) = I$. Since $L_j^\dagger L_j$ and $h_j$ are both $O(1)$ local operators with $\norm{L_j^\dagger L_j}, \norm{h_j} \leq 1$, using the standard linear-combination of unitaries (LCU) toolbox \cite{Gilyen2019}, we can construct a $(\alpha_D, 0)$-block-encodings $U_D$ and an $(\alpha_H, 0)$-block-encoding $U_H$ of $H$, with $\alpha_D, \alpha_H \leq O(K)$, and with $O(K)$ gates. Finally, an exact $(1, 0)$-block encoding $R$, $U_R$, can be implemented as follows: Since the operators ${R}^{(\delta)}_{a, b}(L_j)$ are $O(1)$ local with $\norm{{R}^{(\delta)}_{a, b}(L_j)} \leq 1$, we can construct a $(1, 0)-$block encoding $U_{{R}^{(\delta)}_{a, b}(L_j)}$ of these operators with $O(1)$ number of gates. Then, implementing the controlled unitary $\sum_{a, b, j} |a, b, j\rangle \!\bra{a, b, j} \otimes U_{{R}^{(\delta)}_{a, b}(L_j)}$ yields a $(1, 0)$-block-encoding of $R$. Since $j \in [1:K]$, this controlled operation can be implemented using $O(K)$ gates. 

In the appendix lemma \ref{lemma:LCHS_interleaved}, we consider the problem of block encoding the operator
\begin{align}
    F = \sum_{s_0, s_1 \dots s_q =0}^{t_0}\ket{s_0, s_1 \dots s_q}\!\bra{s_0, s_1 \dots s_q} \otimes \exp(Qs_q)R_q\exp(Qs_{q-1}) R_{q-1} \dots \exp(Qs_1)R_1\exp(Qs_0),
\end{align}
where the sum runs over time-indices $s_j$ that are integer multiples of $\delta$, $Q = -(iH + D)$ and $R_1, R_2 \dots R_q$ are some other operators. This operator performs evolution with the non-Hermitian operator $Q$ for times stored in the control register and with intermediate applications of operators $R_1,R_2 \dots R_q$. By building on the framework in \cite{An2026Quantum} for non-Hermitian time evolution, we show in appendix \ref{app:imag_time} ( \cref{lemma:LCHS_interleaved}) that a $(\beta, \epsilon)$-block-encoding, where $1\leq \beta \leq O(1)$, can be constructed for $F$ by $O(\poly(r, q, \alpha_H, \alpha_D, t, \log(1/\epsilon)))$ calls to $U_H, U_D$ and $U_{R_j}$ as well as the same number of additional single and two-qubit gates. Applying this lemma with $q =  pn_c$, $t = t_0$, $r = \lfloor\log_2(1/\delta)\rfloor$, $H, D$ as defined above and with $R_j$ being the operator $R$ applied on the $j^\text{th}$ element of the register storing $\mathsf{ExType}(\bm a, \bm b)$, we obtain a block encoding for the operator
\begin{align}
    W' =\sum_{\bm a, \bm b \in B^{\leq n_c}} \ket{\mathsf{TimeDiff}(\bm a, \bm b), \mathsf{ExType}(\bm a, \bm b)}\bra{\mathsf{TimeDiff}(\bm a, \bm b), \mathsf{ExType}(\bm a, \bm b)} \otimes \tilde{R}^{(\delta)}_{\bm a, \bm b},
\end{align}
where $\tilde{R}^{(\delta)}_{\bm a, \bm b}$ is as defined in \cref{lemma:combining_eff_hamil}. Finally noting that $\tilde{W} = U_\text{PRE}^\dagger W' U_\text{PRE}$, we obtain the lemma.
\end{proof}
This completes the implementation of $\tilde{W}$. In order to apply \cref{lemma:basic_impl_lemma}, we need to additionally construct $U_\text{in}$. Note from \cref{lemma:basic_impl_lemma} that there are multiple different choices for $U_\textnormal{in}$ corresponding to different choices of the function $g$ introduced in the lemma.  The concrete choice that we will make for $U_\text{in}$ is below:
\begin{subequations}\label{eq:input_unitary_target}
    \begin{align}
    U_\text{in}\ket{\mathsf{ex}(\bm a)}\ket{0}\ket{0}_a = \frac{1}{\sqrt{Z}} \sum_{\bm b\in B^{\leq n_c}}  \sqrt{G({\bm a, \bm b})} \ket{\mathsf{ex}(\bm a)}\ket{\mathsf{ex}(\bm b)}\ket{0}_a + \ket{\perp_{\bm a}},
\end{align}
where 
\begin{align}
    G({\bm a, \bm b} )= (K t_0)^{\abs{\mathsf{ex}(\bm a)\setminus \mathsf{ex}(\bm b)} - \abs{\mathsf{ex}(\bm b)\setminus \mathsf{ex}(\bm a)}} {\delta}^{\abs{\mathsf{ex}(\bm b)\setminus \mathsf{ex}(\bm a)}},
\end{align}
\end{subequations}
and $\ket{0}_a\bra{0} \perp_{\bm a}\rangle = 0$. Note that
\begin{align}
     G({\bm a, \bm b} ) G(\bm b, \bm a) = {\delta}^{\abs{\mathsf{ex}(\bm b)\setminus \mathsf{ex}(\bm a)} + \abs{\mathsf{ex}(\bm a)\setminus \mathsf{ex}(\bm b)}} = {\delta}^{\abs{\bm a \triangle \bm b}},
\end{align}
which is exactly the square of the coefficient appearing in $\Pi^{\leq n_c}\hat{U}\Pi^{\leq n_c}$. Therefore, this choice of $G(\bm a, \bm b)$ satisfies the requirements of \cref{lemma:basic_impl_lemma}.

The next lemma addresses the question of implementing $U_\text{in}$. This has two aspects: \emph{First}, we need to show that $Z$ in this equation is does not become very large otherwise, as per \cref{lemma:basic_impl_lemma}, the target operator $\Pi^{\leq n_c}O\Pi^{\leq n_c}$ would only be implemented with a large post-selection overhead. \emph{Second}, we need to show that for an appropriately chosen $Z$, $U_\text{in}$ can be implemented with a number of gates scaling as $O(\poly(n_c, t_0, K, \log(1/\delta)))$ and not polynomially with $1/\delta$. We address the choice of $Z$ in the lemma below and show that the specific choice of $U_\text{in}$ made above results in us being able to choose $Z = O(1)$. We relegate the proof of the full implementation of $U_\text{in}$ to the appendix \ref{app:inp_unitary}.
\begin{lemma}\label{lemma:inp_unitary_circuit}
    If $t_0 Kn_c p \leq 1$, then 
    \begin{enumerate}
        \item[(a)]$U_\textnormal{in}$ as defined Eq.~\eqref{eq:input_unitary_target} is well-defined with $Z = 8^p$.
        \item[(b)]$U_\textnormal{in}$ can be implemented exactly with $O(\textnormal{poly}(n_c, t_0, K, \log(1/\delta)))$ continuous single and two-qubit gates.
    \end{enumerate}
\end{lemma}
\begin{proof}[Proof of (a)] $Z$ must be chosen such that
\begin{align}
    \forall \bm{a}\in B^{\leq n_c}: \sum_{\bm b\in {B}^{\leq n_c}} G({\bm a, \bm b}) \leq Z.
\end{align}
Next, we consider upper-bounding the summation. Note that the summand depends entirely on $\abs{\mathsf{ex}(\bm a) \setminus \mathsf{ex}(\bm b)}$ and $\abs{\mathsf{ex}(\bm b) \setminus \mathsf{ex}(\bm a)}$ i.e., on the number of excitations in $\bm a$ but not in $\bm b$ and the number of excitations in $\bm b$ and not in $\bm a$. Re-index the summation with $q_a \in [0:\abs{\mathsf{ex}(\bm a)}]$ which tracks $\abs{\mathsf{ex}(\bm a) \setminus \mathsf{ex}(\bm b)}$ and $q_b \leq  n_c p - q_a$ which tracks $\abs{\mathsf{ex}(\bm b) \setminus \mathsf{ex}(\bm a)}$. Then,
\begin{align}
    \sum_{\bm b\in {B}^{\leq n_c}} G(\bm a, \bm b) &\leq \sum_{q_a = 0}^{\abs{\mathsf{ex}(\bm a)}} \sum_{q_b = 0}^{n_c p -q_a } {\abs{\mathsf{ex}(\bm a)} \choose q_a} {KT_0 - \abs{\mathsf{ex}(\bm a)}\choose q_b} (Kt_0)^{q_a - q_b} \delta^{q_b}\nonumber\\
    &\leq  \sum_{q_a = 0}^{\abs{\mathsf{ex}(\bm a)}}{\abs{\mathsf{ex}(\bm a)} \choose q_a}(Kt_0)^{q_a} \sum_{q_b = 0}^{KT_0}  {KT_0 \choose q_b} \bigg(\frac{\delta}{Kt_0}\bigg)^{q_b}\nonumber \\
    &=(1 + Kt_0)^{\abs{\mathsf{ex}(\bm a)}}\bigg(1 + \frac{\delta}{Kt_0}\bigg)^{KT_0}\nonumber \\
    &\leq \exp(Kt_0 n_cp) \exp(\delta KT_0/Kt_0) \leq e^2.
\end{align}
    Thus a choice of $Z = 8$ would ensure that $U_\text{in}$ can be defined.
\end{proof}

\begin{proof}[Proof of \cref{lemma:coll_model_main}] 
\cref{lemma:controlled_coefficients} gives a $(\beta, \epsilon/2)-$block-encoding of $\tilde{W}$, $U_{\tilde{W}}$, for some $1\leq \beta \leq O(1)$ with ${O}\!\left(\operatorname{poly}\!\left(n_c,K,T_0 \delta,\log(1/\epsilon), \log(1/\delta)\right)\right)$ gates and ancillas. From lemma \ref{lemma:combining_eff_hamil} and the assumption $\delta \leq \epsilon/6K^2 t_0$, it follows that
\begin{align}
    \norm{\beta \langle 0|_a U_{\tilde{W}}\ket{0}_a - W} \leq \norm{\beta \langle 0|_a U_{\tilde{W}}\ket{0}_a - \tilde{W}} + \norm{W - \tilde{W}} \leq \frac{\epsilon}{2} + 3K^2 t_0 \delta \leq \epsilon,
\end{align}
and thus $U_{\tilde{W}}$ is a $(\beta, \epsilon)$-block-encoding of $W$. Using $U_\text{in}$ as defined above together with \cref{lemma:basic_impl_lemma} implies that there is a unitary $\tilde{V}$ which is a $(8\beta,\epsilon)$-block-encoding of $\Pi^{\leq n_c} \hat{U}\Pi^{\leq n_c}$.

Next, we turn to implementing $\Pi^{\leq n_c^\text{out}} \hat{U}\Pi^{\leq n_c^\text{in}}$. We are given in the lemma $\norm{\Pi^{\leq n_c^\text{out}} \hat{U}\Pi^{\leq n_c^\text{in}} - \hat{U}\Pi^{\leq n_c^\text{in}}} \leq \epsilon_c$. Since $\hat{U}$ is a unitary, we obtain that 
\begin{align}
    \norm{(\Pi^{\leq n_c^\text{out}}\hat{U} \Pi^{\leq n_c^\text{in}})^{\dagger}(\Pi^{\leq n_c^\text{out}}\hat{U} \Pi^{\leq n_c^\text{in}}) - \Pi^{\leq n_c^\text{in}}} =  \norm{\Pi^{\leq n_c^\text{in}}\hat{U}^\dagger (I -\Pi^{\leq n_c^\text{out}})\hat{U}\Pi^{\leq n_c^\text{in}}} \leq \epsilon_c.
\end{align}
Consequently, the singular-values of $\Pi^{\leq n_c^\text{out}} \hat{U}\Pi^{\leq n_c^{\text{in}}}$, with the input space being $\text{Im}(\Pi^{\leq n_c^{\text{in}}})$, are all in the interval $[(1-\epsilon_c)^{1/2}, (1 + \epsilon_c)^{1/2}]$. Thus, it follows that if we consider the singular-value decomposition $\Pi^{\leq n_c^{\text{in}}} \hat{U}\Pi^{\leq n_c^{\text{out}}} = V_L \Sigma V_R^\dagger$, then $\norm{\Pi^{\leq n_c^\text{out}} \hat{U}\Pi^{\leq n_c^\text{in}} - V_L V_R^\dagger} \leq \norm{\Sigma - I} \leq \epsilon_c$. Now, if a perfect block $(\beta, 0)$-block-encoding of $\Pi^{\leq n_c} \hat{U} \Pi^{\leq n_c}$ was available and since $\Pi^{\leq n_c^{\text{out}}} \hat{U} \Pi^{\leq n_c^{\text{in}}} = \Pi^{\leq n_c^{\text{out}}}(\Pi^{\leq n_c} \hat{U} \Pi^{\leq n_c})\Pi^{\leq n_c^{\text{in}}}$, then \citep[Theorem 26]{Gilyen2019} shows that with $O(\beta \log(1/\Delta))$ calls to this block-encoding, $O(\beta \log(1/\Delta))$ calls to the reflection unitaries $I - 2\Pi^{\leq n_c^{\text{in}}}$ and $I - 2\Pi^{\leq n_c^{\text{out}}}$ and a similar order of single and two-qubit gates allows us to implement a unitary $\Delta$-close to $V_L V_R^\dagger$. We note that
\begin{enumerate}
    \item[(1)] The reflection unitaries $I - 2\Pi^{\leq n_c^{\text{in}}}$ and $I - 2\Pi^{\leq n_c^{\text{out}}}$ can be easily implemented within the space $\Pi^{\leq n_c}$ since, in the excitation-location representation (\cref{def:ex_loc}), the number of excitations is stored in the last $\lfloor \log_2 (n_c + 1) \rfloor$ qubits. The reflection unitary can be implemented with $O(\log n_c)$ gates by simply comparing the stored excitation number with $n_c^{\text{in}}/n_c^\text{out}$ and then applying a phase.
    \item[(2)] We do not have access to an exact block encoding of $\Pi^{\leq n_c}\hat{U}\Pi^{\leq n_c}$, but only an approximate one. This is accounted for in \cref{lemma:appx_U_to_U} of appendix \ref{app:appx_U_to_U} and if a $(8\beta, \epsilon)$-block-encoding of $\Pi^{\leq n_c}\hat{U}\Pi^{\leq n_c}$ is available, we obtain a unitary that is $O(\epsilon) + \epsilon_c$ close to $\Pi^{\leq n_c^\text{out}} \hat{U}\Pi^{\leq n_c^\text{in}}$ with $O(\log O(1/\epsilon))$ calls to its approximate block encoding. 
\end{enumerate}

\end{proof}

\subsection{Proof of main theorem}
We have assembled all the ingredients needed to prove \cref{th:algorithm}. To simulate $e^{\mathcal{L}t}$, we first write it as
\begin{align}
    e^{\mathcal{L}t} = (e^{\mathcal{L}t_0})^{t/t_0},
\end{align}
where $t_0$ will be chosen below. To obtain a circuit implementing $e^{\mathcal{L}t}$ to an error $\epsilon$, each of the channels $e^{\mathcal{L}t_0}$ need to be implemented to an error $\epsilon' = \epsilon (t_0/t)$.

We will now focus on implementing $e^{\mathcal{L}t_0}$: From \cref{lemma:chopping_unitaries} and \cref{lemma:coll_model_main}, we obtain that $e^{\mathcal{L}t_0}$ can be approximated to error $\epsilon'$ by a 3-stage collision model depicted in Fig. . with the truncation-length $l$ satisfying 
\begin{align}\label{eq:const_4}
    l = O(t_0)+O\bigg(\log \bigg(\frac{N}{\epsilon'}\bigg)\bigg) \leq O(t_0) + O\bigg(\log \bigg(\frac{Nt}{t_0\epsilon}\bigg)\bigg).
\end{align}
and the collision-model time step $\delta$ satisfying 
\begin{align}\label{eq:const_3}
    \delta \leq O\bigg(\frac{\epsilon'}{Nt_0}\bigg ) \leq O\bigg(\frac{\epsilon}{Nt}\bigg ).
\end{align}
The number of collision-model time step (per stage) is $T_0 = t_0/\delta$ and the collision-model acts on the $N$ system qudits and a total of $O(T_0 N)$ ancillary qubits.

Next, we truncate the ancillas to low-excitation subspaces. Index the ancillas introduced in the collision model by $(e, \alpha,\tau)$, where $e \in \{(x, x+1): x\in [1:N)\}$ is an edge, $\tau \in [1:T_0]$ is a collision-model time step and $\alpha \in \{1, 2 \dots d^4 -1\}$ indexes the different jump operators per edge. We divide the ancillas into groups of $(d^4 - 1)lT_0$ via $\mathcal{A}_k = \{(e, \alpha, \tau): e\in \{(kl + 1, kl+2), (kl+2, kl+3) \dots ((k + 1)l, (k + 1)l + 1)\}, \alpha \in (1, 2 \dots d^4-1), \tau\in[1:T_0]\}$. Denote by $\Pi_{\mathcal{A}_k}^{\leq n}$ the projector on the subspace with the ancillas in $\mathcal{A}_k$ with $\leq n$ excitations. Consider now a unitary $\hat{U}_{i,[kl + 1, (k+p)l]}$ in  stage $i \in \{1, 2,3\}$ of the collision model acting on sites in $[kl + 1, (k+p)l]$. Note that in the collision-model at hand, $p \in\{1, 2\}$. Using \cref{lemma:ex_trunc}, we obtain that this unitary satisfies the leakage bound
\begin{align}
\left\Vert{\bigg(I - \prod_{k' = k}^{k+p-1}\Pi^{\leq n_{c, i}}_{\mathcal{A}_{k'}}\bigg) \hat{U}_{i,[kl + 1, (k+p)l]}  \bigg(\prod_{k' = k}^{k+p-1}\Pi^{\leq n_{c, i-1}}_{\mathcal{A}_{k'}} \bigg)}\right \Vert    &\leq \sum_{k' = k}^{k + p - 1}\norm{(I - \Pi^{\leq n_{c, i}}_{\mathcal{A}_{k'}})\hat{U}_{i,[kl + 1, (k+p)l]}\Pi^{\leq n_{c, i}}_{\mathcal{A}_{k'}}} \nonumber \\
&\leq p e^{2lt_0} e^{n_{c, i-1} - n_{c, i}/2},
\end{align}
where $n_{c, i}$ is the excitation number cutoff introduced at the $i^\text{th}$ stage. Since the ancillas are initially not excited (i.e., $n_{c,0} =0$), this allows us to choose 
\begin{align}\label{eq:const_2}
    n_{c, 1}, n_{c, 2}, n_{c, 3} \leq O(lt_0) + O\bigg(\log\bigg(\frac{N}{\epsilon'}\bigg)\bigg) \leq O(lt_0) + O\bigg(\log\bigg(\frac{Nt}{t_0\epsilon}\bigg)\bigg),
\end{align}
and approximate each of the unitaries $\hat{U}_{i,[kl + 1, (k+p)l]}$ by their truncations  and maintain a total error $\leq O(\epsilon')$ in the full collision model. Finally, we use \cref{lemma:coll_model_main} to implement these approximate unitaries with $n_c = \max_i n_{c, i}$: to apply this lemma, we additionally need to ensure that $t_0, l$ and $n_c$ satisfy
\begin{align}\label{eq:const_1}
    n_c t_0 l \leq O(1).
\end{align}
The constraints in \cref{eq:const_4}, \cref{eq:const_3}, \cref{eq:const_2} and \cref{eq:const_1} can be satisfied with a choice of 
\begin{align}
    l \leq O\bigg(\log\bigg(\frac{Nt}{\epsilon}\bigg)\bigg), t_0 \leq O\bigg(\frac{1}{\text{polylog}(Nt/\epsilon)}\bigg), n_c\leq O\bigg(\text{polylog}\bigg(\frac{Nt}{\epsilon}\bigg)\bigg) \quad \text{and}\quad \delta \leq O\bigg(\frac{\epsilon}{Nt}\bigg).
\end{align}
Finally, as per \cref{lemma:coll_model_main}, each of the truncated unitaries can be implemented to an error $O(\epsilon'/N)$ (to ensure that the total error in one $t_0$-time-evolution is $\leq\epsilon'$) with 
\begin{align}
    O\bigg(\poly\bigg(lt_0, n_c, \log\bigg(\frac{1}{\delta}\bigg), \log\bigg(\frac{N}{\epsilon'}\bigg)\bigg)\bigg) \leq O\bigg(\text{polylog}\bigg(\frac{Nt}{\epsilon}\bigg)\bigg),
\end{align}
and the same order of ancillas. This yields the target circuit complexity and ancilla count stated in the theorem.
\subsection{Higher dimensional and time-dependent extension}
Our construction extends to lattices of any fixed spatial dimension
with fixed interaction range and bounded coordination number.
On each short time interval, we apply the HHKL decomposition successively along the spatial directions~\cite{Haah2021}. This yields $M$ stages of
forward or backward system--environment evolution,
where $M$ is the number of stages introduced in
\cref{eq:collision_model_U}. Here
$M$ depends only on the spatial dimension but remains
independent of $N$, $t$, and $\varepsilon$.
Within each stage, the evolution factorizes over
spatially disjoint blocks, and each block has linear size
$O(\log(Nt/\varepsilon))$. The excitation-number truncation and compressed implementation
therefore retain polylogarithmic overheads, preserving the
$O(Nt\,\polylog(Nt/\varepsilon))$ gate complexity and
$O(t\,\polylog(Nt/\varepsilon))$ circuit depth.

The construction also extends to piecewise differentiable
time-dependent local terms $h_e(\tau)$ and $L_{\alpha,e}(\tau)$
satisfying the piecewise-slow-variation assumptions of
Ref.~\cite{Haah2021}, provided that
\cref{eq:assumption_of_local_norm} holds uniformly in time
and their matrix elements can be evaluated coherently
at binary-encoded times. The system-environment Lieb--Robinson bound in
\cref{lemma:lr_bound} remains applicable in this
setting~\cite{TrivediRudner2024}.
\cref{lemma:LCHS_interleaved} analyzing the non-Hermitian time-evolution with intermediate jumps extends to time-dependent problems, with its unitary
components implemented using the truncated Dyson-series
algorithm of Ref.~\cite{Kieferova2019}.
Within each differentiable interval, the additional time-discretization error is controlled by bounds on the first derivatives of the time-dependent coefficients. We therefore refine the collision and Hamiltonian-simulation time grids so that no time step crosses a boundary between differentiable pieces. In the compressed implementation and the truncated Dyson-series algorithm, this incurs only polylogarithmic overhead in the required grid sizes~\cite{Kieferova2019}. Consequently, provided that the first-derivative bounds are constant or grow at most polynomially in $Nt$, the overall complexity remains
$O(Nt\polylog(Nt/\varepsilon))$
in gate count and
$O(t\polylog(Nt/\varepsilon))$
in circuit depth.

\section{Discussion}

In closed systems, pure-state phase equivalence can be formulated using either local unitary circuits or continuous local Hamiltonian evolution. Efficient Hamiltonian simulation relates these discrete- and continuous-time descriptions with only polylogarithmic overhead~\cite{Haah2021}.
An analogous question arises for mixed states, whose phase equivalence can be defined through mutual convertibility under either rapid-local-Lindbladian evolution~\cite{Coser2019} or shallow
circuits of local quantum channels~\cite{Ma2023, Rakovszky2024}. While the discrete-to-continuous direction is transparent if additional ancillas are allowed,
our algorithm establishes the continuous-to-discrete direction:
a geometrically local Lindbladian evolution for time $t$ can be approximated by a local quantum-channel circuit of depth $O(t\,\polylog(Nt/\varepsilon))$.
Therefore, our result rigorously establish the equivalence between the two descriptions of mixed-state phase.

An important open question concerns the minimum ancilla resources required for efficient simulation.
Without refreshing, our algorithm uses $O(t\polylog(Nt/\varepsilon))$ ancilla per site in total. However, the final channel admits a Stinespring dilation with only $O(1)$ ancillas per site. This raises the question of if there might  exist algorithms that use fewer ancillas (e.g., $O(\polylog(Nt/\varepsilon))$) if we simultaneously demand near-optimal $O(Nt\,\polylog(Nt/\varepsilon))$ gate complexity? A second open question concerns lower bounds for purely dissipative, geometrically local Lindbladians. A linear-in-$t$ circuit-depth lower bound can already be enforced by locality, for example through dissipative dynamics that propagates an XOR over distance $\Theta(t)$. The more challenging question is whether, at fixed diamond-norm error, such dynamics can require $\Omega(Nt)$ gates. Existing Hamiltonian lower bounds~\cite{Haah2021} and dissipative no-fast-forwarding results in sparse black-box models~\cite{ChildsLi2017} do not directly establish this extensive gate-count bound in the geometrically local, purely dissipative setting. Establishing such bounds, or identifying general improvements specified to local, purely dissipative dynamics, would clarify how irreversibility affects the computational costs of quantum simulation.

\paragraph{Note added}
During the preparation of this manuscript, two related works appeared on algorithms for lattice Lindbladian dynamics.
Wang and Ye~\cite{WangYe2026} give a query-optimal algorithm, but with worse gate complexity.
Mizuta~\cite{Mizuta2026} pursues a different construction that gives an algorithm with optimal scaling, but only in the restrictive case of sparsely located dissipation.

\section*{Acknowledgements}
The authors acknowledge the use of OpenAI's language model GPT-5 (Sol) for the proofs in \cref{sec:compressed-collision-model}. In particular, the language model was used to generate a possible compression scheme, which was checked, significantly simplified and rewritten by the authors. All the other parts of the proof were done by the authors with LLMs assisting only in editing and improving the presentation of the proofs. The authors take full responsibility for the correctness of the analysis presented in the paper.

R.T. acknowledges discussion with Jeongwan Haah, Matthew Hastings, Alexey Gorshkov and Ignacio Cirac. R.T. acknowledges support from the European Union’s Horizon Europe research and innovation program under grant agreement number 101221560 (ToNQS) and from the Munich Center for Quantum Science and Technology (MCQST), funded by the Deutsche Forschungsgemeinschaft (DFG) under Germany’s Excellence Strategy (EXC2111-390814868).

\bibliographystyle{alpha}
\bibliography{references}

\appendix

\section{Proof of the splitting identity}\label{sec:splitting-proof}
\begin{proof}[Proof of \cref{lemma:splitting_remainder}]

    We begin by writing
    \begin{align}
        \mathcal{U}_{A\cup B \cup C}(t, s) = \mathcal{U}_{A\cup B}(t, s) \mathcal{U}_C(t, s) \mathcal{W}(t, s),\label{eq:U_splitting}
    \end{align}
    where
    \begin{align}
    \mathcal{W}(t, s) = \mathcal{U}_{A\cup B}(s, t) \mathcal{U}_C(s, t) \mathcal{U}_{A\cup B \cup C}(t, s).
    \end{align}
    Note that $\mathcal{W}(t, s)$ satisfies the differential equation
    \begin{equation}
    \frac{d}{dt}\mathcal{W}(t, s) = -i \mathcal{G}(t, s) \mathcal{W}(t, s),
    \end{equation}
    where
    \begin{equation}
     \mathcal{G}(t, s) = \mathcal{U}_{A\cup B}(s, t) \mathcal{U}_C(s, t)\mathcal{H}^{(\partial)}_{BC}(t)\mathcal{U}_C(t, s)\mathcal{U}_{A\cup B}(t, s),
    \end{equation}
    with
    \begin{equation}
    \mathcal{H}^{(\partial)}_{BC}(t) = \mathcal{H}_{A\cup B \cup C}(t) - \mathcal{H}_{A\cup B}(t) - \mathcal{H}_{C}(t).
    \end{equation}
    It can also be noted that since $A, B, C$ are disjoint regions with $d(A, C) \geq 2$, $\mathcal{H}^{(\partial)}_{BC}(t) = \mathcal{H}_{B \cup C}(t) -\mathcal{H}_{B}(t) - \mathcal{H}_C(t)$.

    Consider now an ``approximate" version of $\mathcal{W}(t, s)$, $\tilde{\mathcal{W}}(t, s)$, defined by
    \begin{equation}
    \tilde{\mathcal{W}}(t, s) = \mathcal{U}_{B}(s, t) \mathcal{U}_{C}(s, t) \mathcal{U}_{B \cup C}(t, s),
    \end{equation}
    which satisfies the generating equation
    \begin{equation}
    \frac{d}{dt}\tilde{\mathcal{W}}(t, s) = -i\tilde{\mathcal{G}}(t, s) \tilde{\mathcal{W}}(t, s),
    \end{equation}
    where
    \begin{equation}
     \tilde{\mathcal{G}}(t, s) = \mathcal{U}_{B}(s, t) \mathcal{U}_C(s, t)\mathcal{H}^{(\partial)}_{BC}(t)\mathcal{U}_C(t, s)\mathcal{U}_{B}(t, s).
    \end{equation}
    It can be noted that 
    \begin{align}
\mathcal{U}_{A\cup B}(t,s )\mathcal{U}_C(t, s)\tilde{\mathcal{W}}(t, s) = \mathcal{U}_{A\cup B}(t, s) \mathcal{U}_B(s, t) \mathcal{U}_{B\cup C}(t, s),\label{eq:def_w_tilde}
\end{align}
    thus yielding the intended approximation of $\mathcal{U}_{A\cup B \cup C}(t, s)$ from the Lemma statement.
    
    We now develop an expression for $\tilde{\mathcal{W}}(t, s) - \mathcal{W}(t, s)$ by integrating their generating equations
    \begin{equation}\label{eq:diff_W}
		\begin{aligned}
			&\tilde{\mathcal{W}}(t,s) - \mathcal{W}(t, s)\\
			&\qquad= i \int^t_{s}\tilde{\mathcal{W}}(t, s)\tilde{\mathcal{W}}^\dagger(\tau, s) (\mathcal{G}(\tau, s) - \tilde{\mathcal{G}}(\tau, s)) \mathcal{W}(\tau, s) d\tau, \\
			&\qquad=i\int^t_{s}\mathcal{U}_B(s, t) \mathcal{U}_C(s, t)\mathcal{U}_{B\cup C}(t, \tau)\mathcal{U}_B(\tau, s)\big( \mathcal{U}_B(s, \tau)\mathcal{H}^{(\partial)}_{BC}(t)\mathcal{U}_B(\tau, s) -\\
			&\qquad \qquad \qquad \qquad \qquad \mathcal{U}_{A\cup B}(s, \tau)\mathcal{H}^{(\partial)}_{BC}(t)\mathcal{U}_{A\cup B}(\tau, s) \big)\mathcal{U}_{A\cup B}(s, \tau) \mathcal{U}_{A\cup B \cup C}(s, \tau) d\tau.
		\end{aligned}
    \end{equation}
    Next, we note that
    \begin{equation}\label{eq:diff_hamil}
		\begin{aligned}
			&\mathcal{U}_{B}(s, \tau) \mathcal{H}^{(\partial)}_{BC}(\tau)\mathcal{U}_{B}(\tau, s) - \mathcal{U}_{A\cup B}(s, \tau) \mathcal{H}^{(\partial)}_{BC}(\tau)\mathcal{U}_{A\cup B}(\tau, s),\\
			& \qquad = -\int^\tau_{s}\frac{\partial}{\partial \tau'}\bigg[\mathcal{U}_{A \cup B}(s, \tau') \mathcal{U}_{B}(\tau', \tau)\mathcal{H}^{(\partial)}_{BC}(\tau)\mathcal{U}_{B}(\tau, \tau')\mathcal{U}_{A \cup B}(\tau', s)\bigg]d\tau',  \\
			&\qquad = i\int^\tau_{s}\mathcal{U}_{A \cup B}(s, \tau')\big[\mathcal{H}^{(\partial)}_{AB}(\tau'), \mathcal{U}_{B}(\tau', \tau)\mathcal{H}^{(\partial)}_{BC}(\tau)\mathcal{U}_{B}(\tau, \tau')\big]\mathcal{U}_{A \cup B}(\tau', s) d\tau'.
		\end{aligned}
    \end{equation}
Combining \cref{eq:U_splitting,eq:def_w_tilde,eq:diff_W,eq:diff_hamil}, we then obtain the lemma statement.
\end{proof}

\section{Commutator estimates}\label{sec:commutator-proofs}
We begin by proving the following useful lemma.
\begin{lemma}\label{lemma:commutator_H_bound}
    For $\phi \in \mathcal{S}_K$, $\tau\geq\tau' \geq 0$, edges $e, f \in \{(x, x + 1):x\in[1:N)\}$ with $e \neq f$ and any $\mathcal{B} \subseteq \{(1, 2), (2, 3) \dots \}$,
    \begin{equation}
    \abs{\Tr ([H_{e}(\tau'), U_{\mathcal{B}}(\tau', \tau) H_{f}(\tau)U_{\mathcal{B}}(\tau, \tau')]\phi)}
    \leq C(K) \exp(\xi_0^{-1}(v(K)\smallabs{\tau - \tau'} - d(e, f))),
    \end{equation}
    where $C_{\mathrm{edge}}=2C_0$ absorbs the support factor $|Y|=2$ in \cref{lemma:lr_bound}, $C(K) = C_{\mathrm{edge}}(8K^2 + 12K + 2)$ and $v(K)$ is as defined in \cref{lemma:lr_bound}.
\end{lemma}
\begin{proof}
	For notational convenience, define $l = d(e, f)$. Recall that
    \begin{equation}
    H_{e}(t) = h_{e} + V_{e}(t) \text{ where }V_{e}(t) =  \sum_\alpha \big( L_{\alpha, e}a^\dagger_{\alpha,e}(t) + \text{h.c.}\big).
    \end{equation}
    Applying a triangle inequality, we then obtain
    \begin{equation}\label{eq:triangle_inequality}
		\begin{aligned}
			&\abs{\Tr ([H_{e}(\tau'), U_{\mathcal{B}}(\tau', \tau) H_{f}(\tau)U_{\mathcal{B}}(\tau, \tau')]\phi)}\\
			&\qquad\leq \abs{\Tr ([h_{e}, U_{\mathcal{B}}(\tau', \tau) h_{f}U_{\mathcal{B}}(\tau, \tau')]\phi)} +  \abs{\Tr ([h_{e}, U_{\mathcal{B}}(\tau', \tau) V_{f}(\tau)U_{\mathcal{B}}(\tau, \tau')]\phi)} + \\
			&\qquad\quad \abs{\Tr ([V_{e}(\tau'), U_{\mathcal{B}}(\tau', \tau) h_{f}U_{\mathcal{B}}(\tau, \tau')]\phi)} + \abs{\Tr ([V_{e}(\tau'), U_{\mathcal{B}}(\tau', \tau) V_{f}(\tau)U_{\mathcal{B}}(\tau, \tau')]\phi)}.
		\end{aligned}
    \end{equation}
    It will also be useful to derive a simple property: Suppose $O$ is an operator, possibly acting jointly on the system and environment, such that $[O, a_{\alpha, f}(\tau)] = [O, a^{\dagger}_{\alpha, f}(\tau)] = 0$. Then consider the commutator $[O, U_{\mathcal{B}}(\tau', \tau)V_{f}(\tau) U_{\mathcal{B}}(\tau, \tau')]$. We can write
    \begin{equation}
    \label{eq:commutator_expansion_basic}
		\begin{aligned}
			&[O, U_{\mathcal{B}}(\tau', \tau)V_{f}(\tau) U_{\mathcal{B}}(\tau, \tau')] 
            = \sum_{\alpha} [O, U_{\mathcal{B}}(\tau', \tau)L_{\alpha, f}a^{\dagger}_{\alpha, f}(\tau)U_{\mathcal{B}}(\tau, \tau')] \\
			&\qquad \qquad \qquad \qquad \qquad +\sum_\alpha [O, U_{\mathcal{B}}(\tau', \tau)L^\dagger_{\alpha, f} a_{\alpha, f}(\tau)U_{\mathcal{B}}(\tau, \tau')].
		\end{aligned}
    \end{equation}
    Next, from \cref{lemma:input_output}, we have that
        \begin{subequations}\label{eq:inp_out_form_edge}
        \begin{align}
        &a_{\alpha, f}(\tau) U_{\mathcal{B}}(\tau, \tau') = U_{\mathcal{B}}(\tau, \tau') a_{\alpha, f}(\tau) - \frac{i}{2}\mathcal{I}_{\mathcal{B}}(f)L_{\alpha, f}U_{\mathcal{B}}(\tau, \tau'),\\
        & U_{\mathcal{B}}(\tau', \tau)a^{\dagger}_{\alpha, f}(\tau) = a^{\dagger}_{\alpha, f}(\tau) U_{\mathcal{B}}(\tau', \tau) + \frac{i}{2}\mathcal{I}_{\mathcal{B}}(f) U_{\mathcal{B}}(\tau', \tau)L^{\dagger}_{\alpha, f}.
        \end{align}
        \end{subequations}
        Using \cref{eq:inp_out_form_edge}, we can then obtain that
        \begin{align*}
        &[O, U_{\mathcal{B}}(\tau', \tau)L^\dagger_{\alpha, f} a_{\alpha, f}(\tau)U_{\mathcal{B}}(\tau, \tau')]\\
        &\qquad = [O, U_{\mathcal{B}}(\tau', \tau)L^\dagger_{\alpha, f}U_{\mathcal{B}}(\tau, \tau')a_{\alpha, f}(\tau)] - \frac{i}{2}\mathcal{I}_{\mathcal{B}}(f)[O, U_{\mathcal{B}}(\tau', \tau)L^\dagger_{\alpha, f}L_{\alpha, f}U_{\mathcal{B}}(\tau, \tau')],\\
        &\qquad = [O, U_{\mathcal{B}}(\tau', \tau)L^\dagger_{\alpha, f}U_{\mathcal{B}}(\tau, \tau')]a_{\alpha,f}(\tau) - \frac{i}{2}\mathcal{I}_{\mathcal{B}}(f)[O, U_{\mathcal{B}}(\tau', \tau)L^\dagger_{\alpha, f}L_{\alpha, f}U_{\mathcal{B}}(\tau, \tau')],
        \end{align*}
        where, in the second step, we have used the fact that $[O, a_{\alpha, f}(\tau)] = 0$. Performing a similar analysis for $[O, U_{\mathcal{B}}(\tau', \tau)L_{\alpha, f}a^\dagger_{\alpha, f}(\tau)U_{\mathcal{B}}(\tau, \tau')]$ and using \cref{eq:commutator_expansion_basic}, we obtain that
    \begin{equation}\label{eq:commutator_expansion_final}
		\begin{aligned}
			&[O, U_{\mathcal{B}}(\tau', \tau)V_{f}(\tau) U_{\mathcal{B}}(\tau, \tau')]\\
			&=\bigg(\sum_\alpha [O, U_{\mathcal{B}}(\tau', \tau)L^{\dagger}_{\alpha, f}U_{\mathcal{B}}(\tau, \tau')]a_{\alpha,f}(\tau) + \sum_\alpha a^{\dagger}_{\alpha,f}(\tau)[O, U_{\mathcal{B}}(\tau', \tau)L_{\alpha, f}U_{\mathcal{B}}(\tau, \tau')]\bigg).
		\end{aligned}
    \end{equation}
     We next return to \cref{eq:triangle_inequality} and bound it term by term.
    \begin{enumerate}
        \item[1.] It follows directly from the Lieb-Robinson bound (\cref{lemma:lr_bound}) and the fact that, for any bounded $h$, $\smallnorm{[h, \cdot]}_\diamond \leq 2 \smallnorm{h}$ that
        \begin{equation}\label{eq:h_h_bound}
        \abs{\Tr ([h_{e}, U_{\mathcal{B}}(\tau', \tau) h_{f}U_{\mathcal{B}}(\tau, \tau')]\phi)} \leq 2C_{\mathrm{edge}} \exp(v(K)\abs{\tau - \tau'}/\xi_{0}- l/\xi_0).
        \end{equation}
        \item[2.] Consider now $\abs{\Tr ([h_{e}, U_{\mathcal{B}}(\tau', \tau) V_{f}(\tau)U_{\mathcal{B}}(\tau, \tau')]\phi)}$. Since $[h_{e}, a_{\alpha,f}(\tau)] = [h_{e}, a^{\dagger}_{\alpha,f}(\tau)] = 0$, we can apply \cref{eq:commutator_expansion_final} to obtain
        \begin{equation}
			\begin{aligned}
				&\abs{\Tr ([h_{e}, U_{\mathcal{B}}(\tau', \tau) V_{f}(\tau)U_{\mathcal{B}}(\tau, \tau')]\phi)}\\
				&\leq \sum_\alpha \bigg( \abs{\Tr ([h_{e}, U_{\mathcal{B}}(\tau', \tau)L_{\alpha, f}U_{\mathcal{B}}(\tau, \tau')]\phi_{\textnormal{R}})} + \abs{\Tr ([h_{e}, U_{\mathcal{B}}(\tau', \tau)L^{\dagger}_{\alpha, f}U_{\mathcal{B}}(\tau, \tau')]\phi_{\textnormal{L}})}\bigg),
			\end{aligned}
        \end{equation}
        where $\phi_{\textnormal{L}} =a_{\alpha, f}(\tau) \phi, \phi_{\textnormal{R}} = \phi a^{\dagger}_{\alpha, f}(\tau)$. Now, we use \cref{lemma:op_space_prop} to write $\phi_{\textnormal{L}}= \sum_{i} c_{\textnormal{L}, i}\phi_{\textnormal{L}, i}$ and $\phi_{\textnormal{R}}= \sum_{i}c_{\textnormal{R}, i}\phi_{\textnormal{R}, i}$ with $\phi_{\textnormal{L}, i}, \phi_{\textnormal{R}, i} \in \mathcal{S}_{K}$ and $\sum_i \smallabs{c_{\textnormal{L}, i}} \leq K$ and $\sum_i \smallabs{c_{\textnormal{R}, i}} \leq K$. Applying the Lieb-Robinson bound (\cref{lemma:lr_bound}) to each of these terms gives the prefactor $4KC_{\mathrm{edge}}\sum_\alpha\norm{L_{\alpha,f}}\leq4KC_{\mathrm{edge}}$. Thus,
        \begin{align}
 &\abs{\Tr ([h_{e}, U_{\mathcal{B}}(\tau', \tau) V_{f}(\tau)U_{\mathcal{B}}(\tau, \tau')]\phi)} \leq4KC_{\mathrm{edge}} \exp(v(K)\abs{\tau - \tau'}/\xi_{0}- l/\xi_0).\label{eq:h_v_bound}
\end{align}

        A similar argument can be used to derive the same upper bound for \\ $\abs{\Tr ([V_{e}(\tau'), U_{\mathcal{B}}(\tau', \tau) h_{f}(\tau)U_{\mathcal{B}}(\tau, \tau')]\phi)}$.

        \item[3.] Finally, consider $\abs{\Tr ([V_{e}(\tau'), U_{\mathcal{B}}(\tau', \tau) V_{f}(\tau)U_{\mathcal{B}}(\tau, \tau')]\phi)}$. To bound this term, we first analyze $\Tr ([V_{e}(\tau'), U_{\mathcal{B}}(\tau', \tau)O_{f}U_{\mathcal{B}}(\tau, \tau')]\phi)$ where $O_{f}$ is a system operator supported on the edge $f$ and $\phi \in \mathcal{S}_K$. We begin by using the triangle inequality
        \begin{align*}
        &\abs{\Tr ([V_{e}(\tau'), U_{\mathcal{B}}(\tau', \tau)O_{f}U_{\mathcal{B}}(\tau, \tau')]\phi)}\\
        &\qquad \qquad \leq \sum_{\sigma \in \{+, -\}} \sum_\alpha \abs{\Tr ([L^{(\sigma)}_{\alpha, e}a^{(\bar{\sigma})}_{\alpha,e}(\tau') , U_{\mathcal{B}}(\tau', \tau)O_{f}U_{\mathcal{B}}(\tau, \tau')]\phi)},
        \end{align*}
        where $O^{(\sigma)} = O$ if $\sigma =+$ and $O^\dagger$ if $\sigma =-$, and $\bar{\sigma} = +$ if $\sigma = -$ and $-$ if $\sigma = +$. Next, using \cref{lemma:input_output}, we first write
        \begin{align*}
        &a_{\alpha, e}(\tau')U_{\mathcal{B}}(\tau', \tau) = U_{\mathcal{B}}(\tau', \tau)a_{\alpha, e}(\tau') + \frac{i}{2}\mathcal{I}_{\mathcal{B}}(e)L_{\alpha, e}U_{\mathcal{B}}(\tau', \tau), \\
        &a_{\alpha, e}(\tau')U_{\mathcal{B}}(\tau, \tau') = U_{\mathcal{B}}(\tau, \tau')a_{\alpha, e}(\tau') - \frac{i}{2}\mathcal{I}_{\mathcal{B}}(e)U_{\mathcal{B}}(\tau, \tau')L_{\alpha, e}.
        \end{align*}
        Using these equations, we can then obtain
        \begin{equation}
			\begin{aligned}
				&\Tr ([L^{\dagger}_{\alpha, e}a_{\alpha, e}(\tau') , U_{\mathcal{B}}(\tau', \tau)O_{f}U_{\mathcal{B}}(\tau, \tau')]\phi) \\
				&\qquad  = \bigg(\Tr ([L^{\dagger}_{\alpha, e}, U_{\mathcal{B}}(\tau', \tau)O_{f}U_{\mathcal{B}}(\tau, \tau')]\phi_{\textnormal{L}}) + \\
				&\qquad \qquad \qquad \qquad \frac{i}{2}\norm{L_{\alpha, e}}\mathcal{I}_{\mathcal{B}}(e) \Tr ([L_{\alpha, e}, U_{\mathcal{B}}(\tau', \tau)O_{f}U_{\mathcal{B}}(\tau, \tau')]\phi_{\textnormal{R}})\bigg),
			\end{aligned}
        \end{equation}
        where $\phi_{\textnormal{L}} = a_{\alpha, e}(\tau')\phi$ and $\phi_{\textnormal{R}} = \phi L^{\dagger}_{\alpha, e}/ \norm{L_{\alpha,e}}$. Note that $\phi_{\textnormal{R}} \in \mathcal{S}_K$ and, from \cref{lemma:op_space_prop}, $\phi_{\textnormal{L}}= \sum_{i} c_i \phi_{\textnormal{L}, i}$ with $\phi_{\textnormal{L}, i} \in \mathcal{S}_{K}$ and $\sum_i \abs{c_i} \leq K$. Thus, using \cref{lemma:lr_bound}, we obtain that
        \begin{equation}
			\begin{aligned}
				\  &\abs{\Tr ([L^{\dagger}_{\alpha, e}a_{\alpha, e}(\tau') , U_{\mathcal{B}}(\tau', \tau)O_{f}U_{\mathcal{B}}(\tau, \tau')]\phi)} \\
				&\qquad \qquad \leq C_{\mathrm{edge}} \norm{O_{f}}\bigg (2K\smallnorm{L_{\alpha, e}} + \smallnorm{L_{\alpha, e}}^2\bigg) \exp(v(K)\abs{\tau - \tau'}/\xi_{0}- d(e, f)/\xi_0).
			\end{aligned}
        \end{equation}
        The same upper bound can be derived for $\smallabs{\Tr ([L_{\alpha, e} a^{\dagger}_{\alpha, e}(\tau') , U_{\mathcal{B}}(\tau', \tau)O_{f}U_{\mathcal{B}}(\tau, \tau')]\phi)}$ following a similar procedure. Summing the two estimates gives the prefactor $2C_{\mathrm{edge}}\norm{O_f}\sum_\alpha(2K\norm{L_{\alpha,e}}+\norm{L_{\alpha,e}}^2)\leq2C_{\mathrm{edge}}(2K+1)\norm{O_f}$, using \cref{eq:assumption_of_local_norm}. Therefore,
        \begin{equation}\label{eq:bound_final_part_3}
			\begin{aligned}
				&\abs{\Tr ([V_{e}(\tau'), U_{\mathcal{B}}(\tau', \tau)O_{f}U_{\mathcal{B}}(\tau, \tau')]\phi)}\\
				&\qquad \qquad \qquad  \leq 2C_{\mathrm{edge}} (2K + 1) \norm{O_{f}}\exp(v(K)\abs{\tau - \tau'}/\xi_{0}- l/\xi_0).
			\end{aligned}
        \end{equation}
        Finally, we note that $[V_{e}(\tau'), a_{\alpha, f}(\tau)] = [V_{e}(\tau'), a^{\dagger}_{\alpha, f}(\tau)] = 0$ and therefore from \cref{eq:commutator_expansion_final}, it follows that
        \begin{equation}
			\begin{aligned}
				&\abs{\Tr ([V_{e}(\tau'), U_{\mathcal{B}}(\tau', \tau)V_{f}(\tau)U_{\mathcal{B}}(\tau, \tau')] \phi)} \\
				&\leq \sum_\alpha \bigg(\abs{\Tr ([V_{e}(\tau'), U_{\mathcal{B}}(\tau', \tau)L^{\dagger}_{\alpha, f}U_{\mathcal{B}}(\tau, \tau')] a_{\alpha,f}(\tau)\phi)} +\\
				&\qquad \qquad \qquad \qquad \qquad \qquad  \abs{\Tr ([V_{e}(\tau'), U_{\mathcal{B}}(\tau', \tau)L_{\alpha, f} U_{\mathcal{B}}(\tau, \tau')] \phi a^{\dagger}_{\alpha, f}(\tau))}\bigg).
			\end{aligned}
        \end{equation}
        Using the fact that, from \cref{lemma:op_space_prop}, $a_{\alpha, f}(\tau)\phi = \sum_{i}c_{\textnormal{L}, i} \phi_{\textnormal{L}, i}$ and $\phi a^{\dagger}_{\alpha, f}(\tau) = \sum_{i}c_{\textnormal{R}, i} \phi_{\textnormal{R}, i}$ where $\phi_{\textnormal{L}, i}, \phi_{\textnormal{R}, i} \in \mathcal{S}_{K}$ and $\sum_i \smallabs{c_{\textnormal{L}, i}} \leq K, \sum_i \smallabs{c_{\textnormal{R}, i}} \leq K$, together with \cref{eq:bound_final_part_3}, gives the prefactor $4K(2K+1)C_{\mathrm{edge}}\sum_\alpha\norm{L_{\alpha,f}}\leq4K(2K+1)C_{\mathrm{edge}}$. Hence we obtain
        \begin{equation}\label{eq:v_v_bound}
			\begin{aligned}
				&\abs{\Tr ([V_{e}(\tau'), U_{\mathcal{B}}(\tau', \tau)V_{f}(\tau)U_{\mathcal{B}}(\tau, \tau')] \phi)}\\
				&\qquad \qquad \qquad  \leq 4K(2K + 1)C_{\mathrm{edge}} \exp(v(K)\abs{\tau - \tau'}/\xi_{0}- l/\xi_0),
			\end{aligned}
        \end{equation}
    \end{enumerate}
    Combining \cref{eq:triangle_inequality,eq:h_h_bound,eq:h_v_bound,eq:v_v_bound}, we obtain the lemma statement.
\end{proof}

\begin{proof}[Proof of \cref{lemma:superop_comm_bound}]
    For notational convenience, define
    \begin{equation}
    V = \prod^{M}_{j = 1}U_{\mathcal{A}_j}(t_j,s_j) \text{ and }V' = \prod^{M'}_{j = 1}U_{\mathcal{A}^{\prime}_j}(t_j', s_j'),
    \end{equation}
    and $\mathcal{V}, \mathcal{V}'$ as their vectorized versions.

    We begin by noting the following simple identity:
    \begin{equation}
    [\mathcal{C}_A, \mathcal{U}\mathcal{C}_B \mathcal{U}^\dagger]\vecket{\Psi} \equiv [A,UBU^\dagger]\Psi - \Psi[A, UBU^\dagger].
    \end{equation}
    Applying this identity, we obtain that for any system state $\sigma_S$,
    \begin{equation}\label{eq:expression_as_commutator}
		\begin{aligned}
			&\vecbra{O_S, I_E}\mathcal{V}[\mathcal{H}_{e}(\tau'), \mathcal{U}_{\mathcal{B}}(\tau', \tau) \mathcal{H}_{f}(\tau)\mathcal{U}_{\mathcal{B}}(\tau, \tau')] \mathcal{V}'\vecket{\sigma_S, \rho_E(0)}\\
			&\qquad = \Tr [V^\dagger O_S V[H_{e}(\tau'), U_{\mathcal{B}}(\tau', \tau) H_{f}(\tau)U_{\mathcal{B}}(\tau, \tau')] V'(\sigma_S\otimes \rho_E(0) )(V')^\dagger)-\\
			&\qquad \qquad \Tr [V^\dagger O_SVV'(\sigma_S\otimes \rho_E(0) )(V')^\dagger [H_{e}(\tau'), U_{\mathcal{B}}(\tau', \tau) H_{f}(\tau)U_{\mathcal{B}}(\tau, \tau')]], \\
			&\qquad = \sum_{i \in \{1, 2\}}(-1)^{i+1}\Tr ([H_{e}(\tau'), U_{\mathcal{B}}(\tau', \tau) H_{f}(\tau) U_{\mathcal{B}}(\tau, \tau')]\phi_i),
		\end{aligned}
    \end{equation}
    where
    \begin{subequations}
         \begin{align}
        &\phi_1 = V'(\sigma_{S}\otimes \rho_E(0)) (V')^\dagger V^\dagger O_S V \equiv \mathcal{V}^\dagger (O_{S})_R \mathcal{V}\mathcal{V}'\vecket{\sigma_S, \rho_E(0)}, \\
        &\phi_2 = V^\dagger O_S VV'(\sigma_S \otimes \rho_E(0)) (V')^\dagger\equiv \mathcal{V}^\dagger (O_S)_L \mathcal{V}\mathcal{V}'\vecket{\sigma_S, \rho_E(0)}.
    \end{align}   
    \end{subequations}
    Now, note that $\phi_1, \phi_2 \in \mathcal{S}_{2M + M'}$. Applying \cref{lemma:commutator_H_bound}, we then obtain that
    \begin{equation}
		\begin{aligned}
			&\abs{\Tr ([H_{e}(\tau'), U_{\mathcal{B}}(\tau', \tau) H_{f}(\tau)U_{\mathcal{B}}(\tau, \tau')] \phi_i)} \\
			&\qquad \qquad \qquad \leq C(2M + M') \exp((v(2M + M')\abs{\tau - \tau'} - d(e, f))/\xi_0),
		\end{aligned}
    \end{equation}
    where $C(\cdot)$ as defined in \cref{lemma:commutator_H_bound} and $v(\cdot)$ is defined in \cref{lemma:lr_bound}. From this estimate and \cref{eq:expression_as_commutator}, the lemma follows.
\end{proof}

\section{Analyzing the collision model}\label{sec:collision-analysis}
Now fix an edge set $\mathcal{A}_{i}$, let $p\in\{\mathrm{e},\mathrm{o}\}$ {denote}
even and odd edges, respectively. The parity of an edge is the parity of
its left endpoint; intersections of edges refer to their endpoint
sets. We introduce 
\begin{equation}
H_{i,p}(s)=\sum_{\substack{e\ \text{of\ }p\\
e\in\mathcal{A}_{i}
}
}H_{e}(s),\ \ \text{with \ \ }H_{e}(\tau)=h_{e}+\sum_{\alpha}(L_{\alpha,e}a^{\dagger}_{\alpha,e}(\tau)+\text{h.c.}).
\end{equation}
The corresponding unitary is 
\begin{equation}
U_{i,p}(t,s)=\mathcal{T}\exp\left[-i\int^{t}_{s}H_{i,p}(\tau)d\tau\right].
\end{equation}
We define the evolution after Trotter decomposition as 
\begin{equation}
U^{\text{split }}_{i,j}=\begin{cases}
U_{i,\mathrm{e}}(j\delta,(j-1)\delta)U_{i,\mathrm{o}}(j\delta,(j-1)\delta) & \text{if}\ \epsilon_{i}=+1,\\
U_{i,\mathrm{o}}((j-1)\delta,j\delta)U_{i,\mathrm{e}}((j-1)\delta,j\delta) & \text{if}\ \epsilon_{i}=-1,
\end{cases}
\end{equation}
and 

\begin{equation}
U^{\text{split}}=\prod^{M}_{i=1}\prod_{j:n_{i}\to m_{i}}U^{\text{split}}_{i,j}\ \ \text{with}\ \ \prod_{j:n_{i}\to m_{i}}U^{\text{split}}_{i,j}=\begin{cases}
U^{\text{split}}_{i,n_{i}}U^{\text{split}}_{i,n_{i}-1}\dots U^{\text{split}}_{i,m_{i}+1} & \text{if }\epsilon_{i}=+1,\\
U^{\text{split}}_{i,n_{i}+1}U^{\text{split}}_{i,n_{i}+2}\dots U^{\text{split}}_{i,m_{i}} & \text{if }\epsilon_{i}=-1.
\end{cases}
\end{equation}
The corresponding reduced Trotter channel is then defined as 
\begin{equation}
\mathcal{E}^{\text{split}}(X)=\operatorname{Tr}_{E}\!\left[U^{\text{split}}\bigl(X\otimes|\mathrm{vac}\rangle\langle\mathrm{vac}|\bigr)(U^{\text{split}})^{\dagger}\right].
\end{equation}

\begin{lemma}[Trotter decomposition error]\label{lemma:continuous_even_odd_splitting}

There is a constant $C$, independent of $M,N,t,\delta$, such that
\begin{equation}
\|\mathcal{E}-\mathcal{E}^{\text{split}}\|_{\diamond}\leq CM^{3}Nt\delta.\label{eq:continuous_local_splitting}
\end{equation}
Here $M$ characterizes the number of forward and backward rounds.
For the HHKL algorithm introduced above, we can choose $M=3$. \end{lemma} 
\begin{proof}
 We use the doubled-operator notation from the previous
subsection: 
\begin{equation}
X_{L}(Y)=XY,\qquad X_{R}(Y)=YX,\qquad\mathcal{H}_{e}(s)=(H_{e}(s))_{L}-(H_{e}(s))_{R}.
\end{equation}

Assume $\epsilon_{i}=+1$. We define $\mathcal{U}^{\text{split }}_{i,j}(r,u)=\mathcal{U}_{i,\mathrm{e}}(r,u)\mathcal{U}_{i,\mathrm{o}}(r,u)$
and $\mathcal{U}_{i,j}(r,u)=\mathcal{U}_{\mathcal{A}_{i}}(r,u).$ First, differentiating
the even/odd product gives 
\begin{equation}
i\partial_{r}\mathcal{U}^{\text{split }}_{i,j}(r,u)= \left(\mathcal{H}_{i,\mathrm{e}}(r) +\mathcal{U}_{i,\mathrm{e}}(r,u)\mathcal{H}_{i,\mathrm{o}}(r) \mathcal{U}_{i,\mathrm{e}}(u,r)\right)\mathcal{U}^{\text{split}}_{i,j}(r,u).
\end{equation}
Variation of constants therefore yields 
\begin{equation}
\begin{aligned}\relax\mathcal{U}_{i,j}(v,u)-\mathcal{U}^{\text{split}}_{i,j}(v,u) =  &  -i\int^{v}_{u} d\tau\,\mathcal{U}_{i,j}(v,\tau)\times\\
  &  \left(\mathcal{H}_{i,\mathrm{o}}(\tau) -\mathcal{U}_{i,\mathrm{e}}(\tau,u)\mathcal{H}_{i,\mathrm{o}}(\tau) \mathcal{U}_{i,\mathrm{e}}(u,\tau)\right)\mathcal{U}^{\text{split}}_{i,j}(\tau,u). 
\end{aligned}
\end{equation}
For fixed $\tau$, the difference in parentheses is another integral:
\begin{equation}
\mathcal{H}_{i,\mathrm{o}}(\tau)-\mathcal{U}_{i,\mathrm{e}}(\tau,u)\mathcal{H}_{i,\mathrm{o}}(\tau)\mathcal{U}_{i,\mathrm{e}}(u,\tau)\quad=i\int^{\tau}_{u}ds\,\mathcal{U}_{i,\mathrm{e}}(\tau,s)[\mathcal{H}_{i,\mathrm{e}}(s),\mathcal{H}_{i,\mathrm{o}}(\tau)]\mathcal{U}_{i,\mathrm{e}}(s,\tau).
\end{equation}
Combining the two identities gives the exact Trotter defect 
\begin{equation}
\begin{aligned}\relax &  \mathcal{U}_{i,j}(v,u)-\mathcal{U}^{\text{split}}_{i,j}(v,u)\\
  &  =\int^{v}_{u} d\tau\int^{\tau}_{u}ds\, \mathcal{U}_{i,j}(v,\tau)\mathcal{U}_{i,\mathrm{e}}(\tau,s) [\mathcal{H}_{i,\mathrm{e}}(s),\mathcal{H}_{i,\mathrm{o}}(\tau)] \mathcal{U}_{i,\mathrm{e}}(s,u)\mathcal{U}_{i,\mathrm{o}}(\tau,u). 
\end{aligned}
\label{eq:continuous_trotter_defect}
\end{equation}
Similarly, backward evolution with $\epsilon_{i}=-1$
uses the adjoint identity. 

Let $\rho_{S}$ be a system density operator, and {let}
$O_{S}=O^{\dagger}_{S}$ with $\|O_{S}\|\leq1$ {be}
a system observable \footnote{More rigorously, the diamond norm requires the channel to act jointly
on the system and an arbitrary reference system. This does not affect
the result, since the reference system can be absorbed into the system
by extending the Hamiltonian trivially, i.e., by tensoring it with
the identity on the reference.{{} Throughout this proof,
$S$ implicitly includes this arbitrary reference, while $N$ counts
only the physical lattice sites.}}. With $\rho_{E}(0)=|\mathrm{vac}\rangle\langle\mathrm{vac}|$, the
diamond difference is 
\begin{equation}
\|\mathcal{E}-\mathcal{E}^{\text{split}}\|_{\diamond}=\sup_{\rho_{S},O_{S}}\left|\vecbra{O_{S},I_{E}}(\mathcal{U}-\mathcal{U}^{\text{split}})\vecket{\rho_{S},\rho_{E}(0)}\right|.
\end{equation}
For {a fixed pair} $(i,j)$, we introduce $\mathcal{U}^{\text{split}}_{>(i,j)}$
as the product of Trotterized factors to the left of $(i,j)$, and
$\mathcal{U}_{<(i,j)}$ as the product of the un-Trotterized factors
to its right. Telescoping sums give 
\begin{equation}
\mathcal{U}-\mathcal{U}^{\text{split}}=\sum_{i,j}\mathcal{U}^{\text{split}}_{>(i,j)}(\mathcal{U}_{i,j}-\mathcal{U}^{\text{split}}_{i,j})\mathcal{U}_{<(i,j)},
\end{equation}
and hence 
\begin{equation}
\|\mathcal{E}-\mathcal{E}^{\text{split}}\|_{\diamond}\leq\sum_{i,j}\sup_{\rho_{S},O_{S}}\left|\vecbra{O_{S},I_{E}}\mathcal{U}^{\text{split}}_{>(i,j)}(\mathcal{U}_{i,j}-\mathcal{U}^{\text{split}}_{i,j})\mathcal{U}_{<(i,j)}\vecket{\rho_{S},\rho_{E}(0)}\right|.\label{eq:telescoping_sum_of_diamond_difference}
\end{equation}

Fix one summand and set $u=(j-1)\delta$, $v=j\delta$. For a forward
evolution, \cref{eq:continuous_trotter_defect} yields
\begin{equation}
\begin{aligned} &  \left|\vecbra{O_{S},I_{E}}\mathcal{U}^{\text{split}}_{>(i,j)}(\mathcal{U}_{i,j}-\mathcal{U}^{\text{split}}_{i,j})\mathcal{U}_{<(i,j)}\vecket{\rho_{S},\rho_{E}(0)}\right|\\
  &  \leq\int^{v}_{u} d\tau\int^{\tau}_{u}ds\,\left|\vecbra{O_{S},I_{E}}\mathcal{U}^{\text{split}}_{>(i,j)}\mathcal{U}_{i,j}(v,\tau)\mathcal{U}_{i,\mathrm{e}}(\tau,s)[\mathcal{H}_{i,\mathrm{e}}(s),\mathcal{H}_{i,\mathrm{o}}(\tau)]\mathcal{U}_{i,\mathrm{e}}(s,u)\mathcal{U}_{i,\mathrm{o}}(\tau,u)\mathcal{U}_{<(i,j)}\vecket{\rho_{S},\rho_{E}(0)}\right|\\
  &  =:\int^{v}_{u}d\tau\int^{\tau}_{u}ds\left|\vecbra{O_{S},I_{E}}\mathcal{W}_{i,j,>}(\tau,s)[\mathcal{H}_{i,\mathrm{e}}(s),\mathcal{H}_{i,\mathrm{o}}(\tau)]\mathcal{W}_{i,j,<}(\tau,s)\vecket{\rho_{S},\rho_{E}(0)}\right|,
\end{aligned}
\label{eq:one_term_error_in_telescoping_sum}
\end{equation}
where we have defined $\mathcal{W}_{i,j,>}(\tau,s)=\mathcal{U}^{\text{split}}_{>(i,j)}\mathcal{U}_{i,j}(v,\tau)\mathcal{U}_{i,\mathrm{e}}(\tau,s)$
and $\mathcal{W}_{i,j,<}{(\tau,s)}=\mathcal{U}_{i,\mathrm{e}}(s,u)\mathcal{U}_{i,\mathrm{o}}(\tau,u)\mathcal{U}_{<(i,j)}$.
For a backward visit, substitute the adjoint of \cref{eq:continuous_trotter_defect}.
We now expand the commutator in this integrand. Locality and the doubled-operator
identity give 
\begin{equation}
\begin{aligned}\relax[\mathcal{H}_{i,\mathrm{e}}(s),\mathcal{H}_{i,\mathrm{o}}(\tau)]  &  =\sum_{\substack{e\ \mathrm{even},\ f\ \mathrm{odd}\\
e,f\in\mathcal{A}_{i}\\
e\ne f,\ e\cap f\ne\varnothing
}
}[\mathcal{H}_{e}(s),\mathcal{H}_{f}(\tau)],\\
  &  =\sum_{\substack{e\ \mathrm{even},\ f\ \mathrm{odd}\\
e,f\in\mathcal{A}_{i}\\
e\ne f,\ e\cap f\ne\varnothing
}
}\Big\{([H_{e}(s),H_{f}(\tau)])_{L}-([H_{e}(s),H_{f}(\tau)])_{R}\Big\}.
\end{aligned}
\end{equation}
Write $X^{(+)}=X$, $X^{(-)}=X^{\dagger}$, and let $\bar{\sigma}$
denote the opposite of $\sigma\in\{+,-\}$. The physical commutator
has four contributions: 
\begin{equation}
\begin{aligned}\relax &  [H_{e}(s),H_{f}(\tau)]\\
  &  =[h_{e},h_{f}]+\sum_{\beta,\sigma}[h_{e},L^{(\sigma)}_{\beta,f}]a^{(\bar{\sigma})}_{\beta,f}(\tau)+\\
  &  \quad\sum_{\alpha,\sigma}[L^{(\sigma)}_{\alpha,e},h_{f}]a^{(\bar{\sigma})}_{\alpha,e}(s)+\sum_{\alpha,\beta,\sigma,\sigma'}[L^{(\sigma)}_{\alpha,e},L^{(\sigma')}_{\beta,f}]a^{(\bar{\sigma})}_{\alpha,e}(s)a^{(\bar{\sigma}')}_{\beta,f}(\tau).
\end{aligned}
\label{eq:expansion_of_commutators}
\end{equation}
We will bound each term in the commutator in the following. For the
first term, {since} $[h_{e},h_{f}]=0$ if $e\cap f=\varnothing$,
we have 
\begin{equation}
\sum_{\substack{e\ \mathrm{even},\ f\ \mathrm{odd}\\
e,f\in\mathcal{A}_{i}\\
e\ne f,\ e\cap f\ne\varnothing
}
}\lVert[h_{e},h_{f}]\rVert\leq2N.\label{eq:commutation_expansion_first_term}
\end{equation}
The other three terms contain the bosonic field $a_{\alpha,e}(s)$
which is unbounded, and therefore, we need to first use the Wick theorem
to contract them with the bath vacuum, and then {bound}
the norm of the remaining system operators. {For} the
last term, we have 
\begin{equation}
\begin{aligned} & \left|\vecbra{O_{S},I_{E}}\mathcal{W}_{i,j,>}(\tau,s)\left(L^{(\sigma)}_{\alpha,e}L^{(\sigma')}_{\beta,f}a^{(\bar{\sigma})}_{\alpha,e}(s)a^{(\bar{\sigma}')}_{\beta,f}(\tau)\right)_{L}\mathcal{W}_{i,j,<}(\tau,s)\vecket{\rho_{S},\rho_{E}(0)}\right|,\\
 & \;\leq(2M)^{2}\lVert L_{\alpha,e}\rVert^{2}\lVert L_{\beta,f}\rVert^{2}\int^{t}_{0}\delta(t_{1}-s)dt_{1}\int^{t}_{0}\delta(t_{2}-\tau)dt_{2},\\
 & \;\leq(2M)^{2}\lVert L_{\alpha,e}\rVert^{2}\lVert L_{\beta,f}\rVert^{2}.
\end{aligned}
\label{eq:two_a_in_diamond_difference}
\end{equation}
In the first inequality, we have repeatedly used
the relation 
\begin{equation}
[a_{\alpha,e}(\tau),U_{\mathcal{A}}(t,s)]={-i\mathcal{I}_{\mathcal{A}}(e)}\int^{t}_{s}dt_{1}U_{\mathcal{A}}(t,{t}_{1}){L}_{\alpha,e}U_{\mathcal{A}}(t_{1},s)\delta(\tau-t_{1})
\end{equation}
and 
\begin{equation}
[a_{\alpha,e}(\tau),\prod_{i,j}U_{i,j}]=\sum_{i,j}U_{>(i,j)}[a_{\alpha,e}{(\tau)},U_{i,j}]U_{<(i,j)}
\end{equation}
to move the annihilation operator to the left of $\rho_{E}(0)$ and
creation operator to the right of $\rho_{E}(0)$ as well as $a_{\alpha,e}(\tau)\rho_{E}(0)=\rho_{E}(0)a^{\dagger}_{\alpha,e}{(\tau)}=0$.
The last inequality of \cref{eq:two_a_in_diamond_difference}
comes from the requirement $\int^{t}_{0}\delta(r-s)dr\leq1$.
The factor $2$ {accounts for the ket and bra branches},
and $M$ is due to the alternating forward and backward evolution
so that each time interval can be visited up to $M$ times. {In
\cref{eq:one_term_error_in_telescoping_sum}, the unitary $\mathcal{W}_{i,j,>}$
and $\mathcal{W}_{i,j,<}$ together partition the evolution from $u$
to $v$ into $(u,s),(s,\tau),(\tau,v)$ for an even edge and $(u,\tau),(\tau,v)$
for an odd edge, so they do not increase this count.} The square is
because we have both $a^{(\bar{\sigma})}_{\alpha,e}(s)$ and $a^{(\bar{\sigma}')}_{\beta,f}(\tau)$
so that we need to transport these two operators separately. {Since
$e\ne f$, the two inserted fields have no direct contraction. }Therefore,
we have
\begin{equation}
\sum_{\substack{e\ \mathrm{even},\ f\ \mathrm{odd}\\
e,f\in\mathcal{A}_{i}\\
e\ne f,\ e\cap f\ne\varnothing\\
\alpha,\beta,\sigma,\sigma'
}
}\left|\vecbra{O_{S},I_{E}}\mathcal{W}_{i,j,>}(\tau,s)\left([L^{(\sigma)}_{\alpha,e},L^{(\sigma')}_{\beta,f}]a^{(\bar{\sigma})}_{\alpha,e}(s)a^{(\bar{\sigma}')}_{\beta,f}(\tau)\right)_{L}\mathcal{W}_{i,j,<}(\tau,s)\vecket{\rho_{S},\rho_{E}(0)}\right|\leq32M^{2}N.\label{eq:commutation_expansion_last_term_error}
\end{equation}
Here we have used that $[L^{(\sigma)}_{\alpha,e},L^{(\sigma')}_{\beta,f}]=0$
if $e\cap f=\varnothing$ and $\sum_{\alpha}\|L_{\alpha,e}\|^{2}\leq1.$

{The} other two terms in \cref{eq:expansion_of_commutators}
can be bounded similarly, leading to
\begin{equation}
\sum_{\substack{e\ \mathrm{even},\ f\ \mathrm{odd}\\
e,f\in\mathcal{A}_{i}\\
e\ne f,\ e\cap f\ne\varnothing\\
\alpha,\sigma
}
}\left|\vecbra{O_{S},I_{E}}\mathcal{W}_{i,j,>}(\tau,s)\left([h_{e},L^{(\sigma)}_{\alpha,f}]a^{(\bar{\sigma})}_{\alpha,f}(\tau)\right)_{L}\mathcal{W}_{i,j,<}(\tau,s)\vecket{\rho_{S},\rho_{E}(0)}\right|\leq8MN.\label{eq:commutation_expansion_second_third_term_error}
\end{equation}
The same analysis applies to the right component. Thus, summing {\cref{eq:commutation_expansion_first_term,eq:commutation_expansion_last_term_error,eq:commutation_expansion_second_third_term_error}}
together and integrating over $d\tau,ds$ in \cref{eq:one_term_error_in_telescoping_sum}
gives rise to 
\begin{equation}
\begin{aligned}\relax\|\mathcal{E}-\mathcal{E}^{\text{split}}\|_{\diamond}  &  \leq100M^{2}N\sum_{i,j}\int^{j\delta}_{(j-1)\delta}d\tau\int^{\tau}_{(j-1)\delta}ds\leq50M^{3}Nt\delta, \end{aligned}
\end{equation}
which proves the Lemma.
\end{proof}

{Having established the continuous Trotterization decomposition,
we now define the collision model that approximates its propagators
with product form.} For every time step $(j-1)\delta\to j\delta$,
define the operator 
\begin{equation}
\sigma_{\alpha,e}[j]=P^{\leq1}_{\alpha,e}(j\delta,(j-1)\delta)\bigg(\int^{j\delta}_{(j-1)\delta}a_{\alpha,e}(t)\frac{dt}{\sqrt{\delta}}\bigg)P^{\leq1}_{\alpha,e}(j\delta,(j-1)\delta),
\end{equation}
where we have defined $P^{\leq k}_{\alpha,e}(t,s)$ as the projector
on $\leq k$ excitation subspace of the operator $\int^{t}_{s}a^{\dagger}_{\alpha,e}(\tau)a_{\alpha,e}(\tau)d\tau$.
The operator $\sigma_{\alpha,e}[j]$ acts nontrivially on the two-dimensional subspace
spanned by $\ket{\text{vac}}_{\alpha,e,j}$ and $\ket{1}_{\alpha,e,j}=\sigma^{\dagger}_{\alpha,e}[j]\ket{\text{vac}}$.
We then recover the collision model unitary as defined in the main text
\begin{equation}
{\begin{aligned}\hat{U}_{i,j}={}  &  \left(\prod_{\substack{e\ {\rm even}\\
e\in\mathcal{A}_{i}
}
} e^{-i\delta h_{e}}\prod_{\alpha}e^{-i\sqrt{\delta}(L^{\dagger}_{\alpha,e}\sigma_{\alpha,e}[j]+\mathrm{h.c.})}\right)\left(\prod_{\substack{e\ {\rm odd}\\
e\in\mathcal{A}_{i}
}
} e^{-i\delta h_{e}}\prod_{\alpha}e^{-i\sqrt{\delta}(L^{\dagger}_{\alpha,e}\sigma_{\alpha,e}[j]+\mathrm{h.c.})}\right) \end{aligned}
}\label{eq:even_odd_collision}
\end{equation}
{for a forward time step. For a backward time step,
use the adjoint of the entire right-hand side, with the same ancilla
labels. This reverses both the layers and the gate order within each
block, and $\hat{U}$ is defined similarly as $U^{\text{split}}$ with
$U^{\text{split }}_{i,j}$ replaced by $\hat{U}_{i,j}$.} {The reduced channel of this ordered collision circuit is}
\begin{equation}
\hat{\mathcal{E}}(X)=\operatorname{Tr}_{E}\!\left[\hat{U}\bigl(X\otimes|\mathrm{vac}\rangle\langle\mathrm{vac}|\bigr)\hat{U}^{\dagger}\right].
\end{equation}
With this notation, we now prove \cref{lemma:collision_model_bound} from \cref{sec:collision}.

\begin{proof}[Proof of \cref{lemma:collision_model_bound}]
We use the doubled-operator notation of \cref{lemma:continuous_even_odd_splitting}:
$\mathcal{U}(X)=UXU^{\dagger}$. Here $\rho_{S}$ is a density operator,
$O_{S}=O^{\dagger}_{S}$ with $\|O_{S}\|\leq1$, and $S$ includes
an arbitrary idle reference. Then 
\begin{equation}
\|\mathcal{E}^{\text{split}}-\hat{\mathcal{E}}\|_{\diamond}=\sup_{\rho_{S},O_{S}}\left|\vecbra{O_{S},I_{E}}(\mathcal{U}^{\text{split}}-\hat{\mathcal{U}})\vecket{\rho_{S},\rho_{E}(0)}\right|.
\end{equation}
First telescope over spatial baths. Label the edges
from left to right by $e_{x}=(x,x+1)$, $1\leq x<N$. For $x=0,\ldots,N-1$, let $U_{>x}$
be the evolution unitary obtained from $U^{\text{split}}$ by replacing
every bath on an edge $e_{y}\ (y>x)$ with its collision registers. Thus 
\begin{equation}
\begin{aligned}\relax U_{>N-1} & =U^{\text{split}},\qquad U_{>0}=\hat{U},\\
\mathcal{U}^{\text{split}}-\hat{\mathcal{U}} & =\sum^{N-1}_{x=1}(\mathcal{U}_{>x}-\mathcal{U}_{>x-1}).
\end{aligned}
\label{eq:telescoping_sum_over_spatial}
\end{equation}
 At fixed $x$, set $e=e_{x}$. To calculate $\mathcal{U}_{>x}-\mathcal{U}_{>x-1}$,
we then telescope over the $T$ time bins. Let $U_{>x,>j}$ additionally
replace the bath on edge $e$ in every time bin $k>j$ with the corresponding
collision register, where $j=0,\ldots,T$. Therefore
\begin{equation}
\begin{aligned}\relax U_{>x,>T} & =U_{>x},\qquad U_{>x,>0}=U_{>x-1},\\
\mathcal{U}_{>x}-\mathcal{U}_{>x-1} & =\sum^{T}_{j=1}(\mathcal{U}_{>x,>j}-\mathcal{U}_{>x,>j-1}).
\end{aligned}
\end{equation}
We then introduce the observable expectation value calculated with
collision registers as
\begin{equation}
\begin{aligned}\relax\langle O_{S}\rangle_{e,j} & :=\vecbra{O_{S},I_{E}}\mathcal{U}_{>x,>j}\vecket{\rho_{S},\rho_{E}(0)}.\end{aligned}
\label{eq:folded_representation_O_sxj}
\end{equation}
 The two telescopes give 
\begin{equation}
\|\mathcal{E}^{\text{split}}-\hat{\mathcal{E}}\|_{\diamond}\leq\sum_{e}\sum^{T}_{j=1}\sup_{\rho_{S},O_{S}}\left|\langle O_{S}\rangle_{e,j}-\langle O_{S}\rangle_{e,j-1}\right|.\label{eq:two_fold_telescoping_sum_together}
\end{equation}
We remark that the bath is traced only at the end.

For a fixed $e,j,$ we introduce $U_{e,j}:=U_{\{e\}}(j\delta,(j-1)\delta)$
and $\hat{U}_{e,j}$ as its collision replacement:
\begin{equation}
\begin{aligned}\relax U_{e,j} & =\mathcal{T}\exp\!\left[-i\int^{j\delta}_{(j-1)\delta}ds\left(h_{e}+\sum_{\alpha}\left(L^{\dagger}_{\alpha,e}a_{\alpha,e}(s)+L_{\alpha,e}a^{\dagger}_{\alpha,e}(s)\right)\right)\right],\\
\hat{U}_{e,j} & =e^{-i\delta h_{e}}\prod_{\alpha}e^{-i\sqrt{\delta}(L^{\dagger}_{\alpha,e}\sigma_{\alpha,e}[j]+L_{\alpha,e}\sigma^{\dagger}_{\alpha,e}[j])}.
\end{aligned}
\end{equation}
 Unfolding \cref{eq:folded_representation_O_sxj} gives
\begin{equation}
\begin{aligned}\relax\langle O_{S}\rangle_{e,j} & =\operatorname{Tr}_{SE}\!\left[(\rho_{S}\otimes\rho_{E}(0))B_{0}U^{\epsilon_{1}}_{e,j}B_{1}\cdots U^{\epsilon_{J}}_{e,j}B_{J}\right],\\
\langle O_{S}\rangle_{e,j-1} & =\operatorname{Tr}_{SE}\!\left[(\rho_{S}\otimes\rho_{E}(0))B_{0}\hat{U}^{\epsilon_{1}}_{e,j}B_{1}\cdots\hat{U}^{\epsilon_{J}}_{e,j}B_{J}\right],\\
J & \leq2M,\qquad\epsilon_{\ell}\in\{+1,-1\},\qquad\|B_{\ell}\|\leq1.
\end{aligned}
\end{equation}
 Here $J$ counts the selected visits on both ket and bra branches,
and ${\epsilon}_{\ell}$ denotes the forward or backward
evolution
with $U^{+1}=U$ and $U^{-1}=U^{\dagger}$. Every operator $B_{\ell}$ contains the intervening {propagators}
$U_{e',j'}$ and $\hat{U}_{e',j'}${, or their adjoints,}
with $(e,j)\neq(e',j')${, and exactly one of the}
$B_{\ell}${{} also contains} $O_{S}$.
It therefore acts trivially on the bath modes associated with edge
$e$ and time bin $j$.

We then expand both $U_{e,j},\hat{U}_{e,j}$ as well as their {adjoints}
in $\langle O_{S}\rangle_{e,j}$ and $\langle O_{S}\rangle_{e,j-1}$
{using} Dyso{n (Taylor) series, respectively.}
For $p,n\geq0$, let $\langle O_{S}\rangle^{p,n}_{e,j}$ and $\langle O_{S}\rangle^{p,n}_{e,j-1}$
be the {coefficients of terms} with $p$ Hamiltonian
{factors} and $2n$ {bath operators.
Contributions with an odd} number of {bath operators}
vanish {in} the {vacuum.} Each Hamiltonian
{factor} contributes one {power} of
$\delta$, and each vacuum pair contributes another power of $\delta$,
so these terms have order $\delta^{p+n}$: 
\begin{equation}
\langle O_{S}\rangle_{e,j}=\sum_{p,n\geq0}\langle O_{S}\rangle^{p,n}_{e,j}\delta^{p+n},\qquad\langle O_{S}\rangle_{e,j-1}=\sum_{p,n\geq0}\langle O_{S}\rangle^{p,n}_{e,j-1}\delta^{p+n}.
\end{equation}
 At zeroth order, 
\begin{equation}
\langle O_{S}\rangle^{0,0}_{e,j}=\langle O_{S}\rangle^{0,0}_{e,j-1}=\operatorname{Tr}_{SE}\!\left[(\rho_{S}\otimes\rho_{E}(0))B_{0}\cdots B_{J}\right]
\end{equation}
since the {zeroth-order} Dyson {term
is the identity.} At first order, the only possibilities are 
\begin{equation}
p+n=1\quad\Longleftrightarrow\quad(p,n)=(1,0)\ \text{or}\ (0,1).
\end{equation}
 For one Hamiltonian {factor,} both local propagators
have the {first-order coefficient} $-i\epsilon_{\ell}h_{e}$,
giving 
\begin{equation}
\begin{aligned}\relax\langle O_{S}\rangle^{1,0}_{e,j} & =\langle O_{S}\rangle^{1,0}_{e,j-1}=-i\sum^{J}_{\ell=1}\epsilon_{\ell}\operatorname{Tr}_{SE}\!\Bigl[(\rho_{S}\otimes\rho_{E}(0))B_{0}\cdots B_{\ell-1}h_{e}B_{\ell}\cdots B_{J}\Bigr].\end{aligned}
\end{equation}

For two {bath operators,} one annihilation and one
creation operator, we {have,}
\begin{equation}
\int^{j\delta}_{(j-1)\delta}ds\int^{s}_{(j-1)\delta}du\,\langle a_{\alpha,e}(s)a^{\dagger}_{\beta,e}(u)\rangle_{0}=\frac{\delta}{2}\delta_{\alpha\beta},\label{eq:continuous_contraction_same_Block}
\end{equation}
if these two operators come from the same $U^{\epsilon_{\ell}}_{e,j}$,
or
\begin{equation}
\int^{j\delta}_{(j-1)\delta}ds\int^{j\delta}_{(j-1)\delta}du\,\langle a_{\alpha,e}(s)a^{\dagger}_{\beta,e}(u)\rangle_{0}=\delta\,\delta_{\alpha\beta},
\end{equation}
if they {come} from different {propagators}
$U^{\epsilon_{\ell}}_{e,j},U^{\epsilon_{\ell'}}_{e,j}$. {For}
the collision operator $\sigma_{\alpha,e}[j]${, the
vacuum moments are}
\begin{equation}
\begin{aligned}\langle\sigma_{\alpha,e}[j]\sigma^{\dagger}_{\beta,e}[j]\rangle_{0} & =\delta_{\alpha\beta},\qquad\langle\sigma^{\dagger}_{\alpha,e}[j]\sigma_{\beta,e}[j]\rangle_{0}=0.\end{aligned}
\end{equation}
{For a nonzero contraction,} when $\sigma_{\alpha,e}[j]$
{and} $\sigma^{\dagger}_{\beta,e}[j]$ {come}
from the same $\hat{U}^{\epsilon_{\ell}}_{e,j}$, {the
second-order} Taylor {term} contributes a {factor}
of $1/2$, agreeing with \cref{eq:continuous_contraction_same_Block}.
Therefore,
\begin{equation}
\begin{aligned}\langle O_{S}\rangle^{0,1}_{e,j}=\langle O_{S}\rangle^{0,1}_{e,j-1} & =-\frac{1}{2}\sum^{J}_{\ell=1}\sum_{\alpha}\operatorname{Tr}_{SE}\!\Bigl[(\rho_{S}\otimes\rho_{E}(0))B_{0}\cdots B_{\ell-1}L^{\dagger}_{\alpha,e}L_{\alpha,e}B_{\ell}\cdots B_{J}\Bigr]-\\
 & \quad\sum_{1\leq\ell<m\leq J}\epsilon_{\ell}\epsilon_{m}\sum_{\alpha}\operatorname{Tr}_{SE}\!\Bigl[(\rho_{S}\otimes\rho_{E}(0))B_{0}\cdots B_{\ell-1}L^{\dagger}_{\alpha,e}\times B_{\ell}\cdots B_{m-1}L_{\alpha,e}B_{m}\cdots B_{J}\Bigr].
\end{aligned}
\end{equation}
 We have thus shown that $\langle O_{S}\rangle_{e,j}$ agrees {with}
$\langle O_{S}\rangle_{e,j-1}$ up to the first order in $\delta$.

We then bound the {higher-order} terms in $\delta$.
After {the} Dyson expansion, the {terms
contributing to} $\langle O_{S}\rangle^{p,n}_{e,j}$ {contain}
$p$ Hamiltonian factors and $2n$ bath operators. Each Hamiltonian
factor can occur in any of the $J$ propagators with integration weight
at most $\delta\|h_{e}\|$. The total weight of the $p$ Hamiltonian
factors is therefore bounded by $(J\delta\lVert h_{e}\rVert)^{p}${.
For} the $2n$ bath operators, Wick's theorem expresses the vacuum
correlation as a sum over $(2n-1)!!$ pairings. For each fixed pairing,
the integrated weight of each contracted pair is at most

\begin{equation}
J^{2}\delta\sum_{\alpha}\|L_{\alpha,e}\|^{2},
\end{equation}
where $J^{2}$ counts the possible propagators containing the two
bath operators. Thus, by including the {combinatorial}
factors from the Leibniz rule {and extracting} the
common factor of $\delta^{p+n}$, we have
\begin{equation}
\begin{aligned}\relax|\langle O_{S}\rangle^{p,n}_{e,j}| & \leq\frac{\binom{p+2n}{p}(2n-1)!!}{(p+2n)!}(J\|h_{e}\|)^{p}\left(J^{2}\sum_{\alpha}\|L_{\alpha,e}\|^{2}\right)^{n}\\
 & =\frac{(J\|h_{e}\|)^{p}}{p!}\frac{\left(\frac{J^{2}}{2}\sum_{\alpha}\|L_{\alpha,e}\|^{2}\right)^{n}}{n!},
\end{aligned}
\end{equation}

For the collision replacement $\langle O_{S}\rangle^{p,n}_{e,j-1}$,
the {bound} is even more transparent because $\lVert\sigma_{\alpha,e}[j]\rVert=1$
and
\begin{equation}
\|L^{\dagger}_{\alpha,e}\sigma_{\alpha,e}[j]+L_{\alpha,e}\sigma^{\dagger}_{\alpha,e}[j]\|=\|L_{\alpha,e}\|.
\end{equation}
 An odd number of {bath operators acting} on any selected
qubit has zero vacuum expectation. The contribution from the Hamiltonian
factors can be bounded similarly, while the {vacuum
expectation requires} the index $\alpha$ in $L_{\alpha,e}$ {to
occur} in {pairs. Counting these pairings,} we obtain
\begin{equation}
\begin{aligned}\relax|\langle O_{S}\rangle^{p,n}_{e,j-1}| & \leq\frac{(J\|h_{e}\|)^{p}}{p!}\sum_{n_{\alpha}\geq0,\sum_{\alpha}n_{\alpha}=n}\prod_{\alpha}\frac{(J\lVert L_{\alpha,e}\rVert)^{2n_{\alpha}}}{(2n_{\alpha})!}\\
 & \leq\frac{(J\|h_{e}\|)^{p}}{p!}\frac{\left(\frac{J^{2}}{2}\sum_{\alpha}\|L_{\alpha,e}\|^{2}\right)^{n}}{n!}
\end{aligned}
\end{equation}
 Consequently, 
\begin{equation}
\begin{aligned}\sup_{\rho_{S},O_{S}}\left|\langle O_{S}\rangle_{e,j}-\langle O_{S}\rangle_{e,j-1}\right| & \leq2\sum_{p+n\geq2}\frac{(2M\delta)^{p}(2M^{2}\delta)^{n}}{p!\,n!}\\
 & =2\left(e^{(2M+2M^{2})\delta}-1-(2M+2M^{2})\delta\right)\\
 & \leq16\mathrm{e}M^{4}\delta^{2}
\end{aligned}
\end{equation}
 for $4M^{2}\delta{<}1$. Here we have used $e^{y}-1-y\leq y^{2}e/2$
if $y<1$. Substitution into the two telescoping sums {in}
\cref{eq:two_fold_telescoping_sum_together} gives
\begin{equation}
\begin{aligned}\relax\|\mathcal{E}^{\text{split}}-\hat{\mathcal{E}}\|_{\diamond} & \leq\sum_{e}\sum^{T}_{j=1}16\mathrm{e}M^{4}\delta^{2}\\
 & \leq16\mathrm{e}M^{4}NT\delta^{2}=16\mathrm{e}M^{4}Nt\delta.
\end{aligned}
\end{equation}
 This proves the statement with $\tilde{C}=16\mathrm{e}$.
\end{proof}

\section{Block encoding of non-Hermitian time-evolution with jump}\label{app:imag_time}

We begin with a technical lemma for this section: In some of the analysis below, given a Hamiltonian $H$, we would like construct a unitary which takes in one of its registers an input time $t$ and conditioned on this register applies the unitary $e^{-iHt}$. The lemma below shows that such a controlled operation can be implemented with a number of gates scaling linearly in the number of bits representing $t$.

\begin{lemma}\label{lemma:contr_H_evol}
    Suppose $H$ is a Hamiltonian on $n$ qubits with $(\alpha, 0)$-block-encoding $U_H$. Given a time step $\delta > 0$ and $T \in \mathbb{Z}_{> 0}$, a unitary $U_\textnormal{H-EVOL}$ on $n + \lfloor \log_2 T\rfloor + n_a$ qubits satisfying
    \begin{align}
    \norm{\bra{0}_aU_\textnormal{H-EVOL}\ket{\tau}\ket{\psi}\ket{0}_a - \ket{\tau}(e^{-iH \tau \delta}\ket{\psi})} \leq \epsilon \quad \text{for}\quad \tau \in [0:T),
    \end{align}
    can be implemented with $O(\alpha T\log(1/\epsilon))$ calls to $U_H$ and $O(\alpha T \log T \log(1/\epsilon))$ other single and two-qubit gates. 
\end{lemma}
\begin{proof}
For notational convenience, we define $m = \lfloor \log_2 T\rfloor$. Represent $\tau \in [0:T)$ with a $m$-bit string $b(\tau)$ such that
\begin{align}
\tau = \sum_{j = 0}^{m}b_j(\tau) 2^j.
\end{align}
The idea behind implementing $U_\textnormal{H-SUM}$ is to implement dynamics for time $\delta$ under the Hamiltonian $\mathcal{H} = \mathcal{T}\otimes H$ where $\mathcal{T} = \sum_{\tau = 0}^{T - 1} \tau \ket{\tau}\!\bra{\tau}$. We now construct a block-encoding of $\mathcal{H}$: Consider a unitary $C$ which acts on two copies of the time-register and an additional ancilla and implements
\begin{align}
    C\ket{0}_a \ket{s}\ket{\tau} = \ket{f(s, \tau)}\ket{s}\ket{\tau} \quad \text{if}\quad f(s, \tau) = \begin{cases}
        1 & \text{if }s < \tau,\\
        0 & \text{otherwise}.
    \end{cases}
\end{align}
Then, consider the unitary $U_\mathcal{T} = (I\otimes H^{\otimes m}\otimes I)C(I\otimes H^{\otimes m}\otimes I)$ where $H$ is the Hadamard gate. It follows by construction that
\begin{align}
    \bra{0}_a\bra{0} U_\mathcal{T} \ket{0}_a\ket{0}\ket{\tau} &= \bra{0}_a\bra{0}(I\otimes H^{\otimes m}\otimes I)C(I\otimes H^{\otimes m}\otimes I)\ket{0}_a\ket{0}\ket{\tau} \nonumber\\
    &=\frac{1}{2^{m/2}}\sum_{s, s' = 0}^{T - 1}\bra{0}_a\bra{s'}C\ket{0}_a\ket{s}\ket{\tau} \nonumber \\
    &=\frac{1}{2^{m/2}}\sum_{s, s' = 0}^{T - 1}\bra{0}_a\bra{s'}C\ket{f(s, \tau)}_a\ket{s}\ket{\tau} \nonumber \\
    &=\frac{1}{2^m}\sum_{s = 0}^{T - 1} \bra{0}f(s, \tau)\rangle \ket{\tau} \nonumber\\
    &= \frac{1}{2^m} \sum_{s = 0}^{\tau - 1}\ket{\tau} = \frac{1}{2^m}.
\end{align}
Thus, $U_{\mathcal{T}}$ is a $(2^m, 0)-$block-encoding of $\mathcal{T}$. Furthermore, note that the unitary $C$, and thus the unitary $U_{\mathcal{T}}$, can be implemented with $O(m)$ gates since it just has to do a bit-wise comparison between two times. 

Finally, a $(\alpha 2^m, 0)$-block encoding for $\mathcal{T}\otimes H$ is given by $U_{\mathcal{T}\otimes H} = U_{\mathcal{T}}\otimes U_H$ \citep[Lemma 53]{Gilyen2019}. Invoking a Hamiltonian simulation algorithm for $\mathcal{T}\otimes H$ (\citep[Theorem 58]{Gilyen2019}, \cite{LowChuang2017}), we obtain that $U_\text{H-SUM}$ can be implemented with $O(\alpha 2^m\delta \log(1/\epsilon))$ calls to $U_{\mathcal{T}\otimes H}$ and the same order of additional single and two-qubit gates. Writing $2^m \leq O(T)$ and accounting for the fact that every call to $U_\mathcal{T}$ requires $O(m) \leq O(\log T)$ gates, we obtain the lemma statement.
\end{proof}
We also remind the reader of the following lemma that provides a method to implement the sum of a set of operators given their block encodings.

\begin{lemma}[Linear combination of block-encoded matrices, Lemma 52 \cite{Gilyen2019}]\label{lemma:block_add}
Suppose $A = \sum_{j = 1}^m w_j A_j$ and $\exists w > 0$ such that $\sum_j \norm{w_j} \leq w$. Then, given a select unitary $U_\textnormal{SELECT} = \sum_{j = 1}^m \ket{j}\!\bra{j}\otimes U_j$ where $U_j$ is a $(\alpha, \epsilon)-$block-encoding of $A_j$. Then, a $(\alpha w, \alpha w\epsilon)$-block-encoding of $A$ can be implemented $O(m)$ single and two-qubit gates and a single call to $U_\textnormal{SELECT}$.
\end{lemma}

\begin{lemma}[Non-Hermitian evolution with interleaved jumps]\label{lemma:LCHS_interleaved}
Given Hermitian operators $H$ and $D$, with $D\succeq 0$ as well as operators $R_1, R_2 \dots R_q$, define the operator $F$
\begin{align}
    F = \sum_{s_0, s_1 \dots s_q =0}^t \ket{s_0, s_1 \dots s_q}\!\bra{s_0, s_1 \dots s_q} \otimes \exp(Qs_q)R_q\exp(Qs_{q-1}) R_{q-1} \dots \exp(Qs_1)R_1\exp(Qs_0),
\end{align}
where $Q = -(iH+D)$ and the sum runs over $s_i = k\delta$ for $k \in \mathbb{Z}_{\geq 0}$ and for some $\delta > 0$. Furthermore, suppose we are given a $(\alpha_H, 0)$-block-encoding of $H$, a $(\alpha_D, 0)$-block-encoding of $D$ and $(1, 0)$-block encodings of $R_i$. Then, a $(\beta, \epsilon)$-block-encoding of $F$, where $\beta \leq O(1)$, can be constructed with $O(\textnormal{poly}(q, t, \alpha_H, \alpha_D, \log(1/\epsilon)))$ calls to the block encodings of $H, D$ and $R_i$ and the same-order of continous single and two-qubit gates.
\end{lemma}
\begin{proof}
We adapt the LCHS kernel of Ref.~\cite{An2026Quantum} to the entire
interleaved product. Define the analytic function
\begin{equation}
\mathcal{F}_{s_0, s_1 \dots s_q}(z)
=
e^{-is_q(H +zD)}R_q e^{-is_{q - 1}(H+zD)}\cdots
R_1e^{-is_0(H+zD)}
\quad \text{and}\quad
g(z)
=
\frac{e^{\sqrt{2}}}{2\pi}
\frac{e^{-\sqrt{1+iz}}}{1-iz}.
\end{equation}
Note that $\mathcal{F}_{s_0, s_1 \dots s_q}(-i) = F_{s_0, s_1 \dots s_q}$. Futhermore, for $\text{Im}(z) \leq 0$,  $\norm{\exp(-i s_l(H_l + zD_l))}$ since $D_j \succeq 0$, tand therefore
$\norm{\mathcal{F}(z)}\leq1$. Furthermore, the kernel
$g(z)$ has a unique pole there at $z=-i$, with residue $i/(2\pi)$ and $g(z) \to 0$ as $\text{Im}(z) \to -\infty$. Therefore, from the Cauchy residue theorem, we obtain that
\begin{equation}
\int_{\mathbb{R}}g(k)\mathcal{F}_{s_0, s_1 \dots s_q}(k)\,dk
=
\mathcal{F}(-i) = F_{s_0, s_1 \dots s_q}.
\end{equation}
This furnishes a relation between $F_{s_0, s_1 \dots s_q}$ and real-time evolution under the Hamiltonian $H + D$.
Next, following Ref.~\cite{An2026Quantum}, we truncate the integral over $k$ to the interval $[-k_\text{max},k_\text{max}]$. Choosing $k_\text{max} = \Theta(\text{polylog}(1/\epsilon))$ ensures that
\begin{equation}
\left\Vert{\int_{-k_\text{max}}^{k_{\text{max}}}g(k)\mathcal{F}_{s_0, s_1 \dots s_q}(k)dk - \int_{-\infty}^\infty g(k)\mathcal{F}_{s_0, s_1 \dots s_q}(k)dk }\right \Vert \leq \frac{\epsilon}{4}.
\end{equation}
Next, following Ref.~\cite{An2026Quantum}, this integral is split into sum of integrals over small intervals within each of which a Gaussian quadrature discretization of the integral is employed. This procedure yields weights $\{w_\ell: \ell \in [1:N_q]\}$ and discretization nodes $\{k_\ell \in [-k_\text{max}, k_\text{max}]: \ell \in [1:N_q]\}$ which are independent of $s_1, s_2 \dots s_q$, such that
\begin{equation}
\left \Vert \int_{-k_\text{max}}^{k_\text{max}}g(k)\mathcal{F}_{s_0, s_1 \dots s_q}(k)dk
-\sum^{N_q}_{\ell=1}w_\ell g(k_\ell)\mathcal{F}_{s_0, s_1 \dots s_q}(k_\ell) \right \Vert \leq  \frac{\epsilon}{4},
\end{equation}
and therefore we obtain an implementation of $F_{s_0, s_1 \dots s_q}$ as a discrete sum:
\begin{align}\label{eq:final_interpolation_bound}
    \left \Vert F_{s_0, s_1 \dots s_q}- \sum_{\ell = 1}^{N_q} w_\ell g(k_\ell)\mathcal{F}_{s_0, s_1 \dots s_q}(k_\ell)\right \Vert \leq \frac{\epsilon}{2}.
\end{align}
To estimate the number of quadrature points $N_q$, we note that $\norm{\mathcal{F}^{(n)}_{s_0, s_1 \dots s_q}(k)} \leq q^n\norm{D}^n S^n \leq q^n \alpha_D^n S^n$. Then, following the exact same analysis as in the proof of Lemma 11 in Ref.~\cite{An2026Quantum}, we obtain that $N_q = \Theta((\alpha_H + \alpha_D)qS\ \text{polylog}(1/\epsilon))$. Finally, it is shown in Ref.~\cite{An2026Quantum} that the the coefficients $w_\ell g(c_\ell)$ also satisfy that
\begin{align}\label{eq:upper_bound_coeff}
    \sum_{\ell = 1}^N \abs{w_\ell g(k_\ell)} \leq \beta ,
\end{align}
where $\beta$ is a $O(1)$ constant independent of all other parameters. This will be useful in constructing the block encodings as shown below.

We begin with the implementation of a unitary $U_\text{EVOL-SUM}$ such that
\begin{align}\label{eq:evol_sum}
U_\text{EVOL-SUM}\ket{s}_1 \ket{\ell}_2 \ket{\psi}_3 = \ket{s}_1 \ket{\ell}_2 (e^{-is(H + k_\ell D)}\ket{\psi}_3).
\end{align}
Here the first-register carries $s \leq S$ specified upto $r$ bits of precision and thus has $r + O(\log S)$ qubits, the second-register has $O(\log N_q)$ qubits and the final register is any ancillas needed for implementing $U_\text{SUM}$. In Ref.~\cite{An2026Quantum} it is shown how to construct a unitary $U_{\text{H-SUM}}$ with $1$ call each to the block-encodings of $H$ and $D$ and $O(\poly(N_q))$ additional single and two-qubit gates such that
\begin{align}
    U_{\text{H-SUM}}\ket{\ell}_2 \ket{\psi}_3 \ket{0}_a =\frac{1}{\alpha_H + k_\text{max}\alpha_D} \ket{\ell}_2 (H + k_\ell D)\ket{\psi}_3 \ket{0}_a + \ket{\text{GARB}},
\end{align}
where $\ket{\text{GARB}}$ is some unnormalized state satisfying $(\ket{0}_a\bra{0}_a) \ket{\text{GARB}} = 0$. Using $\hat{U}_{\text{H-SUM}}$ together with lemma \ref{lemma:contr_H_evol}, we obtain an $\epsilon'$-approximation to $U_\text{EVOL-SUM}$  Eq.~\eqref{eq:evol_sum} with $O((\alpha_H + k_\text{max}\alpha_D)S\log(1/\epsilon'))$ calls to the block-encoding of $H$ and $D$ and $O((\alpha_H + k_\text{max}\alpha_D)S\log S \log(1/\epsilon'))$ single and two-qubit gates.

\begin{figure}
    \centering
    \includegraphics[width=0.8\linewidth]{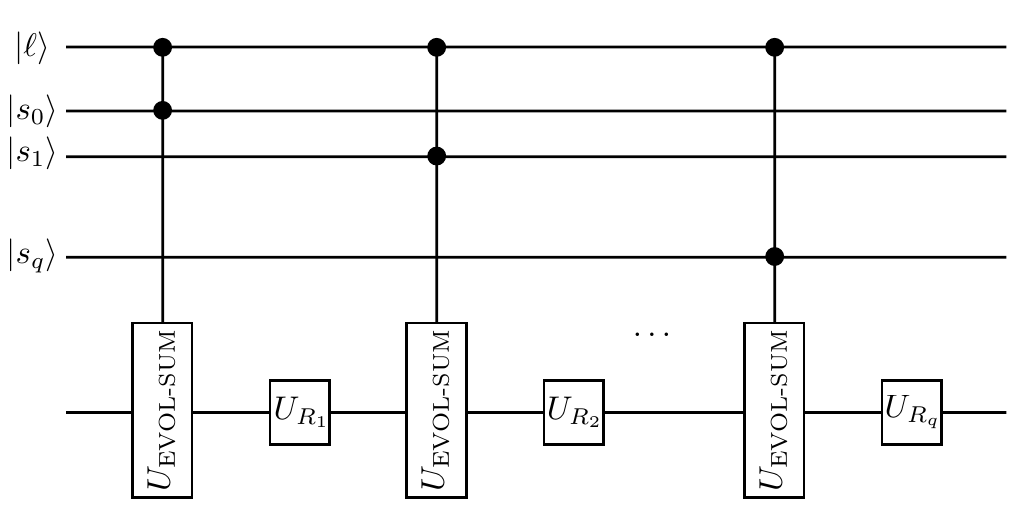}
    \caption{Construction of the selection unitary used in lemma \ref{lemma:LCHS_interleaved}. }
    \label{fig:evol_sum}
\end{figure}
\end{proof}

Now, using th $\hat{U}_\text{EVOL-SUM}$ and the given (1,0)-block-encodings of $R_i$, $U_{R_i}$, we implement the circuit shown in Fig.~\ref{fig:evol_sum}. This implements a selection unitary $U_\text{SELECT}$ which satisfies
\begin{align}
    U_\text{SELECT} = \sum_{\ell = 1}^{N_q}\ket{\ell}\!\bra{\ell} \otimes U_{\mathcal{F}(k_\ell)},
\end{align}
where $U_{\mathcal{F}(k_\ell)}$ is a $(1, q\epsilon)$-block-encoding of 
\begin{align}
    {\mathcal{F}(k_\ell)} = \sum_{s_0, s_1 \dots s_q}\ket{s_0, s_1 \dots s_q}\!\bra{s_0, s_1 \dots s_q} \otimes \mathcal{F}_{s_0, s_1 \dots s_q}(k_\ell).
\end{align}
Finally, employing Linear-Combination-of-Unitaries as stated in \cref{lemma:block_add} and using Eq.~\eqref{eq:upper_bound_coeff}, we can use $U_\text{SELECT}$ to implement $U_{\hat{F}}$, a $(\beta, \beta q\epsilon')$-block encoding of the operator
\begin{align}
   \hat{F} = \sum_{\ell = 1}^{N_q} w_\ell g(k_\ell) \sum_{s_0, s_1 \dots s_q} \ket{s_0, s_1 \dots s_q}\!\bra{s_0, s_1 \dots s_q} \otimes \mathcal{F}_{s_0, s_1 \dots s_q}(k_\ell),
\end{align}
with a single call of $U_\text{SELECT}$ and $O(N_q)$ continous single and two qubit gates.

Finally, we note that from Eq.~\eqref{eq:final_interpolation_bound} that $\norm{\hat{F} - F} \leq \epsilon/2$. Thus, $U_{\hat{F}}$ is a $(\beta, \beta q\epsilon' + \epsilon/2)$-block-encoding of $F$. Setting $\epsilon' = \epsilon/2\beta q$ we obtain the lemma statement.

\section{Construction of the input unitary $U_\text{in}$}\label{app:inp_unitary}
In this section, we will prove the following lemma which will be used to prove part (b) of \cref{lemma:inp_unitary_circuit}.
\begin{lemma}\label{lemma:input_unitary}
    Given $n = 2^m$ qubits and integer $n_c = 2^{m_c}-1$, define $C^{\leq n_c} = \{b \in \{0, 1\}^n: \sum_i b_i \leq n_c\}$ i.e., $C^{\leq n_c}$ as the set of bit strings with $\leq n_c$ ones. Given $\nu = 2^{-r}$ such that $n n_c \nu\leq 1$, define for $a, b \in C^{\leq n_c}$,
    \[
    f_{a, b} =  (n\nu)^{\abs{a \setminus b} - \abs{b \setminus a}} \nu^{\abs{b \setminus a}},
    \]
    where $\abs{a \setminus b} = |\{i \in [1:n]: a_i = 1 \text{ and } b_i = 0\}|$ and $\abs{b\setminus a}$ is similarly defined. Represent $b \in C^{\leq n_c}$ with its excitation-location representation $\mathsf{ex}(b)$ (\cref{def:ex_loc}) that can be stored in $n_c m + m_c$ bits. Then, a unitary $U$ on $2(n_c m + m_c) +n_a$ satisfying
    \[
    U\ket{\mathsf{ex}(a)}\ket{0}\ket{0}_a =\frac{1}{\sqrt{8}} \sum_{b \in C^{\leq n_c}} \sqrt{f_{a, b}} \ket{\mathsf{ex}(a)}\ket{\mathsf{ex}(b)}\ket{0}_a + \ket{\perp_a},
    \]
    for some $(\ket{0}_a\bra{0})\ket{\perp_a} = 0$, can be implemented with 
    \[
    O(\textnormal{poly}(m, n_c, r)) = O(\textnormal{poly}(\log n, n_c, \log(1/\nu)))
    \]
    single and two-qubit gates and $n_a = O(\textnormal{poly}(m, n_c, r))$ ancillas.
\end{lemma}
We begin with a few simple lemmas which will be useful in the proof of \cref{lemma:input_unitary}.
\begin{lemma}\label{lemma:brute_force_control}
    For $b \in \{0, 1\}^p$, suppose we are given $p_0(b), p_1(b) \dots p_{2^r - 1}(b) \geq 0$ such that $\sum_{j = 0}^{2^r - 1}p_j(b) \leq 1$. Then, there is a unitary $U$ acting on $p+r+n_a$ qubits such that
    \begin{align}
        U\ket{b}\ket{0}\ket{0}_a = \sum_{c = 0 }^{2^r - 1}\sqrt{p_c(b)}\ket{b}\ket{c}\ket{0}_a + \ket{\perp_b},
    \end{align}
    for some $(I\otimes I\otimes \ket{0}_a\!\bra{0})\ket{\perp_b} \in (\mathbb{C}^{2})^{\otimes r}$ such that $U$ can be implemented with $O(2^{p + 2r})$ (possibly continuous) single and two-qubit gates.
\end{lemma}
\begin{proof}
    This is a straightforward consequence of the fact that unitaries over $n$ qubits can be exactly implemented with $O(4^n)$ single and two-qubit continuous gates. More specifically, given $p_0(b), p_1(b) \dots p_{2^r - 1}(b)$, one can explicitly construct a unitary $U_b$ as a $2^{r+1}\times 2^{r + 1}$ matrix acting on $r + 1$ qubits which satisfy
    \[
    U_b \ket{0}_2 \ket{0}_3 = \sum_{c = 0}^{2^r - 1} \sqrt{p_c(b)} \ket{c}_2\ket{0}_3 + \ket{\phi_b} \ket{1}_3,
    \]
    for some vector $\ket{\phi_b}$. This unitary can be implemented with $O(2^{2r})$ single and two-qubit gates. Now, we simply implement its controlled version to obtain $U$ i.e.,
    \[
    U = \sum_{b \in  \{0, 1\}^p} \ket{b}_1\!\bra{b} \otimes U_b,
    \]
    which can be done with $O(2^p 2^{2r})$ continuous single and two-qubit gates.
\end{proof}

\begin{definition}[Lexicographic ranking]\label{def:lex_ranking}
    Consider given $M \in \mathbb{Z}_{>0}$ and $k \in [1:M]$. Assign integers $i_1, i_2 \dots i_k \in [1:M]$ with $i_1 < i_2 < \dots <i_k$ to an integer $r_{i_1, i_2 \dots i_k} \in [1:{M \choose k}]$ via the following recursive procedure:
    \begin{enumerate}
        \item[(1)] For $k = 1$, assign $r_{i} = i$.
        \item[(2)] For $k > 1$, assign 
        \begin{align}\label{eq:lex_ranking_recursion}
            r_{i_1, i_2 \dots i_k} = r_{i_1, i_2 \dots i_{k - 1}} + \sum_{j = k}^{i_k - 1} {j - 1 \choose k - 1}.
        \end{align}
    \end{enumerate}
    As an example, for $k = 2$, this would assign the ranking
    \begin{align*}
    &(1, 2) \to 1, \\
    &(1, 3) \to 2, \\
    &(2, 3) \to 3, \\
    &(1, 4) \to 4, \\
    &(2, 4) \to 5, \\
    &(3, 4) \to 6, \\
    &\quad  \vdots
    \end{align*}
\end{definition}
\begin{lemma}[Lexicographic ranking]\label{lemma:lex_ranking}
    Consider given $M \in \mathbb{Z}_{>0}$ and $k \in [1:M]$. Assign integers $i_1, i_2 \dots i_k \in [1:M]$ with $i_1 < i_2 < \dots <i_k$ to their rank $r_{i_1, i_2 \dots i_k} \in [1:{M \choose k}]$ as per \cref{def:lex_ranking}.
    Then, the map $(i_1, i_2 \dots i_k) \to r_{i_1, i_2 \dots i_k}$ is a bijection and both the map and its inverse can be computed in space $O(\textnormal{poly}(k, \log M))$ and time $O(\textnormal{poly}(k, \log M))$.
\end{lemma}
\begin{proof}
We begin by making the simple observation that $i \in [1:M]$ can be stored in $O(\log_2 M)$ bits and ${M \choose k}$ can be stored in $O(k\log_2 M)$. 

Consider first the task of computing $r_{i_1, i_2 \dots i_k}$ given $i_1, i_2 \dots i_k$. For that, it is convenient to note the combinatorial identity
\begin{align}
    \sum_{j = k}^{i_k - 1}{j - 1 \choose k - 1} = {i_k - 1 \choose k}.
\end{align}
From this identity and Eq.~\eqref{eq:lex_ranking_recursion}, it follows that
\begin{align}\label{eq:lex_ranking_recursion_new} 
    r_{i_1, i_2 \dots i_k} = r_{i_1, i_2 \dots i_{k - 1}} + {i_k - 1 \choose k}.
\end{align}
Since, for $m, n \in \mathbb{Z}_{> 0}$ with $n \leq m$, ${m \choose n}$ can be computed in time and space $O(\poly(\log m, n))$, using the recursive definition of the lexicographic ranking, we can compute $r_{i_1, i_2 \dots i_k}$ in time $\leq O(\poly(k, \log M))$.

Next, consider the task of computing $i_1, i_2 \dots i_k$ given $r_{i_1, i_2 \dots i_k}$. This computation proceeds by computing $i_k$ first, then using Eq.~\eqref{eq:lex_ranking_recursion_new} to compute $r_{i_1, i_2 \dots i_{k - 1}}$ from which one can compute $i_{k - 1}$ and so on. To compute $i_k$, note that it is the maximum integer that satisfies
\begin{align}
    {i_k - 1\choose k} < r_{i_1, i_2 \dots i_k}.
\end{align}
Since ${i_k - 1 \choose k}$ is an increasing function of $i_k$, a simple binary search yields $i_k$ in space and time $O(\poly(k, \log M))$.
\end{proof}

\begin{lemma}[Preparing uniform superpositions, \cite{babbush2018encoding}]\label{lemma:prep_uniform_super}
    Given $m \in \mathbb{Z}_{>0}$, then a unitary $U$ acting on $m + m$ qubits and satisfying
    \begin{align}
        U\ket{j}\ket{0} =\ket{j}\bigg( \frac{1}{\sqrt{j}}\sum_{i = 0}^{j - 1} \ket{i}\bigg) \quad \text{where}\quad j \in [1:2^m],
    \end{align}
    can be implemented with $O(\textnormal{poly}(m))$ gates and the same order of additional ancillas.
\end{lemma}

\begin{proof}[Proof of \cref{lemma:input_unitary}]
    We remind the reader that we assume in the lemma that
    \begin{align}
        n = 2^{m}, n_c = 2^{m_c} - 1\quad \text{and}\quad \nu = 2^{-r}.
    \end{align}
    We also introduce $n_\text{rep}$ given by
    \begin{align}
        n_\text{rep} = n_c \log_2 n + \log_2 (n_c + 1) = n_c m + m_c.
    \end{align}
    Recall that $n_\text{rep}$ is the number of bits needed to store the representation $e(a)$ of a bit-string $a \in C^{\leq n_c}$ defined in the lemma.
    
    We first verify that $\sum_{b \in C^{\leq n_c}} f_{a, b} \leq 8$ since that is required for the unitary to be well-defined. We begin by rewriting
    \begin{align}\label{eq:rewrite_summation}
        \sum_{b \in C^{\leq n_c}} f_{a, b} = \sum_{q_a = 0}^{\abs{a}} \sum_{q_b= 0}^{n_c-q_a}p_{\abs{a}}(q_a, q_b)
    \end{align}
    where $\abs{a} = |\{i \in [1:n]: a_i = 1\}|$ and
    \begin{align}
        p_{\abs{a}}(q_a, q_b) = \begin{cases}{\abs{a} \choose q_a} {n - \abs{a} \choose q_b} (n \nu)^{\abs{a} - q_a - q_b} \nu^{q_b} & \text{if}\quad  q_a \leq \abs{a}, q_b \leq n_c - q_a,\\
        0 & \text{otherwise}.
        \end{cases}
    \end{align}
    To obtain Eq.~\eqref{eq:rewrite_summation} we have replaced the sum over $b \in C^{\leq n_c}$ with a sum over $q_a$, the number of 1s in $a$ that are not in $b$ and $q_b$, the number of 1s in $b$ that are not in $a$. Next, 
    \begin{align}
        \sum_{q_a = 0}^{\abs{a}} \sum_{q_b = 0}^{n_c - q_a} p_{\abs{a}}(q_a, q_b)  &\leq \sum_{q_a = 0}^{\abs{a}}{\abs{a}\choose q_a} (n\nu)^{\abs{a} - q_a} \sum_{q_b = 0}^{n_c} {n \choose q_b} n^{-q_b} \nonumber \\
        &\leq  (1 + n\nu)^{\abs{a}} (1 + n^{-1})^{n} \nonumber \\
        &\leq e^{1 + n \nu \abs{a}}.
    \end{align}
    Finally, note that by assumption $n \nu \abs{a} \leq n n_c \nu \leq 1$ and therefore
    \begin{align}
        \sum_{b \in C^{\leq n_c}} f_{a, b} =  \sum_{q_a = 0}^{\abs{a}} \sum_{q_b= 0}^{n_c-q_a}p_{\abs{a}}(q_a, q_b)  \leq e^{2} \leq 8.
    \end{align}

    For notational convenience, given $a$, we will denote by $\ell(a)$ the list of indices in the bit string $a$ which contain a $1$ i.e. $\ell(a) = \{j_1, j_2 \dots j_{\abs{a}}\}$ such that $a_{j_i} = 1$. It will be convenient to think of two bit strings $a, b \in C^{\leq n_c}$ as specified by the 3 sets $\ell(a) \setminus \ell(b)$ (locations of 1s in $a$ but not in $b$), $\ell(b) \setminus \ell(a)$ (locations of 1s in $b$ but not in $a$) and $\ell(a) \cap \ell(b)$ (locations of 1s both in $a$ and $b$). Now, given $a \in C^{\leq n_c}$, we begin by defining the following two unitaries:
    \begin{enumerate}
        \item[(1)] Design a unitary $U_\text{COEFF}$ on $\mathbb{C}^{n_\text{rep}}\otimes \mathbb{C}^{m_c} \otimes \mathbb{C}^{m_c}\otimes \mathbb{C}^2$ such that
        \begin{align}\label{eq:U_COEFF}
        U_\text{COEFF}\ket{e(a)}\ket{0}\ket{0}\ket{0} &= \frac{1}{\sqrt{8}}\sum_{q_a, q_b = 0}^{n_c} \sqrt{p_{\abs{a}}(q_a, q_b)}\ket{e(a)}\ket{q_a}\ket{q_b}\ket{0} +\nonumber\\
        &\qquad \qquad \qquad \qquad \bigg(1 - \frac{1}{8}\sum_{q_a, q_b = 0}^{n_c}p_{\abs{a}}(q_a, q_b)\bigg)^{1/2} \ket{a}\ket{\phi'_a}\ket{1},
        \end{align}
        for some normalized state $\ket{\phi_a'} \in \mathbb{C}^{m_c} \otimes \mathbb{C}^{m_c}$.
        \item[(2)] Design a unitary $U_\text{SP}$ acting on $\mathbb{C}^{n_\text{rep}}\otimes \mathbb{C}^{n_\text{rep}}\otimes \mathbb{C}^{m_c}\otimes \mathbb{C}^{m_c}$ such that
        \begin{align}\label{eq:U_SP}
            U_\text{SP}\ket{\mathsf{ex}(a)}\ket{0}\ket{q_a}\ket{q_b} = \frac{1}{\big[{\abs{a} \choose q_a}{n_c - \abs{a} \choose q_b} \big]^{1/2}} \bigg(\sum_{\substack{b\in C^{\leq n_c}\\ q_a = \abs{a \setminus b}, q_b = \abs{b\setminus a}}} \ket{\mathsf{ex}(a)}\ket{\mathsf{ex}(b)}\bigg)\ket{0}\ket{0}.
        \end{align}
    \end{enumerate}
    It is clear that $U$ from the lemma can be implemented by
    \begin{align}
        U = U_\text{SP}\cdot U_\text{COEFF}.
    \end{align}
    We now construct $U_\text{COEFF}$ and $U_\text{SP}$ and upper-bound their gate complexity.

    \emph{Construction of }$U_\text{COEFF}$. $U_\text{COEFF}$ can be implemented using \cref{lemma:brute_force_control}: Importantly, note that $p_{\abs{a}}(q_a, q_b)$ depends only on $\abs{a}$ i.e., the number of 1s in $a$. This is stored in the last $m_c$ bits of $e(a)$. Therefore, while constructing $U_\text{COEFF}$, we only need to control on the last $m_c$ bits of the first register in Eq.~\eqref{eq:U_COEFF}. From \cref{lemma:brute_force_control}, it thus follows that $U_\text{COEFF}$ can be implemented with $O(n_c^4 2^{m_c}) \leq O(\poly(n_c))$ single and two-qubit continuous gates.

    \emph{Construction of }$U_\text{SP}$. Constructing $U_\text{SP}$ is subtler: Observe from Eq.~\eqref{eq:U_SP} that $U_\text{SP}$ starts with the bit string $a$ at its input, and generates a uniform superposition of all possible bit strings $b \in C^{\leq n_c}$ with $q_a$ excitations inside of $a$ and $q_b$ excitations outside of $a$. 
    \begin{enumerate}
    \item[(1)] We begin by, with $\mathsf{ex}(a), q_a, q_b$ reversibly computing ${\abs{a}\choose q_a}$ and ${n - \abs{a}\choose q_b}$ and storing them in additional registers with $O(\poly(n_c, \log n))$ bits:
    \begin{align}
        \ket{\mathsf{ex}(a)}\ket{q_a}\ket{q_b} \to \ket{\mathsf{ex}(a)}\ket{q_a}\ket{q_b}\ket{{\abs{a}\choose q_a}}\ket{{n_c - \abs{a}\choose q_b}}.
    \end{align}
    Note that this can be done with a reversible circuit with $O(\poly(n_c, \log n))$ gates and the same order of additional ancillas.

    \item[(2)] Next, we use \cref{lemma:prep_uniform_super}, conditioned on the computed values of ${\abs{a}\choose q_a}$ and ${n - \abs{a}\choose q_b}$, to prepare a uniform superpositions in additional registers. After that, we can reversibly return the registers storing ${\abs{a}\choose q_a}$ and ${n_c - \abs{a}\choose q_b}$ to $\ket{0}$ and discard them.
    \begin{align}
        \ket{\mathsf{ex}(a)}\ket{q_a}\ket{q_b}\ket{{\abs{a}\choose q_a}}\ket{{n_c - \abs{a}\choose q_b}} \to \ket{\mathsf{ex}(a)}\ket{q_a}\ket{q_b}\bigg(\sum_{I = 0}^{{\abs{a}\choose q_a} - 1}\ket{I}\bigg)\bigg(\sum_{J = 0}^{{n - \abs{b}\choose q_b} - 1} \ket{J}\bigg).
    \end{align}
    From \cref{lemma:prep_uniform_super}, this can be accomplished with $O(\poly(n_c, \log n))$ gates.
    \item[(3)] Finally, we treat the indices in the uniform superpositions as lexicographic ranks over a combinatorial number of objects and map it to a list of indices. Specifically, using lemma \ref{lemma:lex_ranking}, we can map
    \begin{align}
        &I \in \big[0:{\abs{a}\choose q_a}\big): \ket{I} \to \ket{i_1, i_2 \dots i_{q_a}} \quad \text{with}\quad 1\leq i_1 < i_2 < \dots <i_{q_a}\leq \abs{a}\quad \text{and},\\
        &J \in \big[0:{n_c - \abs{a}\choose q_b}\big):\ket{J} \to  \ket{j_1, j_2 \dots j_{q_b}} \quad \text{with}\quad 1\leq j_1 < j_2 < \dots <j_{q_b} \leq n_c - \abs{a}.
    \end{align}
    Then, we use the list $\ket{\mathsf{ex}(a)}$ and the indices $\ket{i_1, i_2 \dots i_{q_a}}$ to map it to a list of excitations in $a$ and map the indices $\ket{j_1, j_2 \dots j_{q_b}}$ to a list of excitations outside of $a$: These lists together yield the excitation locations in $b$.  Finally, the total number of excitations $q_a+q_b$ can be appended to the produced excitations to to obtain $\mathsf{ex}(b)$. Finally, the registers containing $q_a, q_b$ can be set to $0$ using the excitation number stored in ${\mathsf{ex}(a)}$ and ${\mathsf{ex}(b)}$. We have thus implemented the transformation
    \begin{align}
        \ket{\mathsf{ex}(a)}\ket{q_a}\ket{q_b}\bigg(\sum_{I = 0}^{{\abs{a}\choose q_a} - 1}\ket{I}\bigg)\bigg(\sum_{J = 0}^{{n - \abs{b}\choose q_b} - 1} \ket{J}\bigg) \to \ket{\mathsf{ex}(a)} \sum_{\substack{b \in C^{\leq n_c}\\ q_a = \abs{a\setminus b}, q_b = \abs{b\setminus a} }}\ket{\mathsf{ex}(b)}.
    \end{align}
    All of these steps can be performed reversibly in time scaling at most as $O(\poly(n_c, \log n))$. 
    \end{enumerate}
\end{proof}
\begin{proof}[Proof of \cref{lemma:inp_unitary_circuit}(b)]
The only subtlety that needs to be addressed is the fact that the bit strings $\bm{a}, \bm{b} \in B^{\leq n_c}$ in $U_\text{in}$ defined in Eq.~\eqref{eq:input_unitary_target} are a collection of bit sub-strings which have $\leq n_c$ excitations. However, we note that $G(\bm a, \bm b)$ defined in Eq.~\eqref{eq:input_unitary_target} satisfies
\begin{align}
    G(\bm a, \bm b) = \prod_{i = 1}^p G(\bm a_i, \bm b_i) \quad \text{where}\quad G(\bm a_i, \bm b_i) = (Kt_0)^{\abs{\mathsf{ex}(\bm a_i)\setminus \mathsf{ex}(\bm b_i) } -\abs{\mathsf{ex}(\bm b_i)\setminus \mathsf{ex}(\bm a_i) }  } \delta^{\abs{\mathsf{ex}(\bm b_i)\setminus \mathsf{ex}(\bm a_i) }},
\end{align}
since for $\bm a, \bm b \in B^{\leq n_c}$, $\bm a_i \cap \bm a_j = \bm a_i \cap \bm b_j = \emptyset$ for $i \neq j$. Thus to implement $U_\text{in}$, we use \cref{lemma:inp_unitary_circuit} to implement $U_{\text{in}, i}$ satisfying
\begin{align}
    U_{\text{in}, i}\ket{\mathsf{ex}(a_i)}\ket{0}\ket{0}_a = \frac{1}{\sqrt{8}} \sum_{\bm b_i \in B_i^{\leq n_c}} \sqrt{G(\bm a_i, \bm b_i)} \ket{\mathsf{ex}(\bm a_i)}\ket{\mathsf{ex}(\bm b_i)}\ket{0}_a + \ket{\perp_{a_i}}.
\end{align}
Then $U_\text{in} = \prod_{i = 1}^p U_{\text{in}, i}$ (applied with different ancilla registers), which proves the lemma statement.
    
\end{proof}
\section{Approximate-unitary to a unitary}\label{app:appx_U_to_U}
\begin{lemma}\label{lemma:appx_U_to_U}
    Suppose $A$ is an operator and $\Pi_\textnormal{in}, \Pi_\textnormal{out}$ are two orthogonal projectors on a Hilbert space $\mathcal{H}$ such that $\norm{\Pi_\textnormal{in}A^\dagger \Pi_\textnormal{out} A \Pi_\textnormal{in} - \Pi_\textnormal{in}} \leq \epsilon_0 < 1$. Furthermore, we are given
    \begin{enumerate}
        \item[(1)]$U_A$ be an $(\alpha, \epsilon)-$block-encoding of $A$.
        \item[(2)]Reflections $R_\textnormal{in}, R_\textnormal{out}$ around the subspaces $\Pi_\textnormal{in}, \Pi_\textnormal{out}$ respectively i.e.,
        \begin{align}
            R_\textnormal{in} = I-2\Pi_\textnormal{in}, R_\textnormal{out} = I-2\Pi_\textnormal{out}.
        \end{align}
    \end{enumerate}
     Then, there is a unitary $U$ on $\mathcal{H}\otimes \mathcal{H}_A$, where $\mathcal{H}_A$ is an ancillary space, such that $\norm{\Pi_\textnormal{out}(\bra{0}_{\mathcal{H}_A} U\ket{0}_{\mathcal{H}_A} -A)\Pi_\textnormal{in}} \leq \epsilon_0 + O(\epsilon)$ and it can be implemented with $O(\alpha \log(1/\epsilon))$ calls to $U_A, R_\textnormal{in},R_\textnormal{out}$ together with the same order of additional single and two-qubit gates.
\end{lemma}
\begin{proof}
    For notational convenience, define $\tilde{A} = \Pi_\text{out}A \Pi_\text{in}$, $\mathcal{H}_\text{in} = \text{Im}(\Pi_\text{in}), \mathcal{H}_\text{out} = \text{Im}(\Pi_\text{out})$. Note that from the assumption on $A$ in the lemma statement, $\norm{\tilde{A}^\dagger \tilde{A} - I_{\mathcal{H}_\text{in}}} \leq \epsilon_0$. Thus, if $\epsilon_0 < 1$, $\tilde{A}$ is invertible on the space $\mathcal{H}_\text{in}$. Consequently, if we write the singular-value decomposition $\tilde{A} = U_A \Sigma_A V_A^\dagger$, where $U_A, V_A, \Sigma_A$ are all operators on $\mathcal{H}$, then $0$ singular values will correspond to right-singular vectors in $\mathcal{H}_\text{in}^\perp$ and the singular values will satisfy $\sigma(\tilde{A}) \in [(1 - \epsilon_0)^{1/2}, (1 + \epsilon_0)^{1/2}] \cup \{0\}$. Furthermore, the polar factor $W_A = U_A V_A^\dagger$ will satisfy
    \begin{align}
        \norm{\Pi_\text{out}({A} - U_A V_A^\dagger)\Pi_\text{in}} &\leq \norm{\Pi_\text{out}(U_A\Sigma_A V_A^\dagger  - U_A V_A^\dagger)\Pi_\text{in}}\nonumber\\
        &\leq \max((1 + \epsilon_0)^{1/2} - 1, 1 - (1 - \epsilon_0)^{1/2}) \leq \epsilon_0.
    \end{align}
    Thus, the polar factor $U_A V_A^\dagger$ furnishes a unitary that approximates $A$ restricted to the input and output subspaces. In. Ref.~\citep[Theorem 26]{Gilyen2019}, it was shown that if a $(\alpha, 0)-$block encoding of $A$ is available, then by $O(\alpha\log(1/\Delta))$ calls to this block encoding and to the reflection unitaries $R_\text{in}, R_\text{out}$ together with the same order of additional single and two-qubit gates, we can implement a unitary within $\Delta$ operator-norm error in $U_A V_A^\dagger$. 

    However, since we do not have an exact block encoding of $A$, we have to account for the block-encoding approximation error into our analysis. More specifically, define $B = \alpha\bra{0}_a U_A \ket{0}_a$, then since $U_A$ is an $(\alpha, \epsilon)$-block-encoding of $A$, then
    \begin{align}\label{eq:B_minus_A}
        \norm{B - A} \leq \epsilon.
    \end{align}
    Furthermore, $U_A$ is a $(\alpha, 0)$-block-encoding of $B$. Due to Eq.~\eqref{eq:B_minus_A} and the Weyl's inequality, we deduce that the singular values of $\tilde{B} = \Pi_\text{out} B \Pi_\text{in}$ are within $\epsilon$ radius of the singular values of $\tilde{A}$. In particular, it then follows that
    \begin{align}
        \sigma(\tilde{B}) = [(1 - \epsilon_0)^{1/2} - \epsilon, (1 + \epsilon_0)^{1/2} + \epsilon] \cup \{0\},
    \end{align}
    with the $0$ singular values again corresponding to the singular vectors being in $\mathcal{H}_\text{in}$. Therefore, Ref.~\citep[Theorem 26]{Gilyen2019} can be used to implement the polar-factor $U_B V_B^\dagger$ of $B$ to error $\Delta$ using $O(\alpha \log(1/\Delta))$ calls to the block encodings of $A$ and the reflection unitaries $R_\text{in}, R_\text{out}$. Finally, we note that two polar-factors would satisfy satisfy \cite{bhatia2013matrix}
    \begin{align}\label{eq:perturbation_polar_factor}
        \norm{(U_A V_A^\dagger - U_B V_B^\dagger)\Pi_\text{in}} \leq \frac{2\norm{\tilde{A} - \tilde{B}}}{\sigma_\text{min}(\tilde{A}) + \sigma_\text{min}(\tilde{B})} \leq \frac{2\epsilon}{2 - 2\epsilon_0 - \epsilon} \leq O(\epsilon).
    \end{align}
    Setting $\Delta = O(\epsilon)$, we then obtain the lemma statement.
\end{proof}
\end{document}